\documentclass{article}

\usepackage{amsmath,amsfonts,bm}

\def\eqref#1{equation~\ref{#1}}

\def\1{\bm{1}}

\DeclareMathAlphabet{\mathsfit}{\encodingdefault}{\sfdefault}{m}{sl}
\SetMathAlphabet{\mathsfit}{bold}{\encodingdefault}{\sfdefault}{bx}{n}

\DeclareMathOperator*{\argmax}{arg\,max}
\DeclareMathOperator*{\argmin}{arg\,min}

\usepackage{hyperref}
\usepackage{url}
\usepackage{algorithm}
\usepackage{algpseudocode}

\usepackage{amsthm}
\usepackage{appendix}
\usepackage{cleveref}
\usepackage{graphicx}
\usepackage{natbib}
\usepackage{multirow}
\usepackage{booktabs}
\newtheorem{theorem}{Theorem}[section]
\newtheorem{lemma}[theorem]{Lemma}

\newtheorem{proposition}[theorem]{Proposition}

\crefname{appsec}{Appendix}{Appendices}

\newtheorem{innercustomgeneric}{\customgenericname}
\providecommand{\customgenericname}{}
\newcommand{\newcustomtheorem}[2]{%
  \newenvironment{#1}[1]
  {%
   \renewcommand\customgenericname{#2}%
   \renewcommand\theinnercustomgeneric{##1}%
   \innercustomgeneric
  }
  {\endinnercustomgeneric}
}

\newcustomtheorem{customtheorem}{Theorem}
\newcustomtheorem{customlemma}{Lemma}
\newcustomtheorem{customclaim}{Claim}
\newcustomtheorem{customproposition}{Proposition}
\newcustomtheorem{customcorollary}{Corollary}
\newcustomtheorem{customfact}{Fact}
\newcustomtheorem{customexample}{Example}
\newcustomtheorem{customnotation}{Notation}
\newcustomtheorem{customobservation}{Observation}
\newcustomtheorem{customconjecture}{Conjecture}

\usepackage{PRIMEarxiv}
\title{A Greedy PDE Router for Blending Neural Operators and Classical Methods}
\author{Sahana Rayan \\
Department of Statistics\\
University of Michigan\\
Ann Arbor, MI 48104, USA \\
\texttt{srayan@umich.edu} \\
\And
Yash Patel \\
Department of Statistics\\
University of Michigan\\
Ann Arbor, MI 48104, USA \\
\texttt{yppatel@umich.edu} \\
\And
Ambuj Tewari \\
Department of Statistics\\
University of Michigan\\
Ann Arbor, MI 48104, USA \\
\texttt{tewaria@umich.edu} \\
}

\begin{document}

\maketitle

\begin{abstract}
When solving PDEs, classical numerical solvers are often computationally expensive, while machine learning methods can suffer from spectral bias, failing to capture high-frequency components. Designing an optimal hybrid iterative solver--where, at each iteration, a solver is selected from an ensemble of solvers to leverage their complementary strengths--poses a challenging combinatorial problem. While greedy selection is desirable for its constant-factor approximation guarantee to the optimal solution under contraction and sequence-structure assumptions, it requires knowledge of the true error at each step, which is unavailable in practice. We address this by proposing an approximate greedy router that efficiently mimics a greedy approach to solver selection. Empirical results on the Poisson, convection–diffusion, and anisotropic diffusion equations show that our method reduces time-to-solution relative to single-solver baselines and existing hybrid approaches such as HINTS.
\end{abstract}

Natural phenomena and engineered systems are often governed by ordinary and partial differential equations (PDEs). Solving these equations enables a variety of tasks - predicting the evolution of the system through simulation (e.g., forecasting weather using the Navier-Stokes equations) \citep{kalnay2003atmospheric,bauer2015quiet,staniforth1991semi}, addressing control problems (e.g., optimizing heat shield design for spacecrafts using heat transfer equations) \citep{anderson1989hypersonic,troltzsch2010optimal}, and tackling inverse problems based on external measurements (e.g., reconstructing brain activity from EEG data using electrophysiological models) \citep{baillet2001electromagnetic,grech2008review}. 

Traditionally, PDEs are solved using finite difference methods \citep{smith1985numerical,leveque2007finite}, which discretize the domain and solve a system of equations over the grid points. Finite element methods \citep{hughes2003finite,bathe2006finite,brenner2008mathematical} instead construct a system of equations based on local basis functions, but likewise rely on numerical linear algebra techniques--such as the Jacobi and Gauss-Seidel \citep{saad2003iterative}--to compute the solution. These iterative methods, however, lack generalization across different initial conditions, boundary conditions, or forcing functions, as even minor changes to any of these parameters require re-solving the entire system of equations from scratch. 

These challenges have motivated neural operators \citep{kovachki2023neural}--a class of machine learning models which learn the solution operator directly, enabling fast inference across a range of parameters. They lift the conventional linear layer to an operator--typically a kernel integral operator--between function spaces. Despite strong approximation guarantees \citep{chen1995universal}, neural operators exhibit spectral bias \citep{liu2021multiscale,liu2024mitigating,you2024mscalefno,khodakarami2025mitigating,xu2025understanding}, similar to traditional neural networks \citep{rahaman2019spectral,xu2019frequency,luo2019theory}. They therefore favor learning low-frequency components of the solution and often struggle to capture high-frequency features that iterative solvers excel at. 

Recognizing this complementarity, \citet{zhang2024blending} proposed HINTS, a hybrid solver that 
interleaves a neural operator pass every $\tau^{\text{th}}$ iteration of a classical solver, with $\tau$ fixed a priori.
While HINTS reports significant empirical gains in convergence and accuracy over the standalone neural operator and Jacobi solver, this fixed schedule can be detrimental: a poorly timed neural correction may increase the error and undo recent progress. 

Although state-dependent solver schedules are desirable, this naturally suggests a reinforcement learning formulation for learning routing policies. However, in our setting, the structure of iterative PDE solvers allows for a simpler approach: each solver induces a known update, and its effect can be evaluated through error-based signals.
Leveraging this, we adopt a greedy rule that selects, at each iteration, the solver with the largest immediate error reduction.
Such a rule can be near-optimal when the final error satisfies (weak) supermodularity--so that applying a beneficial solver earlier is at least as valuable as applying it later. To support this approach, we make the following contributions:
\begin{itemize}
    \item We present a general hybrid PDE solver in which a routing rule chooses, at each iteration, among classical solvers and neural operators. With oracle access to the true error, we show that greedy routing achieves a constant-factor approximation to the optimal strategy for linear PDEs under zero-preserving contraction and a sequence-monotonicity condition, with the former establishing weak supermodularity.
    \item  Given that the true error is not observed at test time, a convex surrogate loss is introduced, which, when minimized, enables the learned model to imitate the greedy router without access to true error information. This alignment is guaranteed through Bayes consistency.
    \item We empirically demonstrate on Poisson, convection--diffusion, and anisotropic-diffusion equations that the approximate greedy router reduces time-to-solution relative to standalone numerical solvers and HINTS. 
\end{itemize}

\section{Background} \label{sec:background}

Consider the following linear PDE:
\begin{equation} \label{eq:mainpde}
    \begin{split}
    \mathcal{L}_{x}^a\left(u\right) = f ,x \in D; \quad
    \mathcal{B}_{x}\left(u\right) = b ,x \in \partial D
    \end{split}
\end{equation}
where $\mathcal{L}_x^a$ is a differential operator with respect to the spatial variable $x \in D$ parameterized by the coefficient function $a$,  $f$ is a forcing function,  $u$ is the PDE solution, $\mathcal{B}_{x}$ is the boundary operator, and $b$ is the boundary term. Such PDEs can be solved numerically using various discretization strategies, such as finite differences \citep{smith1985numerical,leveque2007finite}, finite element, and spectral methods \citep{boyd2001chebyshev,canuto2006spectral}. Here, we focus on finite differences, which replace derivatives with difference quotients (e.x., $\partial_x u(x) \approx (u(x + h) - u(x))/h$) on a uniform grid.

Formally, let $h$ be the grid size and $G_h(D)$ be a uniform grid with spacing $h$ over $D$. For a function $g: D \to \mathbb{R}$, we denote its restriction to $G_h(D)$ by $g_h \in \mathbb{R}^N$,
where $N = |G_h(D)|$. If $\mathcal{L}_h^a \in \mathbb{R}^{N \times N}$ is the discretized operator, the PDE reduces to a linear system of equations
\begin{equation} \label{eq:mainpdediscrete}
    \mathcal{L}_h^au_h = f_h,
\end{equation}
with the boundary conditions incorporated in $\mathcal{L}_h^a$. Unlike much of the neural operator literature, which retains function space formulations for discretization invariance, we follow the conventions of classical numerical analysis and represent functions and operators as vectors and matrices. Discretization invariance is not essential here, since we focus on fixed-grid problems.

\subsection{Iterative PDE Solvers} \label{sec:finitedifference}
Direct methods like Gauss Elimination and Thomas Algorithm can be computationally expensive in high-dimensional domains when solving systems like \Cref{eq:mainpdediscrete}. In contrast, iterative methods like Jacobi and Gauss-Seidel offer computational speedups by iteratively updating the solution, gradually converging to the true solution for the PDE. The general iterative update is
\begin{equation}\label{eq:iterativeupdate}
    u^{(t + 1)} = u^{(t)} + C\left(f_h - \mathcal{L}_h^au^{(t)}\right),
\end{equation}
where $u^{(t)}$ is the $t^{\text{th}}$ iterate of the solution and $C$ is a preconditioning matrix. Jacobi uses $C = D^{-1}$, where $D$ is the diagonal of $\mathcal{L}_h^a$, and Gauss-Seidel uses $C = (D + L)^{-1}$, where $L$ denotes the strict lower triangular part. 
However, these methods tend to damp low frequency parts of the error slowly. 




\subsection{Neural Operators} \label{sec:neuraloperators}

Suppose we observe samples $\{(a_i, f_i, u_i)\}_{i = 1}^N$ such that $(a_i, f_i) \sim \mathcal{P}$ are i.i.d. samples, and $u_i = \mathcal{G}^{*}(a_i, f_i)$ be generated by a deterministic solution operator $\mathcal{G}^{*}$. 
A neural operator (NO) $\mathcal{G}_{\theta}$ seeks to approximate $\mathcal{G}^*$ by minimizing the expected squared $L^2(D)$ error:
\begin{equation*} 
    \mathcal{R}_{\text{NO}}(\mathcal{G}_{\theta})  = \mathbb{E}_{(a, f) \sim  \mathcal{P}}\left[ \int_{D} \left\|u_i(x) - \mathcal{G}_{\theta}(a_i, f_i)(x)\right\|_2^2dx\right]
\end{equation*}
Neural operators use discretization-invariant layers. Prominent instantiations include Fourier Neural Operators (FNO), which use spectral transforms \citep{li2020fourier}, Graph Neural Operators, which perform graph-based aggregation over sampled points \citep{li2020multipole}, and DeepONets, which employ a trunk–branch decomposition to map input functions to output functions \citep{lu2021learning}.

\subsection{Greedy Optimization} \label{sec:bggreedy}

Consider maximizing $g(S)$ over $S\subseteq \Omega$ with $|S| \leq T$ where $g: 2^{\Omega} \to \mathbb{R}$ is a set function defined on subsets of a set $\Omega$. An exhaustive search is infeasible, so an efficient alternative is the greedy algorithm. It builds the solution subset iteratively by adding $\omega^*$ that yields the largest marginal gain, i.e. $\omega^* =  \argmax_{\omega \in \Omega \setminus S} g(S \cup \omega)$. When $g$ is non-negative, monotone (adding elements never decreases the value), and submodular (diminishing returns), the greedy algorithm achieves a $(1 - e^{-1})$ approximation to the optimal solution \citep{nemhauser1978analysis}. Formally, a function $g$ is submodular if $g(A \ \cup \{\omega\}) - g(A) \geq g(B \ \cup \{\omega\}) - g(B)$ for all $A \subseteq B \subseteq \Omega$ and $\omega \in \Omega \setminus B$; intuitively, delaying an addition cannot increase its benefit. For set function minimization problems, supermodularity, where the inequality flips, plays an analogous role, and the greedy rule achieves constant-factor approximation guarantees \citep{liberty2017greedy}.

The suboptimality of the greedy algorithm has been extensively studied for both set-function minimization \citep{bounia2023approximating}  and maximization \citep{das2018approximate,feige2011maximizing,bian2017guarantees,harshaw2019submodular} when standard assumptions--non-negativity, monotonicity, and sub/supermodularity--are weakened. In contrast, results on greedy sequence maximization \citep{streeter2008online,alaei2021maximizing,zhang2015string,bernardini2020through,van2024performance,tschiatschek2017selecting}--where the ordering of the elements affects the function--have led to sequential analogues of submodularity and monotonicity. We leverage these tools to analyze the suboptimality of the greedy solution to minimizing the final error.


\section{General framework for Hybrid Solvers}

To solve \Cref{eq:mainpdediscrete}, consider the following hybrid iterative update:
\begin{equation}
     u^{\left(t + 1\right)}_h = u^{\left(t\right)}_h + C_{S_t}\left(f_h - \mathcal{L}_h^au^{\left(t \right)}_h\right)
\end{equation}
where $\mathcal{C} =\{C_{j}\}_{j = 1}^K$ is a set of preconditioning functions and $S_t \in [K] = \{1, \dots, K\}$ indexes the function chosen at step $t$. Here, we use ``preconditioning function'' broadly to include classical linear updates ($C_j(x) = C_jx$ where $C_j$ is a matrix), and learned models, such as neural operators. 
This update generalizes the classical iterative update in \Cref{eq:iterativeupdate}, by allowing $C_j$ to be non-linear and enabling the preconditioning function to be adaptively chosen at every step. $\mathcal{C}$ can also accommodate parameterized solver families (e.g., different Jacobi relaxation weights), enabling adaptive selection of solver parameters.  As the number of solvers $K$ grows, the chance of selecting a more effective update increases. Furthermore, HINTS \citep{zhang2024blending} is a special case of this hybrid solver with $K = 2$, where  $C_1$ is a neural operator and $C_2$ is a classical preconditioner. Its routing rule is given by $S_t = \1_{t \bmod \tau > 0} + 1$, which selects $C_1$ every $\tau$ steps and $C_2$ otherwise.

Let $e^{(t)}_h = u_h - u^{(t)}_h$ denote the error at step $t$. Then, 
\begin{align*}
    e^{(t + 1)}_h = u_h -  u^{(t)}_h -C_{S_t}\left(\mathcal{L}_h^au_h - \mathcal{L}_h^au^{\left(t\right)}_h\right) = \left(I- C_{S_t} \circ \mathcal{L}_h^a\right)(e^{(t)}_h)
\end{align*}
where $I$ is the identity map and $I - C_{j} \circ\mathcal{L}_h^a$ is the error propagation function for solver $j$. Here, ``$\circ$'' denotes function composition (matrix multiplication for linear updates).
The objective is to select a sequence $S = (S_1, \dots, S_T)$ that minimizes the error norm after $T$ steps or $\|e^{(T)}_h\|_2^2$:
\begin{equation} \label{eq:errorobjective}
    \min_{S_t \in [K], |S| \leq T} h(S) := \left\|\left(I - C_{S_{|S|}} \circ \mathcal{L}_h^a\right) \circ \dots \circ \left(I - C_{S_1} \circ \mathcal{L}_h^a\right)  \left(e^{(0)}_h \right)\right\|^2_2
\end{equation}
with compositions applied from right to left, so that the $S_1$ update acts on the initial error. Our main analysis treats each solver application as one action and provides guarantees over a fixed number of routing steps. When solver runtimes differ, these guarantees characterize convergence in terms of actions but do not directly imply faster wall-clock convergence. In our experiments, we address this issue by grouping solver iterations into approximately equal-cost macro-actions, as detailed in \Cref{sec:experiments}. \Cref{sec:cost-aware} extends the analysis to actions with genuinely heterogeneous costs.



\section{Greedy Algorithm} \label{sec:greedy}
\Cref{eq:errorobjective} defines a combinatorial optimization problem with a search space exponential in $T$, making exact optimization intractable for large $T$ or $K$. As a starting point, we consider an ``omniscient'' greedy algorithm that assumes access to the true initial error $e^{(0)}_h$. This is unrealistic since knowledge of $e^{(0)}_h$ could directly recover the solution via $u_h = u^{(0)}_h + e^{(0)}_h$, but it provides a useful benchmark.
In \Cref{sec:approxgreedy}, we relax this assumption using a practical learning strategy that is Bayes consistent with this omniscient rule, thereby recovering the guarantees shown below.

As discussed in \Cref{sec:bggreedy}, when supermodularity and monotonicity hold, the greedy rule--such as the one described in \Cref{alg:greedy}--enjoys constant-factor approximation guarantees. However, classical results focus on \textit{set} functions, with only recent extensions made to sequences. Building on this line, we introduce a sequence-based notion of weak supermodularity.

\begin{algorithm} 
\caption{Greedy algorithm for a hybrid PDE solver}\label{alg:greedy}
\begin{algorithmic}
\Require $\{C_{j}\}_{j = 1}^{K}, T, \mathcal{L}_h^a, e^{(0)}_h$
\State $S^{0} \gets \emptyset$
\For{$t < T$}
    \State $S^{t + 1} \gets S^{t} \ \oplus \argmin_{j \in [K]}\|( I - C_{j} \circ \mathcal{L}_{h}^a)(e^{(t)}_h)\|^2_2$
    \State $e^{(t + 1)}_h \gets (I - C_{S_{t+ 1}} \circ \mathcal{L}_{h}^a) (e^{(t)}_h)$
\EndFor
\State \Return $S^{T}$
\end{algorithmic}
\end{algorithm}

Let $\Omega^*$ denote the space of sequences with elements in $\Omega$. For $S, S' \in \Omega^*$, we denote their concatenation as $S \oplus S'$. $S$ is a prefix of $S'$ or $S \preceq S'$ if $S' = S \ \oplus L$ for some $L \in \Omega^*$.
A function $g: \Omega^* \to \mathbb{R}$ is prefix monotonically non-increasing if $g(S \ \oplus S') \leq g(S)$ for all $S, S' \in \Omega^*$, and postfix monotonically non-increasing if $g(S' \ \oplus S) \leq g(S)$ for all $S, S' \in \Omega^*$. A prefix non-increasing function $g$ is sequence supermodular if, for all $S', S \in \Omega^*: S \preceq S'$, it holds that
\begin{equation}\label{eqn:supermod}
    g(S) - g(S \ \oplus \omega ) \geq g(S') - g(S' \ \oplus \omega ), \quad \forall \omega \in \Omega
\end{equation}

However, $h(S)$ (defined in \Cref{eq:errorobjective}) may not satisfy this property in general. Therefore, we introduce weak sequence supermodularity. A prefix non-increasing function $g$ is weakly supermodular with respect to $S' \in \Omega^*$ if there exists $\alpha(S') \geq 1$,  for every $S \in \Omega^{\leq|S'|}$,  such that
\begin{equation}\label{eqn:alpha_supermod}
     g(S) - g(S\ \oplus S') \leq \alpha(S') \sum_{i \in [|S'|]} g(S) - g(S \ \oplus S'_i)
\end{equation}
The parameter $\alpha(S')$ or the supermodularity ratio quantifies deviation from exact sequence supermodularity.
Expanding the $g(S) - g(S\ \oplus S')$ as a telescoping sum 
$\sum_{i = 1}^{|S'|} g(S \ \oplus \left(S'_1, \dots, S'_{i - 1}\right)) - g(S \ \oplus \left(S'_1, \dots, S'_{i}\right))$
shows that the marginal decrease from appending $S_i'$ after its predecessors
is controlled by the effect of appending $S_i'$ directly to $S$. Thus, postponing the inclusion of $S_i'$ cannot yield a significantly larger benefit than adding it earlier. The supermodularity ratio $\alpha(S)$ quantifies the extent to which future gains from delays may exceed immediate gains. 

Having introduced these notions, we now characterize the suboptimality of greedy solutions of weakly supermodular and 
postfix monotonic sequence functions in \Cref{th:suboptgreedy}.
\begin{theorem} \label{th:suboptgreedy}
    Let $g: \Omega^* \to \mathbb{R}$ be a weakly supermodular function with respect to the optimal solution $O = \argmin_{S \in \Omega^T} g(S)$ with a supermodularity ratio of $\alpha(O)$ and 
    postfix monotonicity. Let the greedy solution of length $T$ be $S^T$. If $\phi_T(\alpha) =  \left(1 - \frac{1}{\alpha T}\right)^T$, then
    \begin{equation*}
        g(\emptyset) - g(S^T)\geq \left(1 - \phi_T(\alpha(O))\right)\left(g(\emptyset) - g(O)\right)
    \end{equation*}
\end{theorem}

The proof of \Cref{th:suboptgreedy} appears in \Cref{sec:proofsuboptgreedy}. As $T \to \infty$, $1 - \phi_T(\alpha(O))$ decreases to $1 - e^{-1/\alpha(O)}$, so the worst-case performance of the greedy rule saturates rather than degrades with horizon. Larger $\alpha(O)$, indicating higher reward for delayed inclusions, loosens the bound.  
\Cref{th:suboptgreedy} requires $h$ to be weakly supermodular with respect to the optimum and postfix monotone. Proposition ~\ref{th:weaklyalphasupermodular} establishes weak-supermodularity under zero-preserving contraction assumptions. We defer the proof to \Cref{sec:proofweaklyalphasupermodular}.

\begin{proposition} \label{th:weaklyalphasupermodular}
    Suppose that for all $j \in [K]$, the error propagation function $I - C_j \circ \mathcal{L}_h^a$ is $\rho_j$-Lipschitz continuous with $\rho_j < 1$, and that $(I - C_j \circ \mathcal{L}_h^a) (0_N) = 0_N$. Then, the function $h$ is weakly supermodular with respect to the optimal solution $O$, with 
    \begin{equation*}
        \alpha(O) = \max\left\{\frac{1}{T - \sum_{i = 1}^{T} \rho_{O_i}^2} , 1\right\} 
    \end{equation*}
\end{proposition}
The conditions of Proposition \ref{th:weaklyalphasupermodular} are natural. For classical solvers, the Lipschitz constant is $\|I_N - C_j \mathcal{L}^a_h\|$, which is typically less than $1$
for well-posed linear elliptic PDEs (i.e., the update is damping errors). 
For neural networks, Lipschitz continuity can be enforced via weight regularization \citep{gouk2021regularisation} and a sufficiently trained model should approximate $(\mathcal{L}^a_h)^{-1}$ well enough to make the Lipschitz constant small.
The requirement $(I - C_j \circ \mathcal{L}_h^a) (0_N) = 0_N$
is both natural and desirable: it precludes spurious updates when the residual $f_h - \mathcal{L}^a_h u_h^{(t)}$ is $0$. This holds by design for classical schemes, and it can be enforced for any learned model by excluding bias terms.

The form of $\alpha(O)$ in Proposition \ref{th:weaklyalphasupermodular} highlights that the suboptimality factor in \Cref{th:suboptgreedy} is governed by the collective contraction factors of the solvers chosen by the optimal solution: as $\sum_{i = 1}^{T} \rho_{O_i}^2 \to T$, $\alpha(O)$ grows, resulting in a weaker bound. 
Zero-preserving contraction guarantees prefix monotonicity, but does not in general guarantee postfix monotonicity. A sufficient condition for postfix monotonicity is pairwise commutativity of the error propagation maps: prepended maps can then be commuted past the remaining sequence and applied as contractions to its output. We do not assume commutativity for general solver ensembles; the approximation guarantee applies when the stated postfix-monotonicity condition holds. We empirically assess zero preservation, contraction, and postfix monotonicity for the solvers used in the experiments in \Cref{app:assumption-satisfaction}.


\Cref{th:suboptgreedy} thus indicates that the approximation guarantee is strongest when the sequence is nearly supermodular ($\alpha(O) \approx 1$). Generally, if all error propagation maps share a common orthonormal eigenbasis,
supermodularity is established.
This occurs, for example, for linear, constant-coefficient PDEs with periodic boundary conditions where the solver ensemble includes Jacobi and  a single-layer linear Fourier Neural Operator. 
The proof of \Cref{th:simultaneoussupermodular} is deferred to \Cref{sec:proofsimultaneoussupermodular}.
\begin{proposition} \label{th:simultaneoussupermodular}
    Let $\|I_N - C_j\mathcal{L}_h^a\| \leq 1$ for all $j \in [K]$ and $\left(I - C_j \mathcal{L}_{h}^a\right) = P\Lambda_jP^{-1}$, where $P$ is an orthogonal matrix. Then, $h$ is supermodular. 
\end{proposition}

\section{Approximate Greedy Router} \label{sec:approxgreedy}

The results in \Cref{sec:greedy} indicate that the error reduction from the greedy solution closely matches that of the optimal sequence, but this is predicated on having access to the initial error, $e^{(0)}_h$. Inaccurate error estimates can result in poor solver decisions, causing errors to amplify. To remedy this, we learn a router $r$ that selects solvers myopically, as in \Cref{alg:greedy}, without access to true errors. 

We adopt the following learning setup. Let $\mathcal{A}, \mathcal{F}, \mathcal{U} \subseteq \mathbb{R}^N$ denote the spaces of coefficient, forcing, and solution functions on the grid $G_h(D)$. We assume an application-specific data distribution $\mathcal{P}_{\mathcal{A} \times \mathcal{F}}$ over $\mathcal{A} \times \mathcal{F}$ that reflects test time conditions.
During training, $(a_h, f_h)$ is drawn from $\mathcal{P}_{\mathcal{A} \times \mathcal{F}}$, and a high-accuracy reference solution $u_h$ is computed, providing the true per-step error. The router $r$ is learned offline on this data; at test time, it operates without access to true errors.


At iteration $t$, router $r$ selects a solver using the coefficient $a_h \in \mathcal{A}$, forcing $f_h \in \mathcal{F}$, and the current iterate $u^{(t)}_h \in \mathcal{U}$. If $r(a_h, f_h, u^{(t)}_h) = j$, the next iterate is computed with solver $j$, resulting in an error of $\|( I - C_{j} \circ  \mathcal{L}_{h}^a)(e^{(t)}_h)\|^2_2$. Learning such a router requires minimizing the following loss:
\begin{equation} \label{eq:routerloss}
    l_{\text{route}}\left(r, a_h, f_h, u^{(t)}_h, u_h\right) = \sum_{j = 1}^K \left\|\left( I - C_{j} \circ \mathcal{L}_{h}^a\right)\left(u_h - u^{(t)}_h\right)\right\|^2_2\1_{r(a_h, f_h, u^{(t)}_h) = j}
\end{equation}
with risk $\mathcal{R}_{\text{route}}\left(r\right) = \mathbb{E}_{a_h, f_h \sim \mathcal{P}_{\mathcal{A} \times \mathcal{F}}}[l_{\text{route}}(r, a_h, f_h, u^{(t)}_h, u_h)]$.  
For analysis, we fix the iterate generation via teacher-forcing \citep{williams1989learning,lamb2016professor}: during training, 
the iterates fed to the router are produced by an oracle greedy rollout, not by the router's own past choices. Formally, with $u_h^{(0)} = 0_N$, for $t \geq 1$, 
\begin{equation*}
    j^*_t \in \text{argmin}_{j \in [K]} \left\|\left( I - C_{j} \circ \mathcal{L}_{h}^a\right)\left(e^{(t - 1)}_h\right)\right\|^2_2, \quad u^{(t)}_h = u^{(t - 1)}_h + C_{j^{*}_t}\left(f_h - \mathcal{L}_h^a u_h^{(t - 1)}\right)
\end{equation*}
Thus, $\{u^{(t)}_h\}$ is a deterministic function of $(a_h, f_h)$. At each $t$, $l_{\text{route}}$ is evaluated on $r(a_h, f_h, u^{(t)}_h)$ which takes in the teacher-forced iterate. Finally, the greedy choice $j^*_{t + 1}$ is used to advance $u^{(t + 1)}_h$. Thus the input distribution seen by $r$ depends on $(a_h, f_h)$, not on the router's predictions. However, full teacher forcing induces a distributional mismatch at test time (exposure bias). In practice, we mitigate this distribution shift using a DAgger-style dataset-aggregation procedure. After training on oracle-generated states, we roll out the learned router with controlled exploration, label the newly visited states using the greedy oracle, and retrain on the aggregated data. See \Cref{sec:training details} for details.


Minimizing $\mathcal{R}_{\text{route}}\left(r\right)$ yields the following Bayes-Optimal Router:
\begin{equation}\label{eq:bayesoptimalrouter}
    r^*(a_h, f_h, u^{(t)}_h) \in \text{argmin}_{j \in [K]} \left\|\left( I - C_{j} \circ \mathcal{L}_{h}^a\right)\left(e^{(t)}_h\right)\right\|^2_2
\end{equation} 
which matches the update in \Cref{alg:greedy}. Hence, $l_{\text{route}}$ is consistent with learning the greedy rule.

\subsection{Surrogate Loss}\label{section:surrogate}

Since \Cref{eq:routerloss} is discontinuous and non-convex, direct minimization is intractable in practice. We instead introduce a surrogate loss that (1) is convex, enabling efficient optimization; (2) upper bounds the original loss; and (3) is Bayes consistent, ensuring the Bayes optimal decision of the original loss is preserved upon minimization. Formally, a surrogate $\phi$ is considered to be Bayes consistent with respect to the loss $l$ if 
\begin{equation*}
    \lim_{n \to \infty} \mathcal{R}_{\phi}(f_n) - \mathcal{R}_{\phi}^* = 0 \implies \lim_{n \to \infty} \mathcal{R}_{l}(f_n) - \mathcal{R}_{l}^* = 0
\end{equation*}
where $\mathcal{R}_{l}(f) = \mathbb{E}\left[l(f(X), Y)\right]$ and $\mathcal{R}_{l}^* = \inf_f\mathcal{R}_{l}(f)$. In other words, over the class of all measurable functions, if the risk of a sequence of learned hypotheses $\{f_n\}$ converges to the optimal risk under $\phi$, it also converges to the optimal risk with respect to the original loss $l$.

To define a surrogate loss for the routing problem, consider a set of scoring functions $\mathbf{g} =\{g_j\}_{j = 1}^K$ with $g_j: \mathcal{A} \times \mathcal{F} \times \mathcal{U} \to \mathbb{R}$ and we define the router as $r(a, f, u^{(t)}) = \text{argmax}_{j \in [K]} g_j(a, f, u^{(t)})$. For example, $\mathbf{g}$ can be a neural network with $K$ outputs. Then, minimizing the following surrogate loss yields the same decision as minimizing \Cref{eq:routerloss}:
\begin{equation} \label{eq:routersurrogate}
    \Psi\left(\mathbf{g}, a_h, f_h, u^{(t)}_h, u_h\right) = -\sum_{j = 1}^K \sum_{k = 1}^K\Tilde{c}_k(a_h, u^{(t)}_h, u_h)\1_{k\neq j}\log\left(\frac{\exp\left(g_j(a_h, f_h, u^{(t)}_h)\right)}{\sum_{m = 1}^K\exp\left(g_m(a_h, f_h, u^{(t)}_h)\right)}\right)
\end{equation}
where $\Tilde{c}_j(a_h, u^{(t)}_h, u_h) = \|( I - C_{j} \circ \mathcal{L}_{h}^a)(u_h - u^{(t)}_h)\|^2_2$.
The $\Psi$-risk is denoted by $\mathcal{R}_{\Psi}(\mathbf{g}) = \mathbb{E}_{a_h, f_h \sim \mathcal{P}_{\mathcal{A} \times \mathcal{F}}}[\Psi(\mathbf{g}, a_h, f_h, u^{(t)}_h, u_h)]$. The convexity of $\Psi$ with respect to $\mathbf{g}$ follows from the convexity of negative log-softmax function in its inputs. Moreover, $\Psi$ upper bounds $l_{\text{route}}$ up to a constant factor, and we refer the reader to \Cref{sec:proofupperbound} for the proof. Finally, \Cref{th:consistency} shows that $\Psi$ achieves Bayes consistency with respect to $l_{\text{route}}$; the proof can be found in \Cref{sec:proofconsistency}.

\begin{theorem} \label{th:consistency}
    Let $\Tilde{c}_j(a_h, u^{(t)}_h, u_h) < \Bar{E} < \infty$  for all $j \in [K]$. If there exists $j\in [K]$ such that $\Tilde{c}_j(a_h, u^{(t)}_h, u_h) > E_{min} > 0$, then, for any collection of solvers $\{C_j\}_{j= 1}^K$ and linear discrete operator $\mathcal{L}_h^a$, $\Psi$ is Bayes consistent surrogate for $l_{\text{route}}$.
\end{theorem}
\Cref{th:consistency} is a cost-sensitive analogue of the classical Bayes-consistency of multiclass cross-entropy for 0-1 loss: in the infinite-sample limit, minimizing cross-entropy recovers the true conditional class probabilities, so the induced decision is Bayes optimal. Our result extends this to cross-entropy with instance-dependent weights $\sum_{k \neq j}^K\Tilde{c}_k(a_h, u^{(t)}_h, u_h)$. The uniform upper bound holds when all preconditioning functions are error-damping. It is also reasonable to assume at least one solver cannot annihilate the error in one step, yielding the lower bound. Under these conditions, minimizing $\Psi$ recovers the Bayes-optimal router of \Cref{eq:bayesoptimalrouter}, i.e., the greedy solution of \Cref{eq:errorobjective}. Following the work of \citet{mao2024regression}, the proof uses standard conditional-risk calibration: we relate the excess risks of $l_{\text{route}}$ and $\Psi$ and take the infinite-sample limit.


\section{Related Works}

\paragraph{Hybrid PDE Solvers:}
Early data-driven solvers sought convergence guarantees by predicting components of iterative schemes such as preconditioners, multi-grid smoothers or restriction matrices \citep{taghibakhshi2021optimization,caldana2024deep,kopanivcakova2025deeponet,huang2022learning,katrutsa2020black}. These works, however, do not leverage the neural surrogates' ability to generalize across coefficients and forcings (\Cref{sec:neuraloperators}) and offer only modest speedups over classical solvers. HINTS \citep{zhang2024blending} introduced hybrid solvers that interleave a classical method with a pre-trained DeepONet on a fixed schedule (e.g., 24 Jacobi steps, then one DeepONet correction), accelerating convergence over the standalone numerical solver. 
Subsequent works extend HINTS to new geometries \citep{kahana2023geometry} and analyze error mode damping, replacing DeepONet with alternatives such as MIONet \citep{jin2022mionet} \citep{hu2025hybrid} or FNO \citep{li2020fourier,cui2022fourier}. 
However, these hybrid solvers are limited in two ways: they, firstly, pair the trained surrogate with a \textit{single} numerical solver, and secondly rely on a fixed, heuristic schedule. Our method addresses these shortcomings, resulting in significant empirical improvements (\Cref{sec:experiments}).

\paragraph{Model Routing:} Routing \citep{shnitzer2023large,hu2024routerbench,ding2024hybrid,huang2025routereval} selects, for each input, a model from a fixed set to optimize a task metric under cost or latency constraints. Simple heuristics--e.g., thresholding a cheap model's uncertainty estimates \citep{chuang2024learning,chuang2025confident}--often poorly balance the cost-accuracy tradeoff, motivating learned routers that map inputs to the best model \citep{hari2023tryage,mohammadshahi2024routoo,vsakota2024fly}. We instead route among numerical and neural PDE solvers at each iteration. Our work is the first to apply routing to the problem of learning hybrid PDE solvers, in contrast to prior approaches that use fixed schedules. Additionally, by exploiting the algebraic structure of solvers, we derive theoretical guarantees for our routing strategy (\Cref{sec:greedy}), unlike typical model-routing settings.

\paragraph{Mixture of Experts:} Classical mixture-of-experts (MoE) \citep{jacobs1991adaptive,jordan1994hierarchical} uses a gating network to combine the outputs of multiple experts, trained jointly with the gate. In modern LLMs, MoE instead performs sparse, token-level routing to a subset of in-layer experts, enabling large capacity without proportional compute, typically with load-balancing mechanisms \citep{shazeer2017outrageously,lepikhin2020gshard,fedus2022switch,jiang2024mixtral}. 
Our method can be regarded as gating of a different kind: ``experts'' (solvers) are not jointly trained with the router, and a single router is reused across iterations, unlike MoE which employs layer-specific gates.

\section{Experiments}\label{sec:experiments}

We evaluate the approximate greedy router on the Poisson, Convection-Diffusion (ConvDiff), and Anistropic Diffusion (Aniso) equations posed on the unit domain $D = [0, 1]^2$:
\begin{align*}
\text{\textbf{Poisson:}}\qquad
&-\Delta u(\mathbf{x}) = f(\mathbf{x}),
&& \mathbf{x}\in D,\\
\text{\textbf{ConvDiff:}}\qquad
&-\Delta u(\mathbf{x})
+20\frac{\partial u}{\partial x_1}(\mathbf{x})
+20\frac{\partial u}{\partial x_2}(\mathbf{x})
=f(\mathbf{x}),
&& \mathbf{x}\in D,\\
\text{\textbf{Aniso:}}\qquad
&-0.01\frac{\partial^2u}{\partial x_1^2}(\mathbf{x})
-\frac{\partial^2u}{\partial x_2^2}(\mathbf{x})
=f(\mathbf{x}),
&& \mathbf{x}\in D.
\end{align*}
We impose periodic boundary conditions and center $f$ to satisfy $\int_D f(\mathbf{x})d\mathbf{x} = 0$; Dirichlet results are deferred to \Cref{sec:dirichlet}. The domain $D$ is discretized on a $128 \times 128$ uniform grid; results on a finer grid are provided in \Cref{sec:finer_grids}. Data is sampled from a zero-mean Hierarchical Gaussian Random Field on the periodic domain with covariance operator $\alpha(- \Delta + \beta I)^{-\gamma}$, where $\alpha, \beta, \gamma$ vary across samples (See \Cref{sec:training details}).  We run two experimental settings. 
First, we restrict to solver pairs ($K = 2$) and compare against HINTS and standalone classical solvers. Second, we assess performance as the solver ensemble grows; since HINTS does not support this setting, we simply compare each ensemble router with its strongest constituent pairwise router. 
Over 64 test samples,
we report the mean charged wall-clock time to a relative solution error tolerance,
$\|u_h-u_h^{(t)}\|_2/\|u_h\|_2$. Solution error is used only for evaluation, not
as a stopping criterion. The paired and growing-ensemble studies use 
tolerances of $10^{-12}$ and $10^{-6}$, respectively.

\paragraph{Paired Solver Experiments:} We seek to optimally solve the aforementioned PDEs with access to a single numerical scheme and a trained neural operator. In particular, we consider settings coupling the neural solver with Jacobi, Gauss-Seidel (GS), and Symmetric Gauss-Seidel (SymGS) solvers, successive over-relaxation (SOR) with $\omega = 1.5$, and multigrid (MG). For each pair, we compare our learned router against the corresponding single-solver schedules and HINTS variants. The reported HINTS results use the best-performing cadence for each PDE and solver family. To account for heterogeneous solver runtimes, we measure each solver update time $c_j$ and define a classical macro-action as $m_j = \max\{1, \text{round}(c_{NO}/c_j)\}$ consecutive iterations, using one neural correction as the reference action. Macro sizes are calibrated on validation data and frozen before testing. We train a separate MLP router for each pair using the loss in \Cref{eq:routersurrogate}.
Training details of the neural operators and the routers can be found in \Cref{sec:training details}. To ensure fairness, comparisons are restricted to methods with identical solver access, and are grouped by solver family (e.g., Jacobi-related methods). We additionally evaluate a \emph{true-greedy} oracle that, at each routing round, evaluates every available macro-action and selects the one yielding the largest immediate error reduction. This oracle has access to the ground-truth error and is therefore not implementable in practice, but helps evaluate how well the learned router approximates the greedy policy. 



\begin{table}[t]
\centering
\caption{Selected-action wall-clock replay time in milliseconds (lower is better), reported as mean (standard error) in milliseconds over successful cases. Brackets indicate the number of successes when coverage is incomplete, and a dash indicates no successful cases. Bold marks the fastest deployable method with full coverage in each solver family; True-Greedy is an oracle.}
\resizebox{\textwidth}{!}{%
\begin{tabular}{lccc}
\hline
Equation & Poisson & Conv & Aniso \\
\hline
Methods & Time (ms) & Time (ms) & Time (ms) \\
\hline
\multicolumn{4}{c}{Jacobi-related Solvers} \\
\hline
Jacobi Only & 6428.52 (4.15) & 10146.76 (14.21) & -- \\
HINTS-Jacobi ($m=50/50/50$) & 2288.27 (125.46) & 3278.90 (161.52) & \textbf{1845.77 (126.04)} \\
Learned Greedy-Jacobi & \textbf{1280.80 (64.97)} & \textbf{2900.52 (181.73)} & 1117.90 (77.21) [62/64] \\
\textit{True-Greedy-Jacobi} & \textit{1226.02 (66.29)} & \textit{2120.53 (103.75)} & \textit{1130.38 (75.91)} \\
\hline
\multicolumn{4}{c}{GS-related Solvers} \\
\hline
GS Only & 10224.51 (6.58) & 9921.16 (13.81) & -- \\
HINTS-GS ($m=5/5/25$) & 48.26 (0.31) & 45.60 (0.15) & 216.46 (1.58) \\
Learned Greedy-GS & \textbf{35.63 (0.71)} & \textbf{35.60 (0.75)} & \textbf{55.53 (2.28)} \\
\textit{True-Greedy-GS} & \textit{34.94 (0.81)} & \textit{35.77 (0.73)} & \textit{53.38 (2.42)} \\
\hline
\multicolumn{4}{c}{SymGS-related Solvers} \\
\hline
SymGS Only & 9210.31 (6.18) & 9770.15 (13.53) & -- \\
HINTS-SymGS ($m=5/5/5$) & 37.57 (0.08) & 40.77 (0.36) & 236.63 (2.77) \\
Learned Greedy-SymGS & \textbf{35.04 (0.68)} & \textbf{36.69 (0.73)} & \textbf{69.44 (4.28)} \\
\textit{True-Greedy-SymGS} & \textit{34.21 (0.79)} & \textit{36.06 (0.82)} & \textit{55.12 (2.73)} \\
\hline
\multicolumn{4}{c}{SOR(1.5)-related Solvers} \\
\hline
SOR(1.5) Only & 3408.10 (2.17) & 2391.65 (3.32) & -- \\
HINTS-SOR(1.5) ($m=10/10/10$) & 56.19 (0.60) & 56.53 (0.16) & 121.00 (0.83) \\
Learned Greedy-SOR(1.5) & \textbf{44.60 (0.97)} & \textbf{45.97 (1.06)} & \textbf{47.98 (0.80)} \\
\textit{True-Greedy-SOR(1.5)} & \textit{42.68 (0.95)} & \textit{46.39 (1.04)} & \textit{46.58 (0.81)} \\
\hline
\multicolumn{4}{c}{MG-related Solvers} \\
\hline
MG Only & 36.19 (0.36) & 41.51 (0.39) & 871.50 (6.76) \\
HINTS-MG ($m=5/5/5$) & 32.09 (0.32) & 36.46 (0.30) & 337.11 (3.57) \\
Learned Greedy-MG & \textbf{29.93 (0.46)} & \textbf{33.71 (0.56)} & \textbf{82.76 (5.05)} \\
\textit{True-Greedy-MG} & \textit{29.86 (0.46)} & \textit{33.35 (0.49)} & \textit{78.04 (5.00)} \\
\hline
\end{tabular}}
\label{tab:errorcomparison}
\end{table}

\begin{figure}[t]
    \centering
    \includegraphics[width=0.5\linewidth]{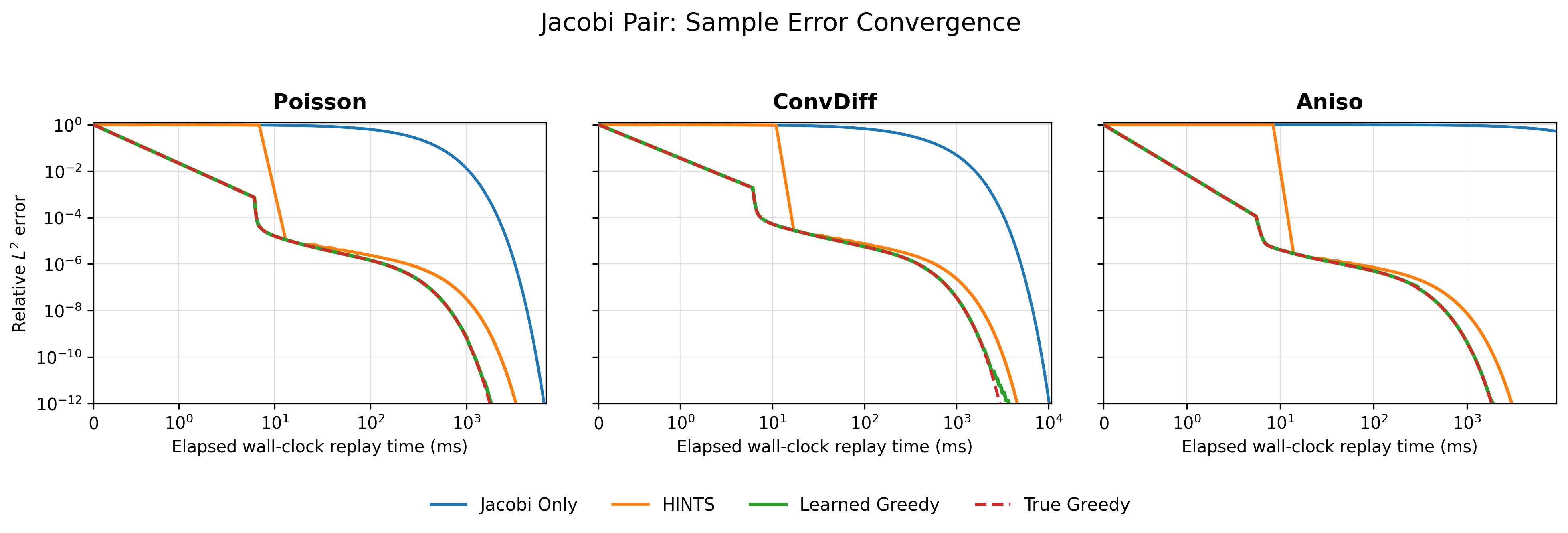}    \caption{Convergence histories for representative test instances using Jacobi. Columns: Poisson, Convection Diffusion, and anisotropic Diffusion. Greedy yields the fastest monotone decay. 
    }
    \label{fig:errorcomparison}
\end{figure}
\begin{figure}
    \centering
    \includegraphics[width=0.5\linewidth]{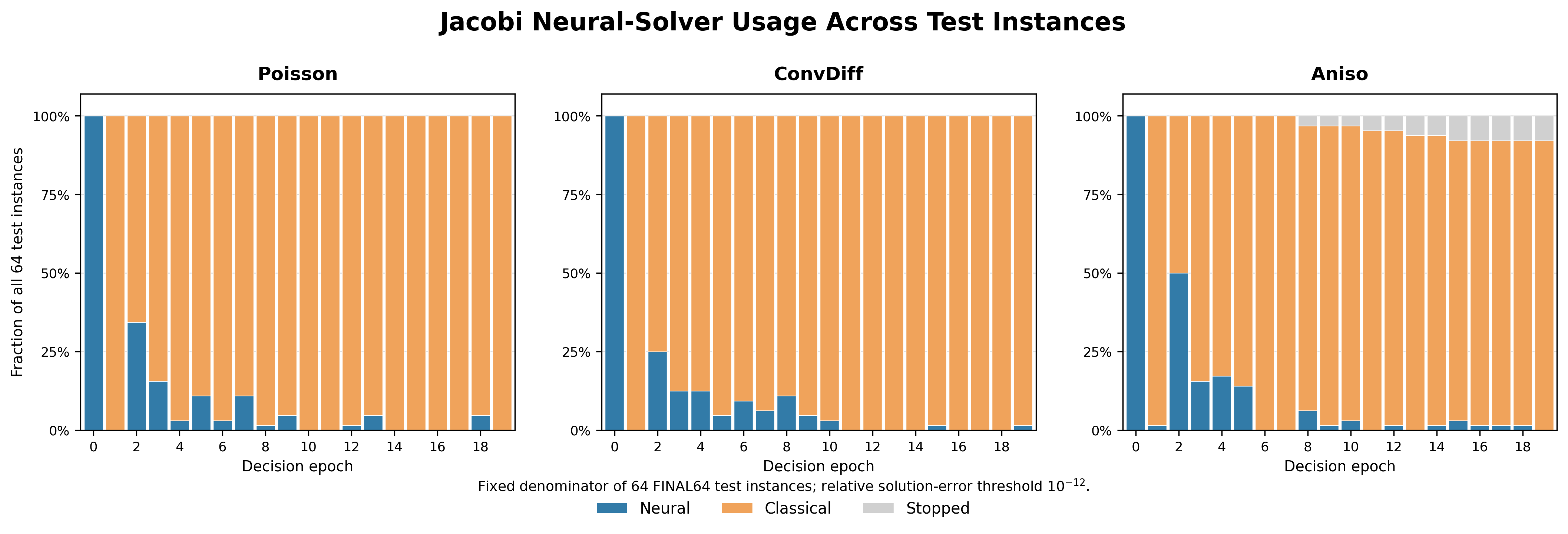}
    \caption{Action usage of the learned Jacobi router across 64 test instances. Bars show the fraction assigned to the neural action, the classical action, or already converged at each of the first 20 routing rounds. Results for other solvers appear in \Cref{sec:solverdecisions}.}
    \label{fig:router_strat}
\end{figure}
As shown in \Cref{tab:errorcomparison}, the greedy router generally reduces time-to-solution relative to its single-solver and HINTS counterparts. Moreover, its performance closely tracks that of the true-greedy oracle, indicating that our surrogate-based training procedure effectively approximates the greedy strategy. The representative convergence histories in \Cref{fig:errorcomparison} illustrate an important distinction between adaptive and fixed routing: a well-timed neural-solver update can substantially accelerate convergence, whereas a poorly timed update can increase the error. The neural-operator usage patterns in \Cref{fig:router_strat} further show that the router does not invoke the learned correction at a fixed frequency. Instead, it adapts neural-operator usage to the evolving solver state.

\paragraph{Solver Ensembles Experiments:} We consider routing over more than two solvers. For each PDE, we construct nested ensembles
$\mathcal{A}_3\subset\cdots\subset\mathcal{A}_6$ using a PDE-specific action order frozen on validation data. For Poisson and ConvDiff, the available actions include weighted Jacobi ($\omega=0.67$), GS, SOR($1.5$), MG, a Fourier neural operator, and a local neural correction. For anisotropic diffusion, we instead use \(x_2\)-line Jacobi and semi-coarsened MG to account for its strong directional coupling, together with SOR and global and directional neural corrections. We compare each ensemble router with its strongest constituent pairwise router, selected on validation data. As shown in \Cref{tab:errorcomparison2}, $\mathcal{A}_3$ substantially reduces time-to-solution relative to the strongest pairwise router across all three PDEs, and the improvement persists as the action set grows. The improvement persists as the Poisson and convection–diffusion ensembles grow, while the larger anisotropic ensembles retain smaller gains. Thus, most of the benefit comes from a small, PDE-appropriate set of complementary actions.
Additional experiments are provided in \Cref{sec:moreexperimentalresults}. 
\begin{table}
\centering
\caption{Selected-action wall-clock replay time for growing solver ensembles, reported in milliseconds as mean (standard error) over successful cases. The pairwise comparator is selected on validation data. Brackets indicate incomplete coverage, a dash indicates no successes, and one-sided paired-test $p$-values are reported for full-coverage comparisons.}
\label{tab:errorcomparison2}
\resizebox{\textwidth}{!}{%
\begin{tabular}{llcccccc}
\hline
\multirow{2}{*}{\textbf{Actions}} &
\multirow{2}{*}{\textbf{Method}} &
\multicolumn{2}{c}{\textbf{Poisson}} &
\multicolumn{2}{c}{\textbf{ConvDiff}} &
\multicolumn{2}{c}{\textbf{Aniso}} \\
\cline{3-8}
& & Time (ms) & $p$-value
& Time (ms) & $p$-value
& Time (ms) & $p$-value \\
\hline

\multirow{2}{*}{$\mathcal{A}_3$}
& Best pairwise router
& 9.21 (0.69) & --
& 18.56 (0.83) & --
& 16.92 (0.22) & -- \\
& Router$(\mathcal{A}_3)$
& \textbf{4.33 (0.27)} & $<10^{-3}$
& \textbf{6.99 (0.30)} & $<10^{-3}$
& \textbf{14.45 (0.37)} & $<10^{-3}$ \\
\hline

\multirow{2}{*}{$\mathcal{A}_4$}
& Best pairwise router
& 9.21 (0.69) & --
& 16.76 (0.76) & --
& 14.93 (0.32) & -- \\
& Router$(\mathcal{A}_4)$
& \textbf{4.33 (0.27)} & $<10^{-3}$
& \textbf{7.05 (0.30)} & $<10^{-3}$
& \textbf{14.60 (0.38)} & $0.057$ \\
\hline

\multirow{2}{*}{$\mathcal{A}_5$}
& Best pairwise router
& 9.21 (0.69) & --
& 14.09 (0.65) & --
& 14.93 (0.32) & -- \\
& Router$(\mathcal{A}_5)$
& \textbf{4.33 (0.27)} & $<10^{-3}$
& \textbf{7.20 (0.31)} & $<10^{-3}$
& \textbf{14.56 (0.38)} & $0.030$ \\
\hline

\multirow{2}{*}{$\mathcal{A}_6$}
& Best pairwise router
& 9.21 (0.69) & --
& 14.09 (0.65) & --
& 14.93 (0.32) & -- \\
& Router$(\mathcal{A}_6)$
& \textbf{4.40 (0.28)} & $<10^{-3}$
& \textbf{7.16 (0.31)} & $<10^{-3}$
& \textbf{14.60 (0.38)} & $0.058$ \\
\hline
\end{tabular}}
\end{table}

\section{Discussion} \label{sec:discussion}

We have introduced an adaptive method for selecting solvers that efficiently minimizes the final error of a PDE in iterative methods, opening many directions for future work. Our neural solvers are trained without accounting for their downstream roles in correction term prediction; jointly training the ML solver and router could yield larger gains. Reinforcement learning may provide a practical approach for learning routing policies that directly optimize solution accuracy under compute budgets and extend naturally to continuously parameterized solver families. Finally, framing routing as an online learning problem could enable adaptation to unknown deployment conditions. 

\section*{AI Use Statement}

We used generative AI tools to provide feedback on experimental methodology, assist with the interpretation of experimental results, and support the development and debugging of experimental and evaluation code. We also used
generative AI for recommended-disclosure tasks, including drafting and editing portions of the manuscript for clarity, suggesting organizational revisions, and limited brainstorming about experimental presentation. We did not use generative AI to formulate the central theoretical claims, supply ingredients
for or write the mathematical proofs, generate research data, conduct literature review, or identify relevant literature. All AI-assisted code, analyses, and text were reviewed and verified by the authors. We take responsibility for the final content of this work, including all claims, experimental results, and artifacts produced with the aid of generative AI.

\bibliography{mybib}
\bibliographystyle{plainnat}

\newpage

\onecolumn

\begin{appendices}

\crefalias{section}{appsec}
\crefalias{subsection}{appsec}

\section{Proofs for \Cref{sec:greedy}} \label{sec:proofgreedy}

\subsection{Proof of Proposition \ref{th:alpha=1}} \label{sec:proofalpha=1}

\begin{proposition} \label{th:alpha=1}
    Any prefix monotonically non-increasing sequence supermodular function $g$ is weakly supermodular with respect to all sequences $S \in \Omega^*$ with $\alpha(S) = 1$
\end{proposition}

\begin{proof}

This proof is adapted from \citet{liberty2017greedy}.

If $g$ is sequence supermodular then,
\begin{align*}
    &g(S) - g(S \ \oplus S' ) \\
    &= \sum_{i = 1}^{|S'|} g(S \ \oplus \left(S'_1, \dots, S'_{i - 1}\right)) - g(S \ \oplus \left(S'_1, \dots, S'_{i}\right)) \\
    &\overset{(a)}{\leq} \sum_{i = 1}^{|S'|} g(S) - g(S \ \oplus  S'_{i}) \\
    &\leq |S'|\max_{i \in [|S'|]} g(S) - g(S \ \oplus  S'_{i})
\end{align*}

(a) by supermodularity

\end{proof}

\subsection{Proof of \Cref{th:suboptgreedy}} \label{sec:proofsuboptgreedy}

\begin{customtheorem}{\ref{th:suboptgreedy}}
    Let $g: \Omega^* \to \mathbb{R}$ be a weakly supermodular function with respect to the optimal solution $O = \argmin_{S \in \Omega^T} g(S)$ with a supermodularity ratio of $\alpha(O)$ and 
    postfix monotonicity. Let the greedy solution of length $T$ be $S^T$. If $\phi_T(\alpha) =  \left(1 - \frac{1}{\alpha T}\right)^T$, then
    \begin{equation*}
        g(S^T) \leq  \left(1 - \phi_T(\alpha(O))\right)
        g(O) +  \phi_T(\alpha(O))g(\emptyset)
    \end{equation*}
\end{customtheorem}

\begin{proof}

This proof strategy is inspired by \citet{streeter2008online} and \citet{liberty2017greedy}

\begin{align*}
    g(S^t) -g(O) &\overset{(a)}{\leq} g(S^t) - g(S^t\ \oplus O) \\
    &=\sum_{i = 1}^{|O|} g(S^t \oplus \left(o_1, \dots o_{i - 1}\right)) -g(S^t\ \oplus \left(o_1, \dots, o_{i}\right))\\
    &\overset{(b)}{\leq}  \alpha(O)\sum_{i = 1}^{|O|} g(S^t) -g(S^t\ \oplus  o_{i}) \\
    &\leq \alpha(O)|O|\max_{i \in [|O|]} g(S^t) -g(S^t\ \oplus  o_{i}) \\
    &\leq \alpha(O) T g(S^t) -  \alpha(O) T\min_{\omega \in \Omega} g(S^t\ \oplus  \omega) \\
    &= \alpha(O) Tg(S^t) -  \alpha(O) Tg(S^{t + 1})
\end{align*}
(a) by postfix monotonicity, (b) by weak supermodularity. 

After rearranging the inequality, we get:

\begin{align*}
    g(S^{t + 1}) &\leq \frac{1}{\alpha(O) T}\left(g(O) +  \left(\alpha(O) T - 1\right)g(S^t)\right) \\
    &= \frac{1}{\alpha(O) T}g(O) + \left(1 - \frac{1}{\alpha(O) T}\right)g(S^t)
\end{align*}

When recursively applying this inequality, we get:

\begin{align*}
     g(S^{T}) &\leq \frac{1}{\alpha(O) T}g(O)\sum_{i = 0}^{T - 1}\left(1 - \frac{1}{\alpha(O) T}\right)^i + \left(1 - \frac{1}{\alpha(O) T}\right)^Tg(\emptyset) \\
     &= \left(1 - \left(1 - \frac{1}{\alpha(O) T}\right)^T\right)g(O) + \left(1 - \frac{1}{\alpha(O) T}\right)^Tg(\emptyset)
\end{align*}

\end{proof}

\subsection{Proof of Proposition \ref{th:weaklyalphasupermodular}} \label{sec:proofweaklyalphasupermodular}


\begin{customproposition}{\ref{th:weaklyalphasupermodular}}
     Suppose that for all $j \in [K]$, the error propagation function $I - C_j \circ \mathcal{L}_h^a$ is $\rho_j$-Lipschitz continuous with $\rho_j < 1$, and that $(I - C_j \circ \mathcal{L}_h^a) (0_N) = 0_N$. Then, the function $h$ is weakly supermodular with respect to the optimal solution $O$, with 
     \begin{equation*}
         \alpha(O) = \max\left\{\frac{1}{T - \sum_{i = 1}^{T} \rho_{O_i}^2} , 1\right\}   
     \end{equation*}
\end{customproposition}

\begin{proof}
For brevity, we use the notation $\left(g_1 \circ \dots \circ g_T\right)(x) = \circ_{t = 1}^{T} g_t(x)$, where composition is applied from right to left so that $g_T$ acts first. In this proof, we use a few properties of Lipschitz continuous functions:
\begin{itemize}
    \item \textbf{Property 1:} If $g$ is $\rho$-Lipschitz continuous and $g(0) = 0$, then $\|g(x)\|_2 = \|g(x) - 0\|_2 = \|g(x) - g(0)\|_2 \leq \rho \|x - 0\|_2 = \rho\|x\|_2$
    \item \textbf{Property 2:} If $g_1$ and $g_2$ are Lipschitz continuous functions with Lipschitz constants of $\rho_1$ and $\rho_2$ respectively, the Lipschitz constant of $g_1 + g_2$ and $g_1 - g_2$ is $\rho_1 + \rho_2$.
    \item \textbf{Property 3:} If $g_1$ and $g_2$ are Lipschitz continuous functions with Lipschitz constants of $\rho_1$ and $\rho_2$ respectively, the Lipschitz constant of $g_1 \circ g_2$ is $\rho_1\rho_2$.
\end{itemize}

In order to prove weakly-$\alpha$-supermodularity, we must first prove prefix monotonicity.

\textbf{Prefix monotonicity: } Let $S \preceq S'$ where $S' = S \ \oplus N$.
\begin{align*}
    h(S') &= \left \|\circ_{t = |S'|}^1\left(I - C_{S'_t} \circ \mathcal{L}_h^a\right) \left(e^{(0)}_h\right)\right\|^2_2 \\
    &= \left \|\circ_{t = |S'|}^{|S| + 1}\left(I - C_{S'_t}  \circ \mathcal{L}_h^a\right)\circ \circ_{t = |S|}^1\left(I - C_{S'_t}  \circ \mathcal{L}_h^a\right) \left(e^{(0)}_h\right)\right\|^2_2 \\
    &\overset{(a)}{\leq}\left( \prod_{t = |S| + 1 }^{|S'|} \rho_{S'_t}^2 \right)\left\|\circ_{t = |S|}^1\left(I - C_{S_t} \circ \mathcal{L}_h^a\right) \left(e^{(0)}_h\right)\right\|^2_2 \\
    &\leq h(S)
\end{align*}
(a) by Property 1

\textbf{Weak supermodularity: } 
We will upper bound $\alpha(O)$ by providing an upper bound for $ h(S) - h(S\ \oplus O) $ and a lower bound for $\sum_{i = 1}^T h(S) - h(S\ \oplus O_i) $. 
\begin{align*}
    h(S) - h(S\ \oplus O) &\leq h(S)
\end{align*}
since $h$ is always non-negative


To lower bound $\sum_{i = 1}^{|O|} h(S) - h(S\ \oplus O_i) $
\begin{align*}
    \sum_{i = 1}^{|O|} h(S) - h(S\ \oplus O_i) &= \sum_{i = 1}^{|O|} \left \|\circ_{t = |S|}^1\left(I - C_{S_t} \circ \mathcal{L}_h^a\right) \left(e^{(0)}_h\right)\right\|^2_2 - \left \|\left(I - C_{O_i} \circ \mathcal{L}_h^a\right) \circ \circ_{t = |S|}^1\left(I - C_{S_t} \circ\mathcal{L}_h^a\right) \left(e^{(0)}_h\right)\right\|^2_2  \\
    &\overset{(a)}{\geq} \sum_{i = 1}^{|O|} \left \|\circ_{t = |S|}^1\left(I - C_{S_t} \circ\mathcal{L}_h^a\right) \left(e^{(0)}_h\right)\right\|^2_2 - \rho_{O_i}^2\left\|\circ_{t = |S|}^1\left(I - C_{S_t}\circ\mathcal{L}_h^a\right) \left(e^{(0)}_h\right)\right\|^2_2 \\
    &= \left \|\circ_{t = |S|}^1\left(I_N - C_{S_t}\circ\mathcal{L}_h^a\right) \left(e^{(0)}_h\right)\right\|^2_2 \left(T - \sum_{i = 1}^{T} \rho_{O_i}^2\right) 
\end{align*}
(a) by Property 3

Finally,

\begin{align*}
    \frac{h(S) - h(S\ \oplus O)}{T\max_i h(S) - h(S\ \oplus O_i)} &\leq \max\left\{ \frac{ \left \|\circ_{t = |S|}^1\left(I_N - C_{S_t}\circ\mathcal{L}_h^a\right) \left(e^{(0)}_h\right)\right\|^2_2}{\left \|\circ_{t = |S|}^1\left(I_N - C_{S_t}\circ\mathcal{L}_h^a\right) \left(e^{(0)}_h\right)\right\|^2_2 \left(T - \sum_{i = 1}^{T} \rho_{O_i}^2\right) }, 1\right\}\\
    &= \max\left\{\frac{1}{T - \sum_{i = 1}^{T} \rho_{O_i}^2} , 1\right\} 
\end{align*}

\end{proof}

\subsection{Proof of \Cref{th:simultaneousallocation}} \label{sec:proofsimultaneousallocation}

\begin{lemma} \label{th:simultaneousallocation}
     Let $\left(I - C_j \mathcal{L}_{h}^a\right) = P\Lambda_jP^{-1}$ where $P$ is an orthogonal matrix and $\Lambda_j = \mathrm{diag}(\lambda_{j1}, \dots, \lambda_{jN})$. If $P^{-1}e^{(0)}_h = z$, then the following equality holds:
     \begin{equation} \label{eq:simultaneous}
          h(S) = \sum_{i = 1}^N z_i^2 \prod_{j = 1}^K \lambda_{ji}^{2m_j(S)}
     \end{equation}
     where $m_j(S) = \sum_{t= 1}^{|S|}\1_{S_t = j}$
\end{lemma}

\begin{proof}

\begin{align*}
    h(S) &= \left\|\prod_{t = |S|}^1\left(I_N - C_{S_t} \mathcal{L}_{h}^a\right)  e^{(0)}_h \right\|^2 \\
    &=  \left\|\prod_{t = |S|}^1\left(P\Lambda_{S_t}P^{-1}\right)  e^{(0)}_h \right\|^2 \\
    &=  \left\|P\prod_{t = |S|}^1\Lambda_{S_t}P^{-1}  e^{(0)}_h \right\|^2 \\
    &= \left\|\prod_{t = |S|}^1\Lambda_{S_t}P^{-1}  e^{(0)}_h \right\|^2\\
    &=  \sum_{i = 1}^N z_i^2 \prod_{t = T}^1 \lambda_{S_ti}^2 \\
    &=  \sum_{i = 1}^N z_i^2 \prod_{j = 1}^K \lambda_{ji}^{2m_j(S)} 
\end{align*}

\end{proof}

\subsection{Proof of \Cref{th:simultaneoussupermodular}} \label{sec:proofsimultaneoussupermodular}
\begin{customproposition}{\ref{th:simultaneoussupermodular}}
    Let $\|I_N - C_j\mathcal{L}_h^a\| \leq 1$ for all $j \in [K]$ and $\left(I - C_j \mathcal{L}_{h}^a\right) = P\Lambda_jP^{-1}$, where $P$ is an orthogonal matrix. Then, $h$ is supermodular. 
\end{customproposition}

\begin{proof}

Let $S \preceq S'$ where $S' = S \ \oplus B$. By \Cref{th:simultaneousallocation},
\begin{align*}
    h(S) &= \sum_{i = 1}^N z_i^2 \prod_{j = 1}^K \lambda_{ji}^{2m_j(S)} 
\end{align*}
where $m_j(S) = \sum_{t = 1}^{|S|} \1_{S_t = j}$ be the number of times a sequence $S$ calls the solver $j$. Recall that $h$ is considered sequence supermodular if $\forall S', S \in \Omega^*$ such that $S \preceq S'$, it holds that
\begin{align*}
    h(S) - h(S \ \oplus \omega ) \geq h(S') - h(S' \ \oplus \omega )
\end{align*}
\begin{align*}
    h(S) - h(S \ \oplus \omega) &= \sum_{i = 1}^N z_i^2 \prod_{j = 1}^K \lambda_{ji}^{2m_j(S)} - \sum_{i = 1}^N z_i^2 \lambda_{\omega i}^{2}\prod_{k = 1}^K \lambda_{ki}^{2m_k(S)} \\
    &=  \sum_{i = 1}^N \left(1 - \lambda_{\omega i}^2\right) z_i^2 \prod_{j = 1}^K \lambda_{ji}^{2m_j(S)}
\end{align*}
Similarly, 
\begin{align*}
    h(S') - h(S' \ \oplus \omega) &= \sum_{i = 1}^N \left(1 - \lambda_{\omega i}^2\right) z_i^2 \prod_{j = 1}^K \lambda_{ji}^{2m_j(S')} \\
    &\overset{(a)}{=}  \sum_{i = 1}^N \left(1 - \lambda_{\omega i}^2\right) z_i^2 \prod_{j = 1}^K \lambda_{ji}^{2\left(m_j(S) + m_j(B)\right)} \\
    &\overset{(b)}{\leq } \sum_{i = 1}^N \left(1 - \lambda_{\omega i}^2\right) z_i^2 \prod_{j = 1}^K \lambda_{ji}^{2m_j(S)} \\
    &= h(S) - h(S \ \oplus \omega)
\end{align*}

(a) Since $m_j(S') = \sum_{t = 1}^{|S'|} \1_{S_t' = j} = \sum_{t = 1}^{|S|} \1_{S_t' = j} + \sum_{t = |S|}^{|S'|} \1_{S_t' = j} = \sum_{t = 1}^{|S|} \1_{S_t = k} + \sum_{t = 1}^{|B|} 1_{B_t = j} = m_j(S) + m_j(B)$, (b) since $|\lambda_{ji}|\leq 1$ and $m_j(B) \geq 0$ for all $j \in [K]$, the additional factors $\lambda_{ji}^{2m_j(B)}$ lie in $[0,1]$.

\end{proof}

\section{Proofs for \Cref{sec:approxgreedy}} \label{sec:proofapproxgreedy}

\subsection{Proof of \Cref{th:routingequivalence}} \label{sec:proofroutingequivalence}

\begin{lemma} \label{th:routingequivalence}

For any set of preconditioning functions $\mathcal{C}$, any discrete operator $\mathcal{L}_h^a$, any router $r$, any $a_h, f_h, u_h^{(t)}, u_h \in \mathcal{A} \times \mathcal{F} \times \mathcal{U} \times \mathcal{U}$, the following equality holds true:
\begin{align*}
     l_{\text{route}}\left(r, a_h, f_h, u^{(t)}_h, u_h\right) = &\sum_{j = 1}^K \sum_{k = 1}^K\left\|\left( I - C_{k} \circ \mathcal{L}_{h}^a\right)\left(u_h - u^{(t)}_h\right)\right\|^2_2\1_{k\neq j}\1_{r(a_h, f_h, u^{(t)}_h) \neq j} \\
     &-\left(K - 2\right)\sum_{j = 1}^K\left\|\left( I - C_{j} \circ \mathcal{L}_{h}^a\right)\left(u_h - u^{(t)}_h\right)\right\|^2_2
\end{align*}
    
\end{lemma}

\begin{proof}
Note that $\sum_{j = 1}^K \1_{r(a_h, f_h, u^{(t)}_h) \neq j}= K - 1$
\begin{align*}
    l_{\text{route}}\left(r, a_h, f_h, u^{(t)}_h, u_h\right) &= \sum_{j = 1}^K \left\|\left( I - C_{j} \circ \mathcal{L}_{h}^a\right)\left(u_h - u^{(t)}_h\right)\right\|^2_2\1_{r(a_h, f_h, u^{(t)}_h) = j} \\
    &= \sum_{j = 1}^K \left\|\left( I - C_{j} \circ \mathcal{L}_{h}^a\right)\left(u_h - u^{(t)}_h\right)\right\|^2_2 - \sum_{j = 1}^K \left\|\left( I - C_{j} \circ \mathcal{L}_{h}^a\right)\left(u_h - u^{(t)}_h\right)\right\|^2_2\1_{r(a_h, f_h, u^{(t)}_h) \neq j} \\
    &= \sum_{j = 1}^K \left\|\left( I - C_{j} \circ \mathcal{L}_{h}^a\right)\left(u_h - u^{(t)}_h\right)\right\|^2_2 - \sum_{j = 1}^K \left\|\left( I - C_{j} \circ \mathcal{L}_{h}^a\right)\left(u_h - u^{(t)}_h\right)\right\|^2_2\1_{r(a_h, f_h, u^{(t)}_h) \neq j} \\
    &+ \left(K - 1\right)\sum_{k = 1}^K\left\|\left( I - C_{k} \circ \mathcal{L}_{h}^a\right)\left(u_h - u^{(t)}_h\right)\right\|^2_2 - \left(K - 1\right)\sum_{k = 1}^K\left\|\left( I - C_{k} \circ \mathcal{L}_{h}^a\right)\left(u_h - u^{(t)}_h\right)\right\|^2_2 \\
    &= \sum_{j = 1}^K \left\|\left( I - C_{j} \circ \mathcal{L}_{h}^a\right)\left(u_h - u^{(t)}_h\right)\right\|^2_2 - \sum_{j = 1}^K \left\|\left( I - C_{j} \circ \mathcal{L}_{h}^a\right)\left(u_h - u^{(t)}_h\right)\right\|^2_2\1_{r(a_h, f_h, u^{(t)}_h) \neq j} \\
    &+ \sum_{j = 1}^K\1_{r(a_h, f_h, u^{(t)}_h) \neq j}\sum_{k = 1}^K\left\|\left( I - C_{k} \circ \mathcal{L}_{h}^a\right)\left(u_h - u^{(t)}_h\right)\right\|^2_2 \\
    &- \left(K - 1\right)\sum_{k = 1}^K\left\|\left( I - C_{k} \circ \mathcal{L}_{h}^a\right)\left(u_h - u^{(t)}_h\right)\right\|^2_2 \\
\end{align*}
\begin{align*}
    &= \sum_{j = 1}^K\left(\sum_{k = 1}^K\left\|\left( I - C_{k} \circ \mathcal{L}_{h}^a\right)\left(u_h - u^{(t)}_h\right)\right\|^2_2  - \left\|\left( I - C_{j} \circ \mathcal{L}_{h}^a\right)\left(u_h - u^{(t)}_h\right)\right\|^2_2\right)\1_{r(a_h, f_h, u^{(t)}_h) \neq j} \\
    &- \left(K - 2\right)\sum_{k = 1}^K\left\|\left( I - C_{k} \circ \mathcal{L}_{h}^a\right)\left(u_h - u^{(t)}_h\right)\right\|^2_2 \\
    &=\sum_{j = 1}^K\sum_{k = 1}^K\left\|\left( I - C_{k} \circ \mathcal{L}_{h}^a\right)\left(u_h - u^{(t)}_h\right)\right\|^2_2\1_{k\neq j}\1_{r(a_h, f_h, u^{(t)}_h) \neq j} \\
    &- \left(K - 2\right)\sum_{k = 1}^K\left\|\left( I - C_{k} \circ \mathcal{L}_{h}^a\right)\left(u_h - u^{(t)}_h\right)\right\|^2_2
 \end{align*}
\end{proof}

\subsection{Proof of \Cref{th:upperbound}} \label{sec:proofupperbound}

\begin{proposition} \label{th:upperbound}
    For any router $r$ defined by $r(a, f, u^{(t)}) = \text{argmax}_{j \in [K]} g_j(a, f, u^{(t)})$, any $a_h \in \mathcal{A}$, $f_h \in \mathcal{F}$, and $u_h^{(t)}, u_h \in \mathcal{U}$, the routing loss $l_{\text{route}}$ satisfies:
    \begin{equation*}
        \log(2)l_{\text{route}}\left(r, a_h, f_h, u^{(t)}_h, u_h\right)\leq \Psi(\mathbf{g}, a_h, f_h, u_h^{(t)}, u_h) 
    \end{equation*}
\end{proposition}

\begin{proof}

By \Cref{th:routingequivalence}, we know that

\begin{align*}
    \log(2)l_{\text{route}}\left(r, a_h, f_h, u^{(t)}_h, u_h\right) = &\log(2)\sum_{j = 1}^K \sum_{k = 1}^K\left\|\left( I - C_{k} \circ \mathcal{L}_{h}^a\right)\left(u_h - u^{(t)}_h\right)\right\|^2_2\1_{k\neq j}\1_{r(a_h, f_h, u^{(t)}_h) \neq j} \\
     &-\log(2)\left(K - 2\right)\sum_{j = 1}^K\left\|\left( I - C_{j} \circ \mathcal{L}_{h}^a\right)\left(u_h - u^{(t)}_h\right)\right\|^2_2 \\
     &\leq \log(2)\sum_{j = 1}^K \sum_{k = 1}^K\left\|\left( I - C_{k} \circ \mathcal{L}_{h}^a\right)\left(u_h - u^{(t)}_h\right)\right\|^2_2\1_{k\neq j}\1_{r(a_h, f_h, u^{(t)}_h) \neq j} \\
     &\overset{(a)}{\leq} -\sum_{j = 1}^K \sum_{k = 1}^K\left\|\left( I - C_{k} \circ \mathcal{L}_{h}^a\right)\left(u_h - u^{(t)}_h\right)\right\|^2_2\1_{k\neq j}\log\left(\frac{\exp\left(g_j(a, f, u^{(t)})\right)}{\sum_{k = 1}^K\exp\left(g_k(a, f, u^{(t)})\right)}\right) \\
     &= \Psi(\mathbf{g}, a_h, f_h, u_h^{(t)}, u_h)
\end{align*}
(a) if $r(a_h, f_h, u^{(t)}_h) \neq j$, $\frac{\exp\left(g_j(a, f, u^{(t)})\right)}{\sum_{k = 1}^K\exp\left(g_k(a, f, u^{(t)})\right)} < 0.5$ which implies that $-\log\left(\frac{\exp\left(g_j(a, f, u^{(t)})\right)}{\sum_{k = 1}^K\exp\left(g_k(a, f, u^{(t)})\right)}\right)  \geq \log(2)\1_{r(a_h, f_h, u^{(t)}_h) \neq j}$

\end{proof}

\subsection{Proof of \Cref{th:consistency}} \label{sec:proofconsistency}

\begin{customtheorem}{\ref{th:consistency}}
    Let $\Tilde{c}_j(a_h, u^{(t)}_h, u_h) < \Bar{E} < \infty$  for all $j \in [K]$. If there exists $j\in [K]$ such that $\Tilde{c}_j(a_h, u^{(t)}_h, u_h) > E_{min} > 0$, then, for any collection of solvers $\{C_j\}_{j= 1}^K$ and linear discrete operator $\mathcal{L}_h^a$, $\Psi$ is Bayes consistent surrogate for $l_{\text{route}}$.
\end{customtheorem}

\begin{proof}

For a given $a_h, f_h$, let $u_h$ be $\mathcal{G}_h\left(a_h, f_h\right)$ where $\mathcal{G}_h$ denotes the solution operator acting on the grid $G_h$. Furthermore, let's consider routers of the form
\begin{equation*}
    r(a, f, u^{(t)}) = \text{argmax}_{j \in [K]} g_j(a, f, u^{(t)})
\end{equation*}

For a given $a_h, f_h, u_h^{(t)} \in \mathcal{A} \times \mathcal{F} \times \mathcal{U}$, let the optimal loss under $l_{\text{route}}$ be $l_{\text{route}}^*\left(a_h, f_h, u^{(t)}_h\right) = \inf_{\Tilde{r}} l_{\text{route}}\left(\tilde{r}, a_h, f_h, u^{(t)}_h, \mathcal{G}_h\left(a_h, f_h\right)\right)$. Similarly, let the optimal loss under $\Psi$ be $\Psi^*\left(a_h, f_h, u^{(t)}_h\right) = \inf_{\Tilde{\mathbf{g}}} \Psi\left(\tilde{\mathbf{g}}, a_h, f_h, u^{(t)}_h, \mathcal{G}_h\left(a_h, f_h\right)\right)$. Let $B_j(a_h, f_h, u^{(t)}_h) = \sum_{k = 1}^K\left\|\left( I - C_{k} \circ \mathcal{L}_{h}^a\right)\left(\mathcal{G}_h\left(a_h, f_h\right) - u^{(t)}_h\right)\right\|^2_2\1_{k\neq j}$.

\begin{align*}
    &l_{\text{route}}\left(r, a_h, f_h, u^{(t)}_h, \mathcal{G}_h\left(a_h, f_h\right)\right) - l_{\text{route}}^*\left(a_h, f_h, u^{(t)}_h\right) \\
    &\overset{(a)}{=} \sum_{j = 1}^K\sum_{k = 1}^K\left\|\left( I - C_{k} \circ \mathcal{L}_{h}^a\right)\left(u_h - u^{(t)}_h\right)\right\|^2_2\1_{k\neq j}\1_{r(a_h, f_h, u^{(t)}_h) \neq j}  -\left(K - 2\right)\sum_{j = 1}^K\left\|\left( I - C_{j} \circ \mathcal{L}_{h}^a\right)\left(u_h - u^{(t)}_h\right)\right\|^2_2 \\
    &\quad- \inf_{\Tilde{r}} \sum_{j = 1}^K\sum_{k = 1}^K\left\|\left( I - C_{k} \circ \mathcal{L}_{h}^a\right)\left(u_h - u^{(t)}_h\right)\right\|^2_2\1_{k\neq j}\1_{\tilde{r}(a_h, f_h, u^{(t)}_h) \neq j}  +\left(K - 2\right)\sum_{j = 1}^K\left\|\left( I - C_{j} \circ \mathcal{L}_{h}^a\right)\left(u_h - u^{(t)}_h\right)\right\|^2_2 \\
    &= \sum_{j = 1}^K B_j(a_h, f_h, u^{(t)}_h)\1_{r(a_h, f_h, u^{(t)}_h) \neq j} - \inf_{\Tilde{r}} \sum_{j = 1}^KB_j(a_h, f_h, u^{(t)}_h)\1_{\tilde{r}(a_h, f_h, u^{(t)}_h) \neq j}  \\
    &= \sum_{k = 1}^K B_k(a_h, f_h, u^{(t)}_h)\left(\sum_{j = 1}^K \frac{B_j(a_h, f_h, u^{(t)}_h)}{\sum_{k = 1}^K B_k(a_h, f_h, u^{(t)}_h)}\1_{r(a_h, f_h, u^{(t)}_h) \neq j} - \inf_{\Tilde{r}} \sum_{j = 1}^K \frac{B_j(a_h, f_h, u^{(t)}_h)}{ \sum_{k = 1}^K B_k(a_h, f_h, u^{(t)}_h)}\1_{\tilde{r}(a_h, f_h, u^{(t)}_h) \neq j}\right)
\end{align*}
(a) by \Cref{th:routingequivalence}

Let $\mathcal{X} = \mathcal{A} \times \mathcal{F} \times \mathcal{U}$ and $\mathcal{Y} = [K]$. Let $\mathcal{P}_{\mathcal{X}}$ denote the degenerate distribution suported at the point $(a_h, f_h, u_h^{(t)})$. We define the conditional distribtion - $P(Y = j \mid X = (a_h, f_h, u_h^{(t)})) = \frac{B_j(a_h, f_h, u^{(t)}_h)}{ \sum_{k = 1}^K B_k(a_h, f_h, u^{(t)}_h)}$ for $j \in [K]$. The risk and optimal risk of $0-1$ loss under this distribution can be written as:
\begin{align*}
    \mathcal{R}_{0-1}(r) &= \sum_{j = 1}^K \frac{B_j(a_h, f_h, u^{(t)}_h)}{\sum_{k = 1}^K B_k(a_h, f_h, u^{(t)}_h)}\1_{r(a_h, f_h, u^{(t)}_h) \neq j} \\
    \mathcal{R}_{0-1}^* &= \inf_r \sum_{j = 1}^K \frac{B_j(a_h, f_h, u^{(t)}_h)}{\sum_{k = 1}^K B_k(a_h, f_h, u^{(t)}_h)}\1_{r(a_h, f_h, u^{(t)}_h) \neq j} 
\end{align*}

If $r(a_h, f_h, u^{(t)}_h) = \argmax_{j \in [k]} g_j(a_h, f_h, u^{(t)}_h)$ for all $x \in \mathcal{X}$, then the he risk and optimal risk of cross entropy loss ($l_{ce}(\mathbf{g}, x, y)-\log\left(\frac{\exp\left(g_y(x)\right)}{\sum_{k = 1}^K\exp\left(g_k(x)\right)}\right)$) under this distribution can be written as:

\begin{align*}
    \mathcal{R}_{ce}(\mathbf{g}) &= -\sum_{j = 1}^K \frac{B_j(a_h, f_h, u^{(t)}_h)}{\sum_{k = 1}^K B_k(a_h, f_h, u^{(t)}_h)}\log\left(\frac{\exp\left(g_j(a_h, f_h, u^{(t)}_h)\right)}{\sum_{k = 1}^K\exp\left(g_k(a_h, f_h, u^{(t)}_h)\right)}\right) \\
    \mathcal{R}_{ce}^* &= \inf_{\mathbf{g}} -\sum_{j = 1}^K \frac{B_j(a_h, f_h, u^{(t)}_h)}{\sum_{k = 1}^K B_k(a_h, f_h, u^{(t)}_h)}\log\left(\frac{\exp\left(g_j(a_h, f_h, u^{(t)}_h)\right)}{\sum_{k = 1}^K\exp\left(g_k(a_h, f_h, u^{(t)}_h)\right)}\right) 
\end{align*}

From Theorem 3.1 of \citep{mao2023cross}, $\mathcal{R}_{0-1}(r) - \mathcal{R}_{0-1}^* \leq \Gamma^{-1}\left(\mathcal{R}_{ce}(\mathbf{g}) - \mathcal{R}_{ce}^*\right)$ if $r(a_h, f_h, u^{(t)}_h) = \argmax_{j \in [k]} g_j(a_h, f_h, u^{(t)}_h)$ where $\Gamma(z) = \frac{1 + z}{2}\log(1 + z) + \frac{1 - z}{2}\log(1 - z)$. Then,

\begin{align*}
    &l_{\text{route}}\left(r, a_h, f_h, u^{(t)}_h, \mathcal{G}_h\left(a_h, f_h\right)\right) - l_{\text{route}}^*\left(a_h, f_h, u^{(t)}_h\right) \\
    &= \sum_{k = 1}^K B_k(a_h, f_h, u^{(t)}_h)\left(\sum_{j = 1}^K \frac{B_j(a_h, f_h, u^{(t)}_h)}{\sum_{k = 1}^K B_k(a_h, f_h, u^{(t)}_h)}\1_{r(a_h, f_h, u^{(t)}_h) \neq j} - \inf_{\Tilde{r}} \sum_{j = 1}^K \frac{B_j(a_h, f_h, u^{(t)}_h)}{ \sum_{k = 1}^K B_k(a_h, f_h, u^{(t)}_h)}\1_{\tilde{r}(a_h, f_h, u^{(t)}_h) \neq j}\right) \\
    &\leq \sum_{k = 1}^K B_k(a_h, f_h, u^{(t)}_h)\Gamma^{-1}\left(-\sum_{j = 1}^K \frac{B_j(a_h, f_h, u^{(t)}_h)}{\sum_{k = 1}^K B_k(a_h, f_h, u^{(t)}_h)}\log\left(\frac{\exp\left(g_j(a_h, f_h, u^{(t)}_h)\right)}{\sum_{k = 1}^K\exp\left(g_k(a_h, f_h, u^{(t)}_h)\right)}\right) \right.\\&\quad\quad\quad\quad\quad\quad\quad\quad\quad\quad\quad\quad- \left.\inf_{\mathbf{g}} -\sum_{j = 1}^K \frac{B_j(a_h, f_h, u^{(t)}_h)}{\sum_{k = 1}^K B_k(a_h, f_h, u^{(t)}_h)}\log\left(\frac{\exp\left(g_j(a_h, f_h, u^{(t)}_h)\right)}{\sum_{k = 1}^K\exp\left(g_k(a_h, f_h, u^{(t)}_h)\right)}\right) \right) \\
    &\overset{(a)}{\leq} \Bar{E}K\left(K- 1\right)\Gamma^{-1}\left(-\sum_{j = 1}^K \frac{B_j(a_h, f_h, u^{(t)}_h)}{\sum_{k = 1}^K B_k(a_h, f_h, u^{(t)}_h)}\log\left(\frac{\exp\left(g_j(a_h, f_h, u^{(t)}_h)\right)}{\sum_{k = 1}^K\exp\left(g_k(a_h, f_h, u^{(t)}_h)\right)}\right) \right.\\&\quad\quad\quad\quad\quad\quad\quad\quad\quad\quad\quad\quad- \left.\inf_{\mathbf{g}} -\sum_{j = 1}^K \frac{B_j(a_h, f_h, u^{(t)}_h)}{\sum_{k = 1}^K B_k(a_h, f_h, u^{(t)}_h)}\log\left(\frac{\exp\left(g_j(a_h, f_h, u^{(t)}_h)\right)}{\sum_{k = 1}^K\exp\left(g_k(a_h, f_h, u^{(t)}_h)\right)}\right) \right) \\
    &\overset{(b)}{\leq} \Bar{E}K\left(K- 1\right)\Gamma^{-1}\left(-\sum_{j = 1}^K \frac{B_j(a_h, f_h, u^{(t)}_h)}{(K - 1)E_{min}}\log\left(\frac{\exp\left(g_j(a_h, f_h, u^{(t)}_h)\right)}{\sum_{k = 1}^K\exp\left(g_k(a_h, f_h, u^{(t)}_h)\right)}\right) \right.\\&\quad\quad\quad\quad\quad\quad\quad\quad\quad\quad\quad\quad- \left.\inf_{\mathbf{g}} -\sum_{j = 1}^K \frac{B_j(a_h, f_h, u^{(t)}_h)}{(K - 1)E_{min}}\log\left(\frac{\exp\left(g_j(a_h, f_h, u^{(t)}_h)\right)}{\sum_{k = 1}^K\exp\left(g_k(a_h, f_h, u^{(t)}_h)\right)}\right) \right) \\
    &= \Bar{E}K\left(K- 1\right)\Gamma^{-1}\left(\frac{\Psi\left(\mathbf{g}, a_h, f_h, u^{(t)}_h, \mathcal{G}_h\left(a_h, f_h\right)\right)- \Psi^*\left(a_h, f_h, u^{(t)}_h\right)}{(K - 1)E_{min}} \right)
\end{align*}
(a) since $\left\|\left( I - C_{j} \circ \mathcal{L}_{h}^a\right)\left(e^{(t)}_h\right)\right\|^2_2 < \Bar{E}$ for all $j \in [K]$, (b) since $\Gamma^{-1}$ is non-decreasing and $\exists j \in [K]$ such that $\left\|\left( I - C_{j} \circ \mathcal{L}_{h}^a\right)\left(e^{(t)}_h\right)\right\|^2_2 > E_{min}$

Finally,

\begin{align*}
    &\lim_{n \to \infty }\mathcal{R}_{\text{route}}\left(r_n\right) - \mathcal{R}_{\text{route}}^* \\
    &\overset{(a)}{=} \lim_{n \to \infty } \mathbb{E}_{a_h, f_h \sim \mathcal{P}_{\mathcal{A} \times \mathcal{F}}}\left[l_{\text{route}}\left(r_n, a_h, f_h, u^{(t)}_h, \mathcal{G}_h\left(a_h, f_h\right)\right) - l_{\text{route}}^*\left(a_h, f_h, u^{(t)}_h\right)\right] \\
    &\leq \lim_{n \to \infty } \mathbb{E}_{a_h, f_h \sim \mathcal{P}_{\mathcal{A} \times \mathcal{F}}}\left[\Bar{E}K\left(K- 1\right)\Gamma^{-1}\left(\frac{\Psi\left(\tilde{\mathbf{g}}, a_h, f_h, u^{(t)}_h, \mathcal{G}_h\left(a_h, f_h\right)\right)- \Psi^*\left(a_h, f_h, u^{(t)}_h\right)}{(K - 1)E_{min}} \right)\right] \\
    &\overset{(b)}{\leq} \lim_{n \to \infty } \Bar{E}K\left(K- 1\right)\Gamma^{-1}\left(\frac{\mathbb{E}_{a_h, f_h \sim \mathcal{P}_{\mathcal{A} \times \mathcal{F}}}\left[\Psi\left(\mathbf{g}_n, a_h, f_h, u^{(t)}_h, \mathcal{G}_h\left(a_h, f_h\right)\right)- \Psi^*\left(a_h, f_h, u^{(t)}_h\right) \right]}{(K - 1)E_{min}}\right) \\
\end{align*}
\begin{align*}
    &= \lim_{n \to \infty } \Bar{E}K\left(K- 1\right)\Gamma^{-1}\left(\frac{\mathcal{R}_{\Psi}\left(\mathbf{g}_n\right) - \mathcal{R}_{\Psi}^*}{(K - 1)E_{min}}\right) \\
    &\overset{(c)}{=} \Bar{E}K\left(K- 1\right)\Gamma^{-1}\left(\frac{\lim_{n \to \infty } \mathcal{R}_{\Psi}\left(\mathbf{g}_n\right) - \mathcal{R}_{\Psi}^*}{(K - 1)E_{min}}\right) \\
    &= \Bar{E}K\left(K- 1\right)\Gamma^{-1}\left(0\right) \\
    &\overset{(d)}{=} 0
\end{align*}

(a) $\mathcal{R}_{\text{route}}^*  = \mathbb{E}_{a_h, f_h \sim \mathcal{P}_{\mathcal{A} \times \mathcal{F}}}\left[l_{\text{route}}^*\left(a_h, f_h, u^{(t)}_h\right)\right]$ since the infimum is taken over all measurable functions, (b) by Jensen's inequality since $\Gamma^{-1}$ is concave, (c) by continuity of $\Gamma^{-1}$ at $0$, (d) $\Gamma^{-1}(0) = 0$


\end{proof}

\section{Cost-aware extension} \label{sec:cost-aware}

The horizon in \Cref{th:suboptgreedy} counts routing decisions.  This formulation is appropriate for deriving wall-clock guarantees when all actions have the same cost, because the number of routing decisions are proportional to wall-clock time. When action costs differ--for example, a neural correction is often more expensive than one Jacobi iteration--a routing decision count guarantee does not directly imply a wall-clock guarantee. In experiments, we mitigate this discrepancy by grouping multiple Jacobi iterations into a macro-action whose measured cost approximately matches that of one neural correction. More generally, heterogeneous costs can be handled through a cost-aware formulation whose objective is to identify the best solver sequence with total cost at most $B$. 

Formally, suppose each action $j \leq K$ has a deterministic cost $c_j>0$.  In keeping with the sequence convention used throughout the paper, for $S=(S_1,\ldots,S_{|S|})\in\Omega^*$, we define
\begin{equation*}
    c(S):=\sum_{t=1}^{|S|}c_{S_t},
    \qquad
    \mathcal{F}(B):=\{S\in\Omega^*:c(S)\leq B\}.
\end{equation*}
$S_1$ is still applied first and concatenation $S\oplus j$ appends action $j$ to the current solver trajectory.  For a total budget $B>0$, let
\begin{equation*}
    O_B\in\argmin_{S\in\mathcal{F}(B)}g(S)
\end{equation*}
be a budget-optimal sequence. We denote $\Delta_j(S):=g(S)-g(S\oplus j)$ for the one-step decrease.  The natural cost-aware greedy rule is the marginal-decrease per unit cost rule or 
\begin{equation}\label{eq:cost-aware-greedy}
    j_t\in\argmax_{j\leq K}
       \frac{\Delta_j(S^t)}{c_j},
    \qquad S^{t+1}=S^t\oplus j_t,
\end{equation}
This greedy algorithm continues as long as $c(S^t)<B$.  The action selected at such a prefix is completed even if it makes $c(S^{t+1})$ exceed $B$, after which the algorithm stops. Under prefix monotonicity, the numerator is nonnegative. Notice that minimizing the next error $g(S\oplus j)$ and maximizing \Cref{eq:cost-aware-greedy} are equivalent only when the actions are equivalent in cost. This property is leveraged in experiments.

This extension to heterogeneous costs requires minor changes to the main paper's weak-supermodularity definition. Specifically, we say that $g$ is \emph{budget-range weakly supermodular with respect to $O_B$} with ratio $\alpha(O_B)$ if \Cref{eqn:alpha_supermod} holds with $S'=O_B$ for every prefix $S$ satisfying $c(S)<B$.  The original length-based condition, which requires the prefixes to satisfy $|S| \leq |O_B|$, is insufficient here: within the same budget, a greedy polict may select more inexpensive actions than the number of actions in $O_B$, producing sequences longer than $O_B$.

With these details established, we can now characterize the suboptimality of a greedy algorithm for optimizing cost-aware solver selection in \Cref{th:cost-aware-greedy}.

\begin{theorem}[Cost-aware greedy with one-action budget slack]
\label{th:cost-aware-greedy}
Let $g:\Omega^*\to\mathbb{R}$ be prefix and postfix monotonically non-increasing and budget-range weakly supermodular with respect to $O_B$, with ratio $\alpha(O_B)$.  Suppose the action costs are fixed and $c_{\max}:=\max_{j\in\Omega}c_j\leq B$.  Let $S_{B+}$ be the sequence returned by \Cref{eq:cost-aware-greedy} under the stopping rule above.  Then $B\leq c(S_{B+})<B+c_{\max}$ and
\begin{equation}\label{eq:cost-aware-bicriteria}
 g(\emptyset)-g(S_{B+})
 \geq \left(1-e^{-1/\alpha(O_B)}\right)
 \bigl(g(\emptyset)-g(O_B)\bigr).
\end{equation}
\end{theorem}

\begin{proof}
Let $c_t=c_{j_t}$ denote the cost of the greedy action selected at step $t$.  For
every greedy prefix $S^t$ with $c(S^t)<B$, we have
\begin{align*}
    g(S^t)-g(O_B)
    &\overset{(a)}{\leq}g(S^t)-g(S^t\oplus O_B)\\
    &\overset{(b)}{\leq}
    \alpha(O_B)\sum_{i=1}^{|O_B|}
       \left(g(S^t)-g(S^t\oplus (O_B)_i)\right)\\
    &=\alpha(O_B)\sum_{i=1}^{|O_B|}c_{(O_B)_i}
       \frac{g(S^t)-g(S^t\oplus (O_B)_i)}{c_{(O_B)_i}}\\
    &\leq \alpha(O_B)c(O_B)
       \max_{j\in\Omega}\frac{g(S^t)-g(S^t\oplus j)}{c_j}\\
    &\overset{(c)}{\leq}\alpha(O_B)B
       \frac{g(S^t)-g(S^{t+1})}{c_t}.
\end{align*}
(a) follows from postfix monotonicity, (b) from budget-range weak
supermodularity, and (c) from $c(O_B)\leq B$ and the density-greedy rule in
\Cref{eq:cost-aware-greedy}.

After rearranging the inequality, we get
\begin{align*}
    g(S^{t+1})-g(O_B)
    &\leq
    \left(1-\frac{c_t}{\alpha(O_B)B}\right)
    \left(g(S^t)-g(O_B)\right).
\end{align*}
Let $m$ be the first index for which $c(S^m)\geq B$.  Recursively applying
the preceding inequality gives
\begin{align*}
    g(S^m)-g(O_B)
    &\leq
    \prod_{t=0}^{m-1}
    \left(1-\frac{c_t}{\alpha(O_B)B}\right)
    \left(g(\emptyset)-g(O_B)\right)\\
    &\leq
    \exp\!\left(-\frac{\sum_{t=0}^{m-1}c_t}
    {\alpha(O_B)B}\right)
    \left(g(\emptyset)-g(O_B)\right)\\
    &\leq e^{-1/\alpha(O_B)}
    \left(g(\emptyset)-g(O_B)\right),
\end{align*}
where the second inequality uses $1-x\leq e^{-x}$ and the last uses
$c(S^m)=\sum_{t=0}^{m-1}c_t\geq B$.  Rearranging yields
\Cref{eq:cost-aware-bicriteria}.  Finally, $c(S^{m-1})<B$ and the last action
costs at most $c_{\max}$, so
$B\leq c(S^m)<B+c_{\max}$.
\end{proof}

If all actions cost $c$ and $B=Tc$, the analysis above is equivalent to the discrete action-count recurrence of \Cref{th:suboptgreedy}, which yields $1-(1-1/[\alpha(O_B)T])^T$.

\section{Experimental Setup}
\label{sec:training details}

\subsection{Data and reference solutions}
All training, development, and test splits use disjoint random seeds. Test forcings follow the zero-mean hierarchical Gaussian random-field distribution described in \Cref{sec:experiments}. For a frequency $k\neq0$, its Fourier coefficient has variance $\alpha(4\pi^2\lVert k\rVert_2^2+\beta)^{-\gamma}$, where $\alpha\sim\operatorname{LogUniform}(10^{-2},10^2)$, $\beta\sim\operatorname{LogUniform}(10^{-1},10^3)$, and $\gamma$ is uniform on $\{0.5,1,1.5,2,2.5,3,4\}$. Hermitian symmetry produces a real field and the zero mode is removed. Neural-corrector training uses an equal mixture of four input distributions. ``GRF'' denotes the forcing distribution above. ``Rough'' draws a GRF error $e$ and uses $\mathcal{L}_h^a e$ as the residual, which emphasizes the high-frequency content produced during an iterative solve. ``White'' draws independent standard-normal grid values before demeaning and spectral restriction, giving an approximately flat in-band spectrum. ``Tilted'' draws complex Gaussian Fourier coefficients and scales them by
\begin{equation*}
 \left(a_x^2 k_x^2+a_x^{-2}k_y^2+1\right)^{-q/2},
 \qquad q\sim\operatorname{Uniform}(-1,3),\quad
 \log a_x\sim\operatorname{Uniform}(-1,1),
\end{equation*}
thereby randomizing both spectral slope and directional anisotropy. These additional distributions expose the correctors to residuals beyond the forcing law used for final evaluation, allowing the neural correctors to behave more like an iterative solver. Reference solutions are computed from the same discrete PDE operator used in the rollout.

\subsection{Neural corrections}
\label{sec:global-corrector-training}
The global correction is a linear Fourier-trunk DeepONet. Its trunk is a fixed orthonormal Fourier basis, and its branch is fit by least squares to map a restricted, normalized residual to the coefficients of its exact error correction. At inference time, the predicted correction is multiplied by the input residual norm, making the map scale equivariant and explicitly zero preserving. The global model in the growing-ensemble experiments, denoted $\mathrm{s32}$, has a $32\times32$ sensor grid and is fit using $16{,}384$ training and $1{,}024$ validation pairs. We also train a stronger $64\times64$-sensor model using $64{,}000$ training and $2{,}000$ validation pairs; this model is used in the paired-solver experiments and as an external pairwise control for the ensemble study. For anisotropic diffusion, the sensor layout retains full resolution in the weak-diffusion direction. A separate model is fit for each PDE. Because the branch is obtained by a direct least-squares solve rather than by gradient descent, these models have no optimizer, learning rate, training epochs, or SGD batch size. Training pairs are generated in chunks of $1{,}024$ only to control memory use; this chunk size is not an optimization hyperparameter.

The local correction is a learned periodic $3\times3$ convolution. It is trained on the residual left after the $\mathrm{s32}$ global correction. We generate $512$ such residuals from the same four-component mixture (GRF, rough, white, and tilted spectra), solve the discrete PDE for their exact remaining errors, normalize each pair by the residual norm, and fit the kernel by ridge least squares. An independent set of $32$ GRF residuals is used for validation. Specifically, the nine kernel coefficients are obtained by a direct ridge least-squares solve with coefficient $10^{-8}$, not by gradient descent.

\subsection{Router training}
\label{sec:router-training-details}
The learned routers are two layer MLPs with ReLU activations. The input features for the routers include the log relative residual, its change over the preceding macro-action and an exponential moving average of that change, the decision index, the previous-action indicator, per-action counts, the number of neural calls, and the number of operations since the last neural call. 

At each decision point, the training oracle applies every available macro-action to a cloned state and labels the example using the true next-error values in \Cref{eq:routersurrogate}. Round zero uses oracle rollouts with $0.1$ probability of random exploration. Subsequent DAgger rounds roll out the current router, relabel the states it visits with the oracle, and refit on the accumulated dataset. The paired-solver study uses $128$ instances per round and two DAgger rounds. The growing-ensemble study uses $256$ instances per round and four DAgger rounds; within that study, every constituent pair, historical control, and growing ensemble receives this same training strength.


\begin{table}[t]
\centering
\small
\caption{Router hyperparameters. A DAgger round denotes an additional
router-policy rollout and oracle relabeling pass after the initial exploratory
oracle round.}
\label{tab:router-hyperparameters}
\begin{tabular}{lrr}
\hline
Hyperparameter & Paired study & Growing ensembles\\
\hline
Feature variants & Scalar & Scalar; pooled residual\\
Hidden width & $64$ & $64$ scalar; $32$ pooled\\
Instances per round & $128$ & $256$\\
DAgger rounds & $2$ & $4$\\
Epochs per fit & $150$ & $150$\\
Optimizer & Adam & Adam\\
Initial learning rate & $2\times10^{-3}$ & $2\times10^{-3}$\\
Schedule & Cosine decay & Cosine decay\\
Batch size & $2{,}048$ & $2{,}048$\\
Training error cutoff & $10^{-9}$ & $10^{-9}$\\
Maximum query epochs & $3{,}000$ & $512$\\
\hline
\end{tabular}
\end{table}

\subsection{Cost matching for paired and ensemble experiments}
\label{sec:cost-matching-setup}
We measure the complete cost $c_j$ of every primitive action on a common single-thread CPU, including the update and subsequent residual refresh. In a paired experiment, the reference unit is one call to that pair's global neural
correction. A classical decision executes
\begin{equation*}
    m_j=\max\{1,\operatorname{round}(c_{\mathrm{NO}}/c_j)\}
\end{equation*}
consecutive classical iterations, whereas a neural decision executes one correction. The router therefore compares approximately equal-cost macro-actions rather than one expensive neural call with one inexpensive sweep. Only the classical updates are grouped by this rule. Each neural correction is one indivisible action: in particular, the cheaper local convolution is not repeated to match the $\mathrm{s32}$ cost.

For each growing-ensemble PDE, all action sets share one cost table and one set of frozen macro-action sizes. The reference unit is the cost of the $\mathrm{s32}$ global correction; every neural action is called once, and every classical action executes the rounded number of iterations implied by the shared unit. The same action consequently has identical semantics in every pair, historical control, and ensemble size. Costs and macro sizes are frozen before final cases are evaluated. Selected-path times charge only executed actions and router overhead.

\Cref{tab:cost-calibration} reports the resulting $128 \times 128$ grid calibration. The primitive and effective macro costs are shown separately. Thus, for example, the Poisson weighted-Jacobi macro costs $6(0.147)=0.884$ ms, close to the $0.897$ ms cost of one $\mathrm{s32}$ call. Multigrid remains substantially more expensive than the reference unit because a V-cycle is indivisible; the cost-adjusted score accounts for this mismatch. The local correction is deliberately cheaper than the unit but is also an indivisible one-call action. The full global model is not a member of the growing ensemble and is reported separately as the neural correction used by the paired controls.

\begin{table}[t]
\centering
\small
\caption{Measured complete-action costs for the $N=128$ experiments. Costs include the update and subsequent residual refresh. For growing-ensemble actions, $m$ is the number of primitive applications in one macro-action and the last column is $m c_j$. The $\mathrm{s32}$ cost is the matching reference. The full global model is used only by paired controls, so no growing-ensemble macro count is assigned to it.}
\label{tab:cost-calibration}
\begin{tabular}{llrrr}
\hline
Equation & Action & Primitive cost (ms) & $m$ & Macro cost (ms)\\
\hline
Poisson & WJacobi(.67) & 0.147 & 6 & 0.884\\
        & GS            & 0.443 & 2 & 0.886\\
        & SOR(1.5)      & 0.448 & 2 & 0.896\\
        & MG            & 3.614 & 1 & 3.614\\
        & Global $\mathrm{s32}$ & 0.897 & 1 & 0.897\\
        & Local $3\times3$      & 0.370 & 1 & 0.370\\
        & Global full (paired control) & 5.115 & -- & --\\
\hline
ConvDiff & WJacobi(.67) & 0.208 & 4 & 0.833\\
         & GS            & 0.502 & 2 & 1.003\\
         & SOR(1.5)      & 0.502 & 2 & 1.003\\
         & MG            & 3.956 & 1 & 3.956\\
         & Global $\mathrm{s32}$ & 0.915 & 1 & 0.915\\
         & Local $3\times3$      & 0.426 & 1 & 0.426\\
         & Global full (paired control) & 5.254 & -- & --\\
\hline
AnisoDiff & WJacobi(.67) & 0.164 & 5 & 0.820\\
          & GS            & 0.464 & 2 & 0.928\\
          & SOR(1.5)      & 0.465 & 2 & 0.930\\
          & MG            & 3.637 & 1 & 3.637\\
          & Global $\mathrm{s32}$ & 0.895 & 1 & 0.895\\
          & Local $3\times3$      & 0.386 & 1 & 0.386\\
          & Global full (paired control) & 5.745 & -- & --\\
\hline
\end{tabular}
\end{table}

\subsection{Action ordering for the growing ensembles}
\label{sec:ensemble-action-order}
The growing-ensemble experiment requires a nested sequence of action sets so that increasing $K$ adds one action without changing the actions already available. We determine this order on $16$ independent samples before router training and final evaluation. Each action is first evaluated as a singleton. Their individual strength is ranked by number of failures, mean successful modeled time, and median terminal relative error. We initialize the pairwise set with the weakest numerical action and the weakest neural action, then append the remaining actions from weaker to stronger according to this ranking. This helps answer the question of whether progressively adding stronger and more complementary actions improves the learned policy; it is not an order chosen to make final performance monotone. All prefixes, including negative ones, are retained in the final evaluation.

\Cref{tab:ensemble-action-order} gives the resulting frozen orders. For an order $(a_1,\ldots,a_6)$, the size-$K$ ensemble is $\mathcal{A}_K=\{a_1,\ldots,a_K\}$. Poisson and convection--diffusion yield the same order, whereas anisotropic diffusion promotes GS and SOR ahead of the global correction because their ranking differs on that PDE.

\begin{table}[t]
\centering
\small
\caption{Frozen action order for the $N=128$ growing-ensemble experiments. Each size-$K$ action set is the length-$K$ prefix of the corresponding row.}
\label{tab:ensemble-action-order}
\resizebox{\textwidth}{!}{%
\begin{tabular}{lrrrrrr}
\hline
Equation & $a_1$ & $a_2$ & $a_3$ & $a_4$ & $a_5$ & $a_6$\\
\hline
Poisson & WJacobi(.67) & Local $3\times3$ & Global $\mathrm{s32}$ & GS & SOR(1.5) & MG\\
ConvDiff & WJacobi(.67) & Local $3\times3$ & Global $\mathrm{s32}$ & GS & SOR(1.5) & MG\\
AnisoDiff & WJacobi(.67) & Local $3\times3$ & GS & SOR(1.5) & Global $\mathrm{s32}$ & MG\\
\hline
\end{tabular}}
\end{table}

\subsection{Selection of the HINTS cadence}
\label{sec:hints-selection}
HINTS replaces a classical update with a neural correction at a fixed period $\tau$. Its performance is sensitive to this period: small $\tau$ can waste expensive neural calls or reintroduce error, whereas large $\tau$ can delay a useful neural correction. For each PDE and classical solver in the paired study, we evaluate $\tau\in\{5,10,25,50\}$ using the same neural model, stopping rule, action costs, and test instances as the learned router. \Cref{tab:errorcomparison} reports the period with the lowest mean selected-action time to the target tolerance, retaining failures. \Cref{tab:hints-cadence-results} reports the selected cadence and its resulting mean time.

\begin{table}[t]
\centering
\small
\caption{Post-selected HINTS results at relative solution error $10^{-8}$ on
the $64$ shared $N=128$ test cases. Each cell reports the selected period
$\tau$ and mean selected-action replay time in milliseconds (standard error in
parentheses). The period is the best of $\{5,10,25,50\}$ on these same cases,
so these are optimistic HINTS results rather than held-out cadence estimates.}
\label{tab:hints-cadence-results}
\resizebox{\textwidth}{!}{%
\begin{tabular}{lrrr}
\hline
Classical solver & Poisson & ConvDiff & AnisoDiff\\
\hline
Jacobi       & $50$; 785.88 (92.58) & $50$; 1190.68 (125.93) & $50$; 353.62 (60.96)\\
WJacobi(.67) & $10$; 17.34 (0.36)   & $10$; 18.23 (0.41)     & $25$; 23.16 (0.15)\\
GS           & $5$; 28.88 (0.38)    & $5$; 30.93 (0.37)      & $25$; 69.99 (1.43)\\
SymGS        & $5$; 30.60 (0.03)    & $5$; 31.64 (0.21)      & $5$; 29.02 (0.59)\\
SOR(1.5)     & $10$; 36.74 (0.50)   & $10$; 43.53 (0.58)     & $10$; 42.50 (0.62)\\
MG           & $5$; 23.44 (0.21)    & $5$; 26.37 (0.24)      & $5$; 57.12 (2.06)\\
\hline
\end{tabular}}
\end{table}

\section{Assumption Satisfaction in the Experiments}
\label{app:assumption-satisfaction}

This section evaluates how closely the experimental macro-actions satisfy the sufficient conditions used in Proposition~\ref{th:weaklyalphasupermodular} and Theorem~\ref{th:suboptgreedy}.  The theorem requires zero-preserving strict contractions in the stated norm for the weak-supermodularity argument, as well as postfix monotonicity for the greedy approximation guarantee.  The measurements below characterize the experimental maps but do not replace these uniform assumptions by a proof.

\paragraph{Zero preservation.}
Every finalized macro-action mapped the zero error to zero in the exact-source $N=128$ audit.  For the classical stationary iterations and local convolution, this also follows algebraically from linearity and the absence of an additive bias.  The global neural correction explicitly returns a correction scaled by the restricted-residual norm and therefore also maps a zero residual to zero.

\paragraph{Contraction diagnostics.}
\Cref{tab:assumption-contraction} reports macro-action norms in the same Euclidean, mean-zero state space used by the objective.  Weighted Jacobi and the local convolution are translation invariant, so their norms are computed exactly from their Fourier multipliers.  For the non-normal GS and SOR maps, we report matrix-free largest-singular-value estimates for the exact deployed macros. The values above one are direct numerical witnesses against strict contraction.  The MG entries are only lower bounds obtained from structured Fourier and seeded random probes, and therefore do not certify contraction. The global neural map is nonlinear; moreover, because its correction is band-limited, every excluded Fourier error mode is unchanged.  Its global Lipschitz constant is consequently at least one, independently of sampling.

\begin{table}[t]
\centering
\caption{Contraction diagnostics for finalized $N=128$ macro-actions in the
mean-zero Euclidean solution-error norm.  ``Exact'' denotes Fourier
diagonalization of the deployed linear map, ``SV'' a matrix-free largest
singular-value estimate, and ``LB'' a sampled lower bound rather than an upper
certificate.  Values greater than one falsify strict contraction.}
\label{tab:assumption-contraction}
\resizebox{\textwidth}{!}{%
\begin{tabular}{lcccccc}
\hline
PDE & WJ(.67) & GS & SOR(1.5) & MG & Global neural & Local convolution \\
\hline
Poisson & $0.997581$ (exact) & $1.019231$ (SV) & $1.369303$ (SV)
        & $\geq0.036064$ (LB) & $\geq1$ (exact witness) & $1.064909$ (exact) \\
ConvDiff & $0.998390$ (exact) & $1.022788$ (SV) & $1.475764$ (SV)
         & $\geq0.024636$ (LB) & $\geq1$ (exact witness) & $1.016401$ (exact) \\
AnisoDiff & $0.999936$ (exact) & $1.050241$ (SV) & $1.459752$ (SV)
          & $\geq0.853678$ (LB) & $\geq1$ (exact witness) & $35.822662$ (exact) \\
\hline
\end{tabular}}
\end{table}

The pairwise experiments additionally route to undamped Jacobi and SymGS. Their independently frozen macro sizes and exact-source $N=128$ diagnostics are shown in Table~\ref{tab:pairwise-extra-contraction}.  Undamped periodic Jacobi has exact norm one for all three PDEs, and is therefore nonexpansive but not a strict contraction.  SymGS is contractive for the Poisson and ConvDiff pairwise macros, whereas its anisotropic macro has norm greater than one.

\begin{table}[t]
\centering
\caption{Additional contraction diagnostics for actions appearing in the
finalized pairwise experiments but not in the growing action pool. Parentheses
give the frozen number of primitive iterations per routed macro.}
\label{tab:pairwise-extra-contraction}
\begin{tabular}{lrr}
\hline
PDE & Undamped Jacobi & SymGS \\
\hline
Poisson & $1.000000$ (44; exact) & $0.985629$ (8; SV) \\
ConvDiff & $1.000000$ (27; exact) & $0.989307$ (7; SV) \\
AnisoDiff & $1.000000$ (32; exact) & $1.048533$ (7; SV) \\
\hline
\end{tabular}
\end{table}

These results illustrate why spectral radii or average error ratios are not sufficient conditions for Proposition~\ref{th:weaklyalphasupermodular}. In particular, lexicographic GS and SOR are non-normal: their asymptotic spectral behavior can be stable even though their one-macro Euclidean operator norm exceeds one.  Hence the strict contraction premise is not uniformly satisfied by the complete experimental action sets.  WJ is individually contractive, although its anisotropic contraction factor is very close to one.

\paragraph{Postfix monotonicity}
Postfix monotonicity requires
\[
 h(S'\oplus S)\leq h(S)
\]
for every prepended sequence $S'$, suffix $S$, and initial state.  The latter does not follow from zero preservation alone and can fail for noncommuting actions even when the individual maps contract.  We tested every ordered pair of one-macro prepended and suffix actions on eight structured Fourier probes and 64 seeded random mean-zero probes.  This produced $36\times72=2592$ comparisons per PDE.  At least one sampled postfix violation occurred for 13 of 36 ordered pairs for Poisson, 13 of 36 for ConvDiff, and 14 of 36 for AnisoDiff.  These witnesses disprove uniform postfix monotonicity for the corresponding pairs; absence of a sampled violation would not prove the property.

The experiments nevertheless exhibit useful approximate postfix behavior on representative states.  To quantify it without conflating it with the uniform theorem assumption, define an empirical $(\mu,\delta)$ condition under a declared probe distribution by
\begin{equation}
 \Pr\!\left[h(S'\oplus S)\leq \mu h(S)\right]\geq 1-\delta.
 \label{eq:empirical-weak-postfix}
\end{equation}
If $r=\|E_S E_{S'}e\|_2/\|E_S e\|_2$, the corresponding squared-objective slack is $\mu=r^2$. \Cref{tab:empirical-postfix} summarizes all 2592 ratios per PDE.

\begin{table}[t]
\centering
\caption{Distributional postfix diagnostics over all ordered one-macro action
pairs and the declared $N=128$ probe set.  Quantiles are for the norm ratio
$r$; the slack for the squared-error objective is $\mu=r^2$.}
\label{tab:empirical-postfix}
\begin{tabular}{lrrrr}
\hline
PDE & Median $r$ & 95th percentile & 99th percentile & Fraction $r\leq1.05$ \\
\hline
Poisson & $0.4666$ & $0.9976$ & $1.0649$ & $98.96\%$ \\
ConvDiff & $0.4706$ & $0.9984$ & $1.0081$ & $99.38\%$ \\
AnisoDiff & $0.7496$ & $11.7255$ & $24.4790$ & $83.22\%$ \\
\hline
\end{tabular}
\end{table}

For ConvDiff, the empirical 99th-percentile squared-objective slack is $\mu\approx1.0081^2\approx1.0163$; for Poisson it is $\mu\approx1.0649^2\approx1.1340$.  Thus Poisson and ConvDiff are close to postfix monotone for most tested action--state combinations, despite rare structured counterexamples.  Indeed, all comparisons on the 64 random probes were non-increasing for these two PDEs; the violations arose from deliberately included structured Fourier directions.  AnisoDiff behaves differently: only $83.22\%$ of all comparisons have $r\leq1.05$, and its large upper quantiles persist even on random probes.  A close-to-one high-probability postfix description is therefore not supported for the complete AnisoDiff action set.

The audit supports two conclusions.  First, zero preservation---one of the structural premises of Proposition~\ref{th:weaklyalphasupermodular}---holds for all finalized actions. Second, although strict contraction and uniform postfix monotonicity do not hold globally, their empirical violations are mild for most tested Poisson and convection–diffusion action–state pairs; anisotropic diffusion exhibits more substantial violations.


\section{Additional Experimental Results} \label{sec:moreexperimentalresults}

We present additional visualizations and tables that further highlight the strengths of our method. 

\subsection{Error Comparisons with Paired $t$-Test Results} \label{sec:moreerrorcomparison}

For each solver family, we compare the learned greedy router against baseline methods (single-solver only and HINTS) using paired $t$-tests with a one-sided alternative hypothesis that the learned router achieves lower error or AUC. Results are reported in \Cref{tab:wallclock-paired-tests}. Across all five solver families, the resulting $p$-values indicate statistically significant improvements of the learned greedy router over both single-solver and HINTS baselines.



\begin{table}[p]
\centering
\caption{Selected-action wall-clock replay time to relative solution error $10^{-12}$, reported as mean (standard error) in milliseconds over 64 shared test instances. Brackets indicate incomplete coverage. Paired one-sided $t$-tests compare each baseline with Learned Greedy under the alternative $T_{\mathrm{Learned}}<T_{\mathrm{baseline}}$, using jointly successful instances. A dash indicates that no paired test is available. True-Greedy is an oracle and is included as a reference.}
\resizebox{\textwidth}{!}{%
\begin{tabular}{lcccccc}
\hline\hline
& \multicolumn{2}{c}{\textbf{Poisson}}
& \multicolumn{2}{c}{\textbf{ConvDiff}}
& \multicolumn{2}{c}{\textbf{Anisotropic diffusion}} \\
\cline{2-7}
\textbf{Method}
& Time (ms) & $p$-value
& Time (ms) & $p$-value
& Time (ms) & $p$-value \\
\hline

\multicolumn{7}{c}{Jacobi-related solvers} \\
\hline
Jacobi Only
& 6428.52 (4.15) & $<10^{-3}$
& 10146.76 (14.21) & $<10^{-3}$
& -- & -- \\
HINTS-Jacobi
& 2288.27 (125.46) & $<10^{-3}$
& 3278.90 (161.52) & $<10^{-3}$
& 1845.77 (126.04) & $<10^{-3}$ \\
Learned Greedy-Jacobi
& 1280.80 (64.97) & --
& 2900.52 (181.73) & --
& 1117.90 (77.21) [62/64] & -- \\
True-Greedy-Jacobi
& 1226.02 (66.29) & --
& 2120.53 (103.75) & --
& 1130.38 (75.91) & -- \\
\hline

\multicolumn{7}{c}{Weighted-Jacobi-related solvers ($\omega=0.67$)} \\
\hline
Weighted Jacobi Only
& -- & --
& -- & --
& -- & -- \\
HINTS-Weighted Jacobi
& 27.55 (0.85) & $<10^{-3}$
& 32.47 (0.96) & $<10^{-3}$
& 27.88 (0.09) & $<10^{-3}$ \\
Learned Greedy-Weighted Jacobi
& 19.95 (0.37) & --
& 21.23 (0.19) & --
& 27.47 (0.14) & -- \\
True-Greedy-Weighted Jacobi
& 18.96 (0.24) & --
& 21.10 (0.21) & --
& 28.32 (0.41) & -- \\
\hline

\multicolumn{7}{c}{GS-related solvers} \\
\hline
GS Only
& 10224.51 (6.58) & $<10^{-3}$
& 9921.16 (13.81) & $<10^{-3}$
& -- & -- \\
HINTS-GS
& 48.26 (0.31) & $<10^{-3}$
& 45.60 (0.15) & $<10^{-3}$
& 216.46 (1.58) & $<10^{-3}$ \\
Learned Greedy-GS
& 35.63 (0.71) & --
& 35.60 (0.75) & --
& 55.53 (2.28) & -- \\
True-Greedy-GS
& 34.94 (0.81) & --
& 35.77 (0.73) & --
& 53.38 (2.42) & -- \\
\hline

\multicolumn{7}{c}{SymGS-related solvers} \\
\hline
SymGS Only
& 9210.31 (6.18) & $<10^{-3}$
& 9770.15 (13.53) & $<10^{-3}$
& -- & -- \\
HINTS-SymGS
& 37.57 (0.08) & $<10^{-3}$
& 40.77 (0.36) & $<10^{-3}$
& 236.63 (2.77) & $<10^{-3}$ \\
Learned Greedy-SymGS
& 35.04 (0.68) & --
& 36.69 (0.73) & --
& 69.44 (4.28) & -- \\
True-Greedy-SymGS
& 34.21 (0.79) & --
& 36.06 (0.82) & --
& 55.12 (2.73) & -- \\
\hline

\multicolumn{7}{c}{SOR-related solvers ($\omega=1.5$)} \\
\hline
SOR Only
& 3408.10 (2.17) & $<10^{-3}$
& 2391.65 (3.32) & $<10^{-3}$
& -- & -- \\
HINTS-SOR
& 56.19 (0.60) & $<10^{-3}$
& 56.53 (0.16) & $<10^{-3}$
& 121.00 (0.83) & $<10^{-3}$ \\
Learned Greedy-SOR
& 44.60 (0.97) & --
& 45.97 (1.06) & --
& 47.98 (0.80) & -- \\
True-Greedy-SOR
& 42.68 (0.95) & --
& 46.39 (1.04) & --
& 46.58 (0.81) & -- \\
\hline

\multicolumn{7}{c}{MG-related solvers} \\
\hline
MG Only
& 36.19 (0.36) & $<10^{-3}$
& 41.51 (0.39) & $<10^{-3}$
& 871.50 (6.76) & $<10^{-3}$ \\
HINTS-MG
& 32.09 (0.32) & $<10^{-3}$
& 36.46 (0.30) & $<10^{-3}$
& 337.11 (3.57) & $<10^{-3}$ \\
Learned Greedy-MG
& 29.93 (0.46) & --
& 33.71 (0.56) & --
& 82.76 (5.05) & -- \\
True-Greedy-MG
& 29.86 (0.46) & --
& 33.35 (0.49) & --
& 78.04 (5.00) & -- \\
\hline\hline
\end{tabular}}
\label{tab:wallclock-paired-tests}
\end{table}

\subsection{Convergence Histories} \label{sec:convergencehistories}

See \Cref{fig:errorcomparison2} for convergence histories for the five solver families.

\begin{figure}[H]
    \centering
    \includegraphics[width=0.55\linewidth]{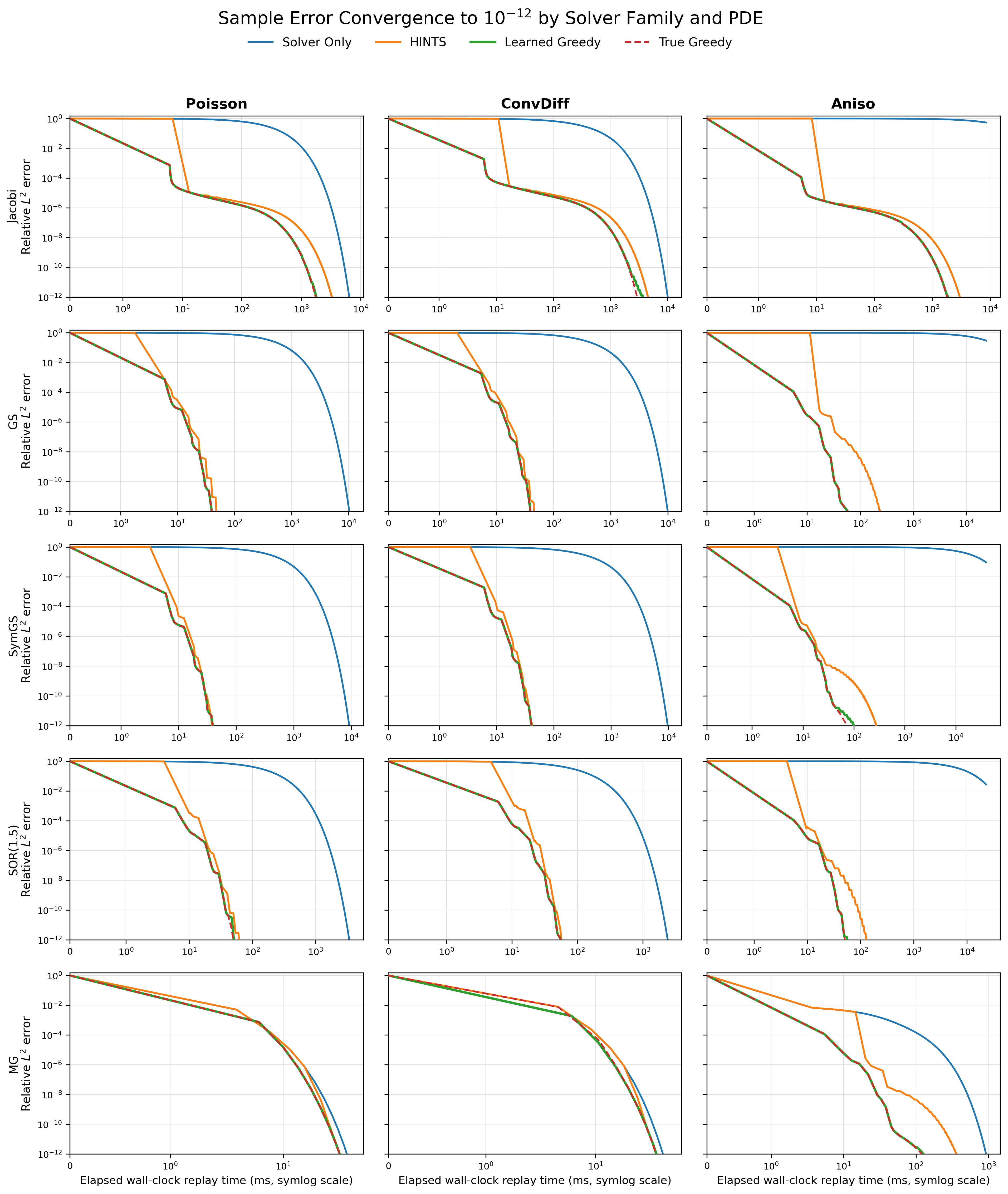}
    \caption{Convergence histories for representative test instances. Rows: Jacobi, Gauss-Seidel, Symmetric Gauss-Seidel, Successive Overrelaxation ($\omega = 1.5$), and Multigrid. Columns: Poisson, Convection Diffusion, and Anisotropic Diffusion. Greedy yields the fastest monotone decay.}
    \label{fig:errorcomparison2}
\end{figure}

\subsection{Prediction and Error Visualizations} \label{sec:error_viz}

\Cref{fig:prediction_visualization,fig:prediction_visualization2,fig:prediction_visualization3} provide qualitative visualizations of sample predictions across routing strategies and numerical solvers for all equations for a fixed wall-clock budget of 10 ms. While these prediction may appear identical to the ground truth, \Cref{fig:error_viz,fig:error_viz_2,fig:error_viz_3} reveals distinct error patterns with varying magnitudes across all predictions. In particular, our learned method consistently achieves the lower-magnitude errors compared to its respective baselines.

\begin{figure}[H]
    \centering
    \includegraphics[width=0.55\linewidth]{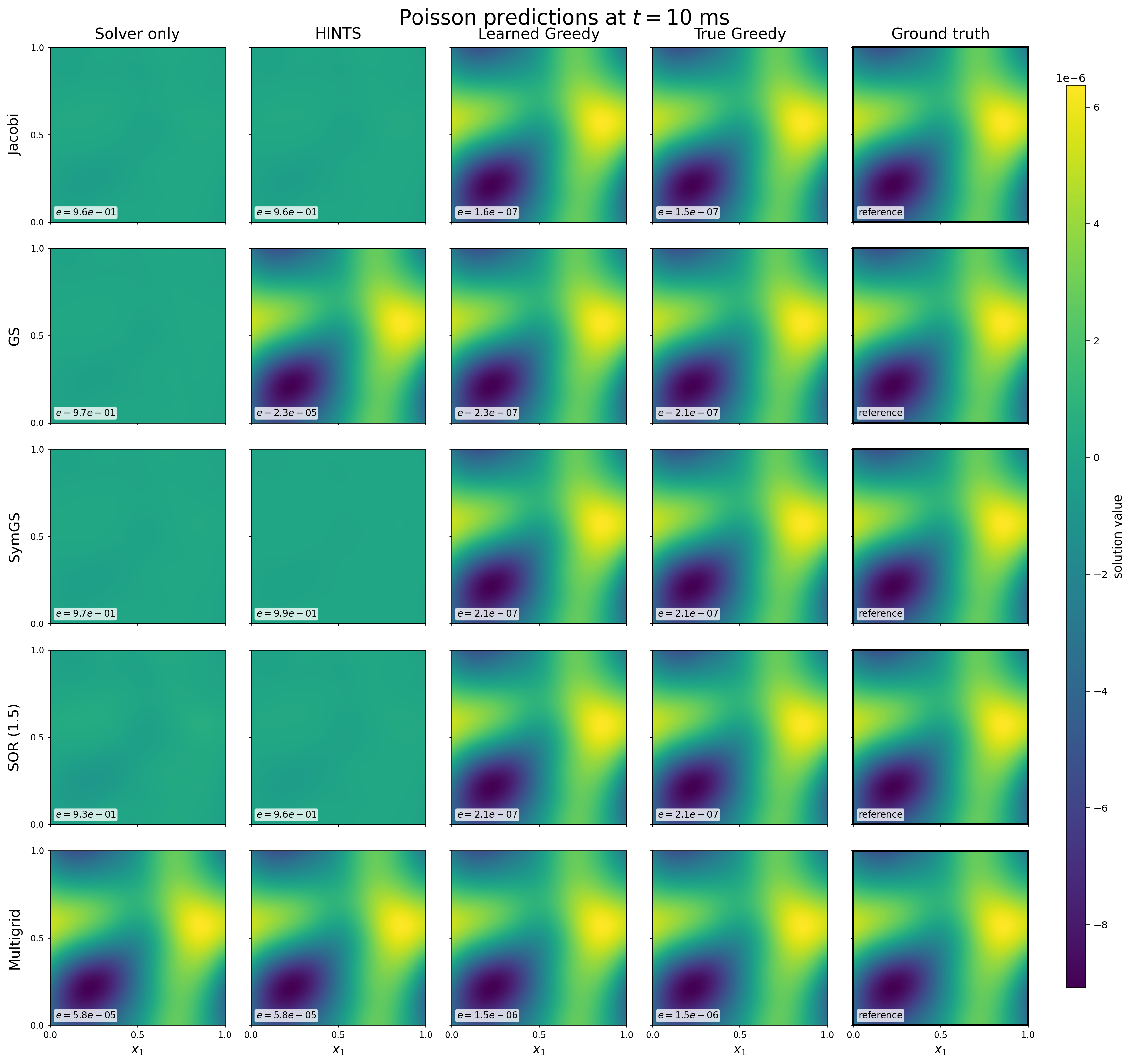}
    \caption{Sample predictions for the Poisson equation across different solver pairings. Each row corresponds to a numerical solver (Jacobi, Gauss–Seidel (GS), Symmetric GS, Damped Jacobi, and Successive Over-relaxation, MG), while columns show predictions from the corresponding solver-only baseline, HINTS, and the learned greedy router. The ground truth solution is shown on the right.}
    \label{fig:prediction_visualization}
\end{figure}
\begin{figure}[H]
    \centering
    \includegraphics[width=0.55\linewidth]{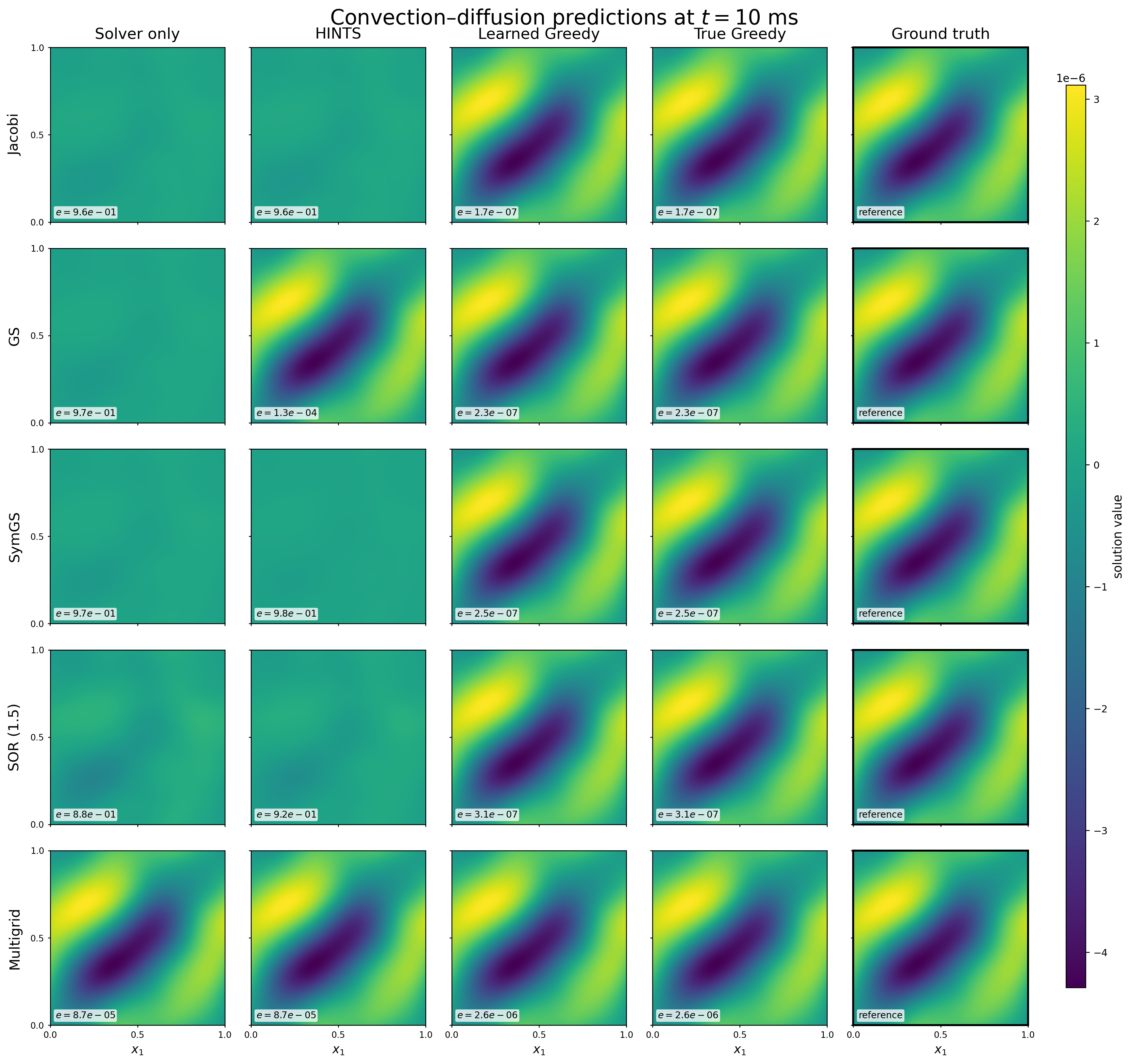}
    \caption{Sample predictions for the convection diffusion equation across different solver pairings. Each row corresponds to a numerical solver (Jacobi, Gauss–Seidel (GS), Symmetric GS, Damped Jacobi, and Successive Over-relaxation, MG), while columns show predictions from the corresponding solver-only baseline, HINTS, and the learned greedy router. The ground truth solution is shown on the right.}
    \label{fig:prediction_visualization2}
\end{figure}
\begin{figure}[H]
    \centering
    \includegraphics[width=0.55\linewidth]{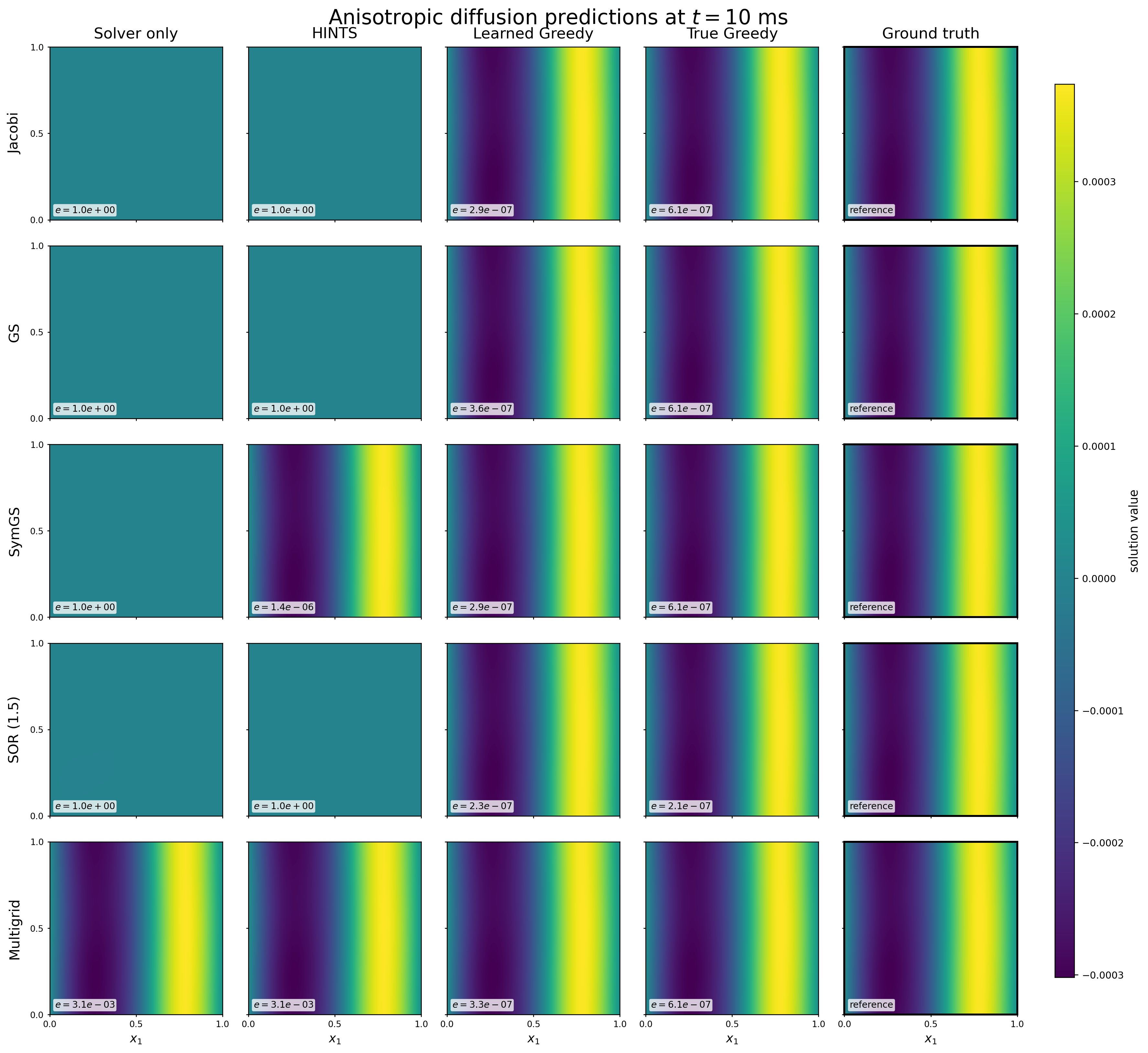}
    \caption{Sample predictions for the anisotropic diffusion equation across different solver pairings. Each row corresponds to a numerical solver (Jacobi, Gauss–Seidel (GS), Symmetric GS, Damped Jacobi, and Successive Over-relaxation, MG), while columns show predictions from the corresponding solver-only baseline, HINTS, and the learned greedy router. The ground truth solution is shown on the right.}
    \label{fig:prediction_visualization3}
\end{figure}
\begin{figure}[H]
    \centering
    \includegraphics[width=0.55\linewidth]{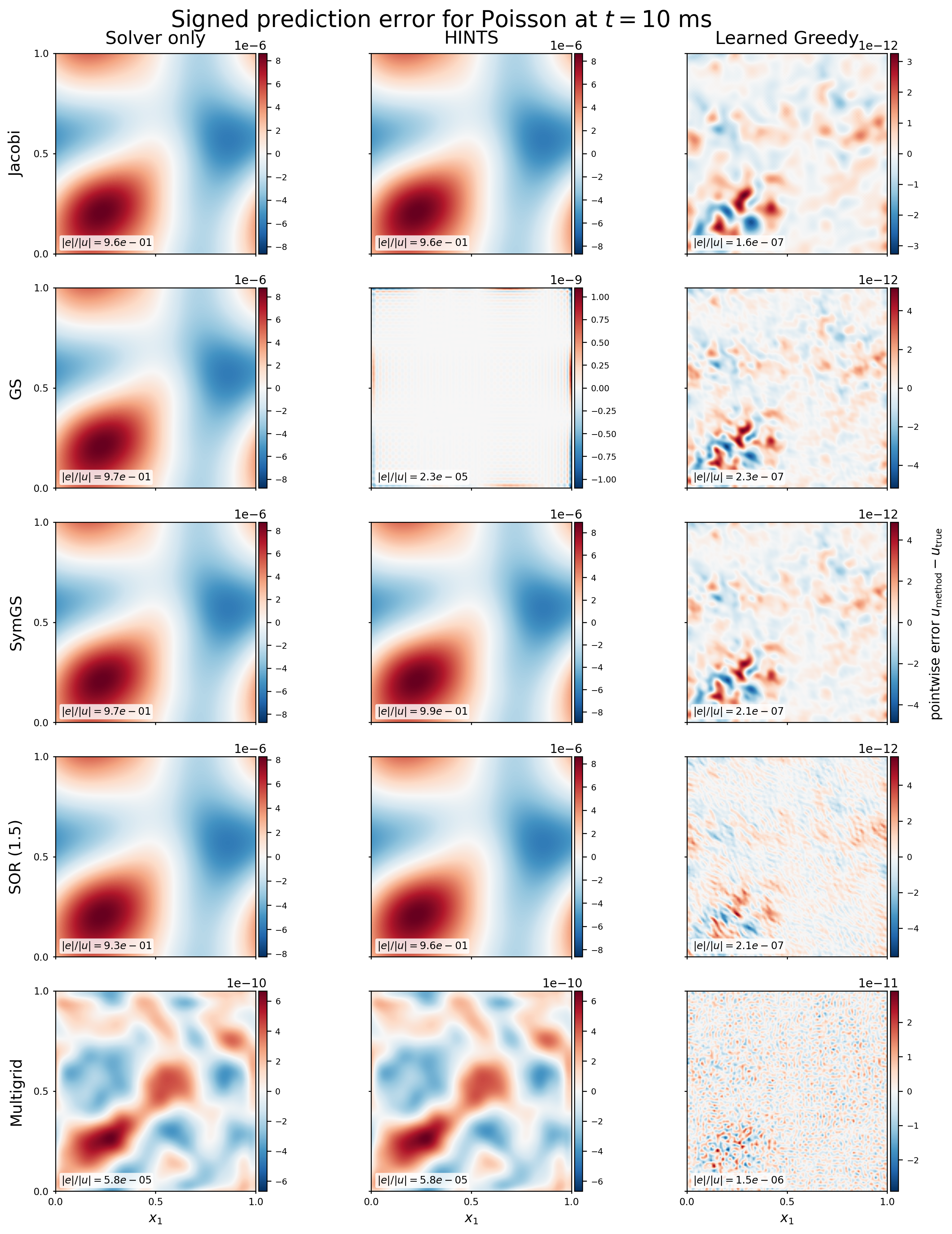}
    \caption{Pointwise error maps for Poisson sample predictions across solver families. Columns correspond to different numerical solvers, while columns show solver-only, HINTS, and learned greedy router. Errors are computed as $u_h - u_h^{(T)} = e_h^{(T)}$. The learned greedy router consistently yields lower-magnitude errors compared to both solver-only and HINTS baselines.}
    \label{fig:error_viz}
\end{figure}
\begin{figure}[H]
    \centering
    \includegraphics[width=0.55\linewidth]{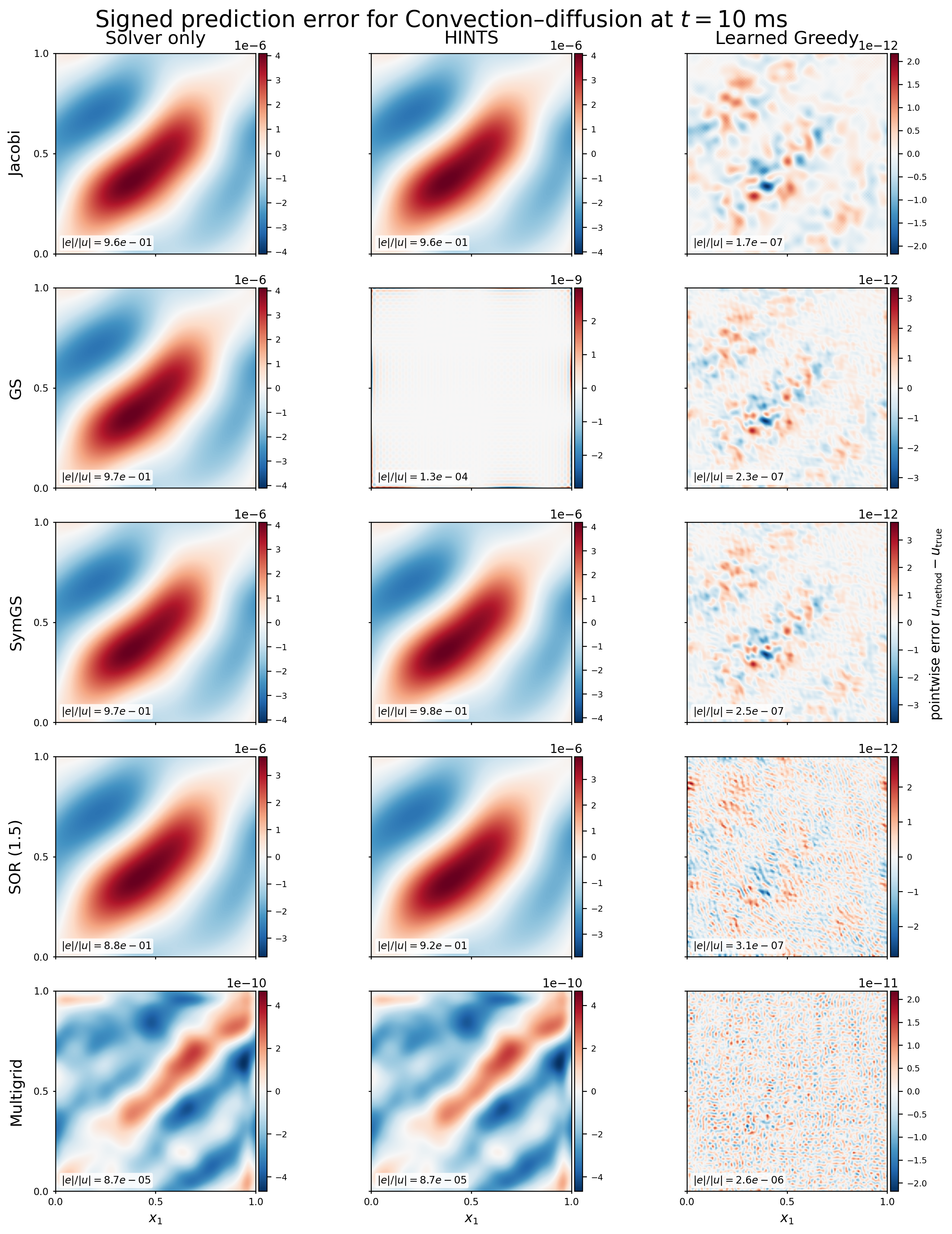}
    \caption{Pointwise error maps for convection diffusion sample predictions across solver families. Columns correspond to different numerical solvers, while columns show solver-only, HINTS, and learned greedy router. Errors are computed as $u_h - u_h^{(T)} = e_h^{(T)}$. The learned greedy router consistently yields lower-magnitude errors compared to both solver-only and HINTS baselines.}
    \label{fig:error_viz_2}
\end{figure}
\begin{figure}[H]
    \centering
    \includegraphics[width=0.55\linewidth]{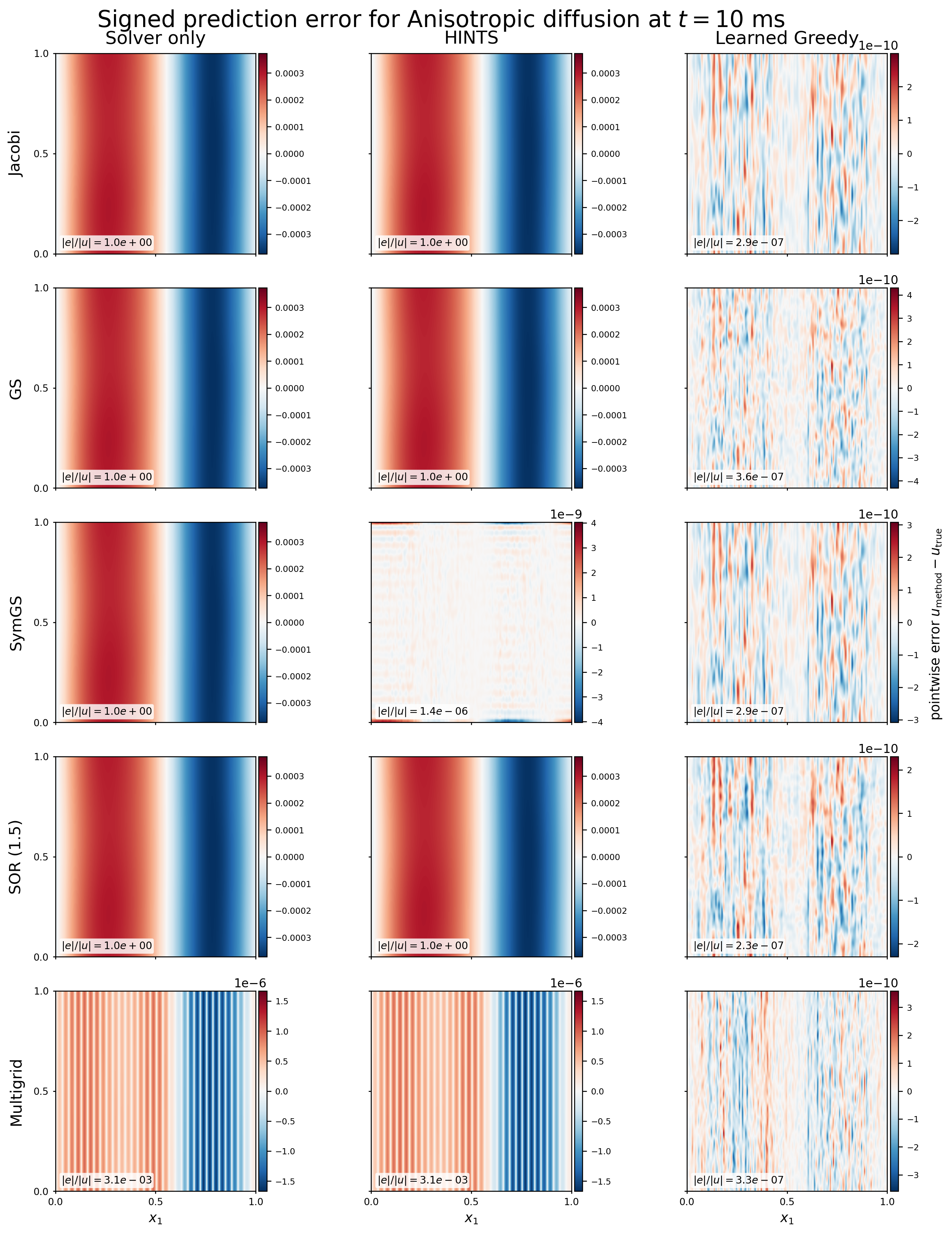}
    \caption{Pointwise error maps for anisotropic diffusion sample predictions across solver families. Columns correspond to different numerical solvers, while columns show solver-only, HINTS, and learned greedy router. Errors are computed as $u_h - u_h^{(T)} = e_h^{(T)}$. The learned greedy router consistently yields lower-magnitude errors compared to both solver-only and HINTS baselines.}
    \label{fig:error_viz_3}
\end{figure}

\subsection{Residual Comparison} \label{sec:residualcomparison}

\begin{table}[H]
\centering
\caption{Relative residual at a fixed 25 ms cost-calibrated selected-action budget and time-averaged residual AUC (lower is better). Values are mean (standard error) over 64 shared test instances. We use $R(t)=\|f_h-\mathcal{L}_h^a u_h(t)\|_2/\|f_h\|_2$ and $\mathrm{AUC}_{25}=(25\,\mathrm{ms})^{-1}\int_0^{25\,\mathrm{ms}}R(t)\,dt$. The state is held constant between completed actions; complete update and residual-refresh costs, the initial residual, and learned-router inference are charged. A learned routing decision is admitted only when its complete cost-matched macro-action fits in the remaining budget. Each reported $p$-value is from a paired one-sided $t$-test of the stated baseline against Learned Greedy, with the alternative that Learned Greedy is lower.}
\resizebox{\textwidth}{!}{%
\begin{tabular}{lcccccccccccc}
\hline\hline
Equation & \multicolumn{4}{c}{Poisson} & \multicolumn{4}{c}{ConvDiff} & \multicolumn{4}{c}{Anisotropic diffusion} \\
\hline
Method & $R(25\,\mathrm{ms})$ & $p$ & $\mathrm{AUC}_{25}$ & $p$ & $R(25\,\mathrm{ms})$ & $p$ & $\mathrm{AUC}_{25}$ & $p$ & $R(25\,\mathrm{ms})$ & $p$ & $\mathrm{AUC}_{25}$ & $p$ \\
\hline
\multicolumn{13}{c}{Jacobi-related solvers} \\
\hline
Jacobi Only & $0.606\;(0.034)$ & $<10^{-3}$ & $0.697\;(0.034)$ & $<10^{-3}$ & $0.659\;(0.035)$ & $<10^{-3}$ & $0.737\;(0.033)$ & $<10^{-3}$ & $0.649\;(0.031)$ & $<10^{-3}$ & $0.737\;(0.029)$ & $<10^{-3}$ \\
HINTS-Jacobi & $0.005\;(9.55\!\times\!10^{-4})$ & $<10^{-3}$ & $0.404\;(0.017)$ & $<10^{-3}$ & $0.005\;(0.001)$ & $1.000$ & $0.535\;(0.022)$ & $<10^{-3}$ & $0.013\;(0.003)$ & $1.000$ & $0.452\;(0.014)$ & $<10^{-3}$ \\
Learned Greedy-Jacobi & $0.004\;(7.74\!\times\!10^{-4})$ & -- & $0.259\;(0.001)$ & -- & $0.005\;(0.001)$ & -- & $0.263\;(0.002)$ & -- & $0.013\;(0.003)$ & -- & $0.248\;(0.004)$ & -- \\
\hline
\multicolumn{13}{c}{GS-related solvers} \\
\hline
GS Only & $0.662\;(0.035)$ & $<10^{-3}$ & $0.741\;(0.033)$ & $<10^{-3}$ & $0.660\;(0.035)$ & $<10^{-3}$ & $0.742\;(0.033)$ & $<10^{-3}$ & $0.690\;(0.032)$ & $<10^{-3}$ & $0.769\;(0.029)$ & $<10^{-3}$ \\
HINTS-GS & $1.99\!\times\!10^{-6}\;(5.61\!\times\!10^{-8})$ & $<10^{-3}$ & $0.271\;(0.008)$ & $0.011$ & $3.78\!\times\!10^{-7}\;(9.04\!\times\!10^{-9})$ & $<10^{-3}$ & $0.265\;(0.008)$ & $<10^{-3}$ & $3.77\!\times\!10^{-4}\;(2.01\!\times\!10^{-5})$ & $<10^{-3}$ & $0.546\;(0.019)$ & $<10^{-3}$ \\
Learned Greedy-GS & $9.02\!\times\!10^{-8}\;(1.52\!\times\!10^{-8})$ & -- & $0.249\;(8.81\!\times\!10^{-4})$ & -- & $1.51\!\times\!10^{-7}\;(2.55\!\times\!10^{-8})$ & -- & $0.231\;(9.68\!\times\!10^{-4})$ & -- & $4.51\!\times\!10^{-6}\;(5.74\!\times\!10^{-7})$ & -- & $0.234\;(0.001)$ & -- \\
\hline
\multicolumn{13}{c}{SymGS-related solvers} \\
\hline
SymGS Only & $0.649\;(0.035)$ & $<10^{-3}$ & $0.735\;(0.033)$ & $<10^{-3}$ & $0.657\;(0.035)$ & $<10^{-3}$ & $0.742\;(0.033)$ & $<10^{-3}$ & $0.682\;(0.032)$ & $<10^{-3}$ & $0.767\;(0.030)$ & $<10^{-3}$ \\
HINTS-SymGS & $7.44\!\times\!10^{-7}\;(3.47\!\times\!10^{-8})$ & $<10^{-3}$ & $0.312\;(0.011)$ & $<10^{-3}$ & $7.57\!\times\!10^{-7}\;(3.53\!\times\!10^{-8})$ & $<10^{-3}$ & $0.327\;(0.011)$ & $<10^{-3}$ & $7.03\!\times\!10^{-5}\;(3.27\!\times\!10^{-6})$ & $<10^{-3}$ & $0.311\;(0.009)$ & $<10^{-3}$ \\
Learned Greedy-SymGS & $3.21\!\times\!10^{-8}\;(5.52\!\times\!10^{-9})$ & -- & $0.253\;(0.001)$ & -- & $5.90\!\times\!10^{-8}\;(1.01\!\times\!10^{-8})$ & -- & $0.259\;(0.001)$ & -- & $2.63\!\times\!10^{-6}\;(3.35\!\times\!10^{-7})$ & -- & $0.238\;(0.002)$ & -- \\
\hline
\multicolumn{13}{c}{SOR (1.5)-related solvers} \\
\hline
SOR (1.5) Only & $0.541\;(0.033)$ & $<10^{-3}$ & $0.666\;(0.034)$ & $<10^{-3}$ & $0.492\;(0.031)$ & $<10^{-3}$ & $0.640\;(0.033)$ & $<10^{-3}$ & $0.588\;(0.032)$ & $<10^{-3}$ & $0.716\;(0.031)$ & $<10^{-3}$ \\
HINTS-SOR (1.5) & $1.50\!\times\!10^{-5}\;(1.14\!\times\!10^{-6})$ & $0.005$ & $0.327\;(0.013)$ & $<10^{-3}$ & $2.91\!\times\!10^{-5}\;(1.25\!\times\!10^{-6})$ & $1.000$ & $0.354\;(0.015)$ & $<10^{-3}$ & $1.64\!\times\!10^{-4}\;(5.63\!\times\!10^{-6})$ & $<10^{-3}$ & $0.340\;(0.011)$ & $<10^{-3}$ \\
Learned Greedy-SOR (1.5) & $1.19\!\times\!10^{-5}\;(2.10\!\times\!10^{-6})$ & -- & $0.255\;(0.002)$ & -- & $6.25\!\times\!10^{-4}\;(1.11\!\times\!10^{-4})$ & -- & $0.266\;(0.002)$ & -- & $1.01\!\times\!10^{-4}\;(1.47\!\times\!10^{-5})$ & -- & $0.241\;(0.002)$ & -- \\
\hline
\multicolumn{13}{c}{MG-related solvers} \\
\hline
MG Only & $6.60\!\times\!10^{-8}\;(5.70\!\times\!10^{-9})$ & $0.998$ & $0.164\;(2.53\!\times\!10^{-4})$ & $1.000$ & $2.79\!\times\!10^{-7}\;(2.25\!\times\!10^{-8})$ & $<10^{-3}$ & $0.196\;(2.12\!\times\!10^{-4})$ & $1.000$ & $0.205\;(0.013)$ & $<10^{-3}$ & $0.458\;(0.019)$ & $<10^{-3}$ \\
HINTS-MG & $1.43\!\times\!10^{-6}\;(1.07\!\times\!10^{-7})$ & $<10^{-3}$ & $0.164\;(2.53\!\times\!10^{-4})$ & $1.000$ & $7.46\!\times\!10^{-7}\;(5.29\!\times\!10^{-8})$ & $<10^{-3}$ & $0.196\;(2.12\!\times\!10^{-4})$ & $1.000$ & $0.001\;(6.69\!\times\!10^{-5})$ & $<10^{-3}$ & $0.426\;(0.017)$ & $<10^{-3}$ \\
Learned Greedy-MG & $1.10\!\times\!10^{-7}\;(1.77\!\times\!10^{-8})$ & -- & $0.278\;(0.005)$ & -- & $3.29\!\times\!10^{-8}\;(5.38\!\times\!10^{-9})$ & -- & $0.286\;(0.006)$ & -- & $2.35\!\times\!10^{-4}\;(3.88\!\times\!10^{-5})$ & -- & $0.259\;(0.005)$ & -- \\
\hline
\end{tabular}%
}
\label{tab:residualcomparison}
\end{table}

\begin{table}[H]
\centering
\caption{Relative residual at a fixed 25 ms cost-calibrated selected-action budget and time-averaged residual AUC for growing solver ensembles and their validation-selected best constituent pairwise routers (lower is better). Values are mean (standard error) over 64 shared test instances. We use $R(t)=\|f_h-\mathcal{L}_h^a u_h(t)\|_2/\|f_h\|_2$ and $\mathrm{AUC}_{25}=(25\,\mathrm{ms})^{-1}\int_0^{25\,\mathrm{ms}}R(t)\,dt$. The state is held constant between completed actions; complete update and residual-refresh costs, the initial residual, and learned-router inference are charged. Each $p$-value is from a paired one-sided $t$-test with the alternative that the ensemble router is lower than its best pairwise comparator.}
\label{tab:ensemble-residualcomparison}
\resizebox{\textwidth}{!}{%
\begin{tabular}{llcccccccccccc}
\hline\hline
Actions & Method & \multicolumn{4}{c}{Poisson} & \multicolumn{4}{c}{ConvDiff} & \multicolumn{4}{c}{Anisotropic diffusion} \\
\hline
& & $R(25\,\mathrm{ms})$ & $p$ & $\mathrm{AUC}_{25}$ & $p$ & $R(25\,\mathrm{ms})$ & $p$ & $\mathrm{AUC}_{25}$ & $p$ & $R(25\,\mathrm{ms})$ & $p$ & $\mathrm{AUC}_{25}$ & $p$ \\
\hline
\multirow{2}{*}{$\mathcal{A}_3$} & Best pairwise router & $1.23\!\times\!10^{-8}\;(1.75\!\times\!10^{-9})$ & $<10^{-3}$ & $0.057\;(2.07\!\times\!10^{-3})$ & $<10^{-3}$ & $2.25\!\times\!10^{-6}\;(2.39\!\times\!10^{-7})$ & $<10^{-3}$ & $0.068\;(2.43\!\times\!10^{-3})$ & $<10^{-3}$ & $4.18\!\times\!10^{-7}\;(2.84\!\times\!10^{-8})$ & $<10^{-3}$ & $0.279\;(7.57\!\times\!10^{-3})$ & $<10^{-3}$ \\
& Router$(\mathcal{A}_3)$ & \textbf{$1.04\!\times\!10^{-13}\;(6.06\!\times\!10^{-15})$} & -- & \textbf{$0.055\;(1.67\!\times\!10^{-3})$} & -- & \textbf{$1.13\!\times\!10^{-11}\;(2.30\!\times\!10^{-12})$} & -- & \textbf{$0.062\;(1.77\!\times\!10^{-3})$} & -- & \textbf{$8.51\!\times\!10^{-8}\;(6.80\!\times\!10^{-8})$} & -- & \textbf{$0.243\;(8.12\!\times\!10^{-3})$} & -- \\
\hline
\multirow{2}{*}{$\mathcal{A}_4$} & Best pairwise router & $1.23\!\times\!10^{-8}\;(1.75\!\times\!10^{-9})$ & $<10^{-3}$ & $0.057\;(2.07\!\times\!10^{-3})$ & $<10^{-3}$ & $7.24\!\times\!10^{-7}\;(8.62\!\times\!10^{-8})$ & $<10^{-3}$ & $0.069\;(2.50\!\times\!10^{-3})$ & $<10^{-3}$ & $6.09\!\times\!10^{-7}\;(7.81\!\times\!10^{-8})$ & $<10^{-3}$ & $0.279\;(7.70\!\times\!10^{-3})$ & $<10^{-3}$ \\
& Router$(\mathcal{A}_4)$ & \textbf{$1.23\!\times\!10^{-13}\;(4.75\!\times\!10^{-15})$} & -- & \textbf{$0.055\;(1.67\!\times\!10^{-3})$} & -- & \textbf{$2.35\!\times\!10^{-13}\;(2.12\!\times\!10^{-14})$} & -- & \textbf{$0.062\;(1.78\!\times\!10^{-3})$} & -- & \textbf{$7.00\!\times\!10^{-8}\;(1.01\!\times\!10^{-8})$} & -- & \textbf{$0.243\;(8.03\!\times\!10^{-3})$} & -- \\
\hline
\multirow{2}{*}{$\mathcal{A}_5$} & Best pairwise router & $1.23\!\times\!10^{-8}\;(1.75\!\times\!10^{-9})$ & $<10^{-3}$ & $0.057\;(2.07\!\times\!10^{-3})$ & $<10^{-3}$ & $6.95\!\times\!10^{-8}\;(9.13\!\times\!10^{-9})$ & $<10^{-3}$ & $0.070\;(2.91\!\times\!10^{-3})$ & $<10^{-3}$ & $6.09\!\times\!10^{-7}\;(7.81\!\times\!10^{-8})$ & $<10^{-3}$ & $0.279\;(7.70\!\times\!10^{-3})$ & $<10^{-3}$ \\
& Router$(\mathcal{A}_5)$ & \textbf{$1.03\!\times\!10^{-13}\;(2.85\!\times\!10^{-15})$} & -- & \textbf{$0.055\;(1.62\!\times\!10^{-3})$} & -- & \textbf{$3.29\!\times\!10^{-11}\;(7.37\!\times\!10^{-12})$} & -- & \textbf{$0.062\;(1.77\!\times\!10^{-3})$} & -- & \textbf{$6.45\!\times\!10^{-8}\;(2.67\!\times\!10^{-8})$} & -- & \textbf{$0.243\;(8.20\!\times\!10^{-3})$} & -- \\
\hline
\multirow{2}{*}{$\mathcal{A}_6$} & Best pairwise router & $1.23\!\times\!10^{-8}\;(1.75\!\times\!10^{-9})$ & $<10^{-3}$ & $0.057\;(2.07\!\times\!10^{-3})$ & $<10^{-3}$ & $6.95\!\times\!10^{-8}\;(9.13\!\times\!10^{-9})$ & $<10^{-3}$ & $0.070\;(2.91\!\times\!10^{-3})$ & $<10^{-3}$ & $6.09\!\times\!10^{-7}\;(7.81\!\times\!10^{-8})$ & $<10^{-3}$ & $0.279\;(7.70\!\times\!10^{-3})$ & $<10^{-3}$ \\
& Router$(\mathcal{A}_6)$ & \textbf{$6.82\!\times\!10^{-13}\;(1.10\!\times\!10^{-13})$} & -- & \textbf{$0.055\;(1.67\!\times\!10^{-3})$} & -- & \textbf{$1.57\!\times\!10^{-10}\;(7.03\!\times\!10^{-11})$} & -- & \textbf{$0.063\;(1.94\!\times\!10^{-3})$} & -- & \textbf{$7.56\!\times\!10^{-8}\;(2.99\!\times\!10^{-8})$} & -- & \textbf{$0.243\;(8.00\!\times\!10^{-3})$} & -- \\
\hline
\end{tabular}%
}
\end{table}

\Cref{tab:residualcomparison} summarizes the performance of single-solver schedules, HINTS, and greedy with respect to the final residuals $r_h^{(T)} =\|f_h - \mathcal{L}_h^a u_h^{(T)}\|/\|f_h\|$ and its residual area under curve $\mathrm{AUC}_{25}=(25\,\mathrm{ms})^{-1}\int_0^{25\,\mathrm{ms}}R(t)\,dt$ . Greedy outperforms its HINTS and single-solver counterparts in most equations. We must note that our greedy router is trained to reduce error, not residual. The same error can induce very different residuals depending on the spectrum $\mathcal{L}_h^a$. \Cref{tab:ensemble-residualcomparison} exhibits how residuals are affected by the number of solvers in the solver ensemble. Similar to error, we observe both the final residual and AUC decrease as the number of solvers increase.

\subsection{Fourier mode-wise error comparison} \label{sec:fouriermodeerrorcomparison}

We assess frequency-resolved performance by projecting the solution er onto the discrete Fourier basis. \Cref{tab:red-fourier-modes-poisson,tab:red-fourier-modes-convdiff,tab:red-fourier-modes-anisodiff}reports the endpoint magnitude and time-averaged AUC of modes 1, 5, and 10 at a fixed \(25\,\mathrm{ms}\) budget. The learned router generally improves low-frequency convergence relative to solver-only schedules, while classical smoothers can retain an advantage for some higher-frequency modes or AUCs. This behavior is consistent with the complementary strengths motivating hybrid solvers: neural corrections are effective on smooth error components, whereas classical iterations efficiently damp oscillatory components. The router nevertheless optimizes the full \(L^2\) error rather than individual Fourier coefficients. By Parseval’s identity,
\[
\|e_h^{(t)}\|_2^2=\sum_m |\widehat e_m^{(t)}|^2,
\]reducing its objective decreases aggregate spectral energy, but does not guarantee improvement in every mode separately.
\Cref{tab:fourier-modes-ensemble-poisson,tab:fourier-modes-ensemble-convdiff,tab:fourier-modes-ensemble-anisodiff} compares each growing-ensemble router with its validation-selected best constituent pairwise router. The ensembles substantially reduce many endpoint mode errors, particularly for convection–diffusion and for the intermediate- and high-frequency modes of the other equations. Their AUC gains are smaller and occasionally mixed, especially for Poisson, indicating that similar spectral trajectories can lead to different endpoint accuracies. Moreover, the results do not improve monotonically from \(\mathcal A_3\) to \(\mathcal A_6\); instead, they largely saturate after a small number of complementary actions. Thus, the frequency-resolved results support the benefit of solver diversity without suggesting that every added action improves every Fourier component.

\begin{table}[h]
\centering
\caption{Fourier-mode error at a 25 ms wall-clock budget for Poisson on 64 shared test instances. Each entry is mean (standard error) of $|\widehat{e}_{k,k}|$, or its time-averaged wall-clock AUC, using NumPy's unnormalized DFT of the demeaned solution error. Lower is better. Bold marks the lowest deployable method within each solver family and column.}
\label{tab:red-fourier-modes-poisson}
\resizebox{\linewidth}{!}{%
\begin{tabular}{lcccccc}
\hline
& \multicolumn{2}{c}{Mode $(1,1)$} & \multicolumn{2}{c}{Mode $(5,5)$} & \multicolumn{2}{c}{Mode $(10,10)$} \\
Method & Error & AUC & Error & AUC & Error & AUC \\
\hline
\multicolumn{7}{c}{Jacobi-related solvers} \\
\hline
Solver Only & $2.53\!\times\!10^{-1}\;(7.91\!\times\!10^{-2})$ & $2.82\!\times\!10^{-1}\;(8.81\!\times\!10^{-2})$ & $3.41\!\times\!10^{-5}\;(1.39\!\times\!10^{-5})$ & $\mathbf{1.39\!\times\!10^{-3}\;(5.66\!\times\!10^{-4})}$ & $2.20\!\times\!10^{-13}\;(9.06\!\times\!10^{-14})$ & $\mathbf{4.67\!\times\!10^{-5}\;(1.92\!\times\!10^{-5})}$ \\
HINTS & $\mathbf{1.75\!\times\!10^{-8}\;(5.54\!\times\!10^{-9})}$ & $1.57\!\times\!10^{-1}\;(4.93\!\times\!10^{-2})$ & $4.62\!\times\!10^{-9}\;(1.30\!\times\!10^{-9})$ & $1.49\!\times\!10^{-3}\;(6.06\!\times\!10^{-4})$ & $3.42\!\times\!10^{-11}\;(8.64\!\times\!10^{-12})$ & $4.71\!\times\!10^{-5}\;(1.94\!\times\!10^{-5})$ \\
Learned Greedy & $1.91\!\times\!10^{-8}\;(6.15\!\times\!10^{-9})$ & $\mathbf{7.83\!\times\!10^{-2}\;(2.45\!\times\!10^{-2})}$ & $\mathbf{3.03\!\times\!10^{-10}\;(9.01\!\times\!10^{-11})}$ & $1.81\!\times\!10^{-3}\;(7.39\!\times\!10^{-4})$ & $\mathbf{1.16\!\times\!10^{-15}\;(3.42\!\times\!10^{-16})}$ & $2.22\!\times\!10^{-4}\;(9.14\!\times\!10^{-5})$ \\
\hline
\multicolumn{7}{c}{GS-related solvers} \\
\hline
Solver Only & $2.74\!\times\!10^{-1}\;(8.56\!\times\!10^{-2})$ & $2.93\!\times\!10^{-1}\;(9.16\!\times\!10^{-2})$ & $3.21\!\times\!10^{-4}\;(1.32\!\times\!10^{-4})$ & $2.28\!\times\!10^{-3}\;(9.33\!\times\!10^{-4})$ & $5.43\!\times\!10^{-8}\;(2.25\!\times\!10^{-8})$ & $\mathbf{9.42\!\times\!10^{-5}\;(3.87\!\times\!10^{-5})}$ \\
HINTS & $\mathbf{1.98\!\times\!10^{-10}\;(6.86\!\times\!10^{-11})}$ & $9.76\!\times\!10^{-2}\;(3.05\!\times\!10^{-2})$ & $\mathbf{4.60\!\times\!10^{-10}\;(1.87\!\times\!10^{-10})}$ & $1.89\!\times\!10^{-3}\;(7.70\!\times\!10^{-4})$ & $\mathbf{4.33\!\times\!10^{-10}\;(1.97\!\times\!10^{-10})}$ & $1.51\!\times\!10^{-4}\;(6.22\!\times\!10^{-5})$ \\
Learned Greedy & $9.98\!\times\!10^{-9}\;(4.43\!\times\!10^{-9})$ & $\mathbf{7.63\!\times\!10^{-2}\;(2.39\!\times\!10^{-2})}$ & $6.15\!\times\!10^{-9}\;(2.35\!\times\!10^{-9})$ & $\mathbf{1.77\!\times\!10^{-3}\;(7.20\!\times\!10^{-4})}$ & $5.78\!\times\!10^{-10}\;(2.08\!\times\!10^{-10})$ & $2.17\!\times\!10^{-4}\;(8.91\!\times\!10^{-5})$ \\
\hline
\multicolumn{7}{c}{SymGS-related solvers} \\
\hline
Solver Only & $2.69\!\times\!10^{-1}\;(8.42\!\times\!10^{-2})$ & $2.91\!\times\!10^{-1}\;(9.10\!\times\!10^{-2})$ & $2.19\!\times\!10^{-4}\;(9.03\!\times\!10^{-5})$ & $2.14\!\times\!10^{-3}\;(8.73\!\times\!10^{-4})$ & $9.18\!\times\!10^{-8}\;(2.66\!\times\!10^{-8})$ & $\mathbf{9.45\!\times\!10^{-5}\;(3.90\!\times\!10^{-5})}$ \\
HINTS & $6.03\!\times\!10^{-9}\;(1.79\!\times\!10^{-9})$ & $1.15\!\times\!10^{-1}\;(3.59\!\times\!10^{-2})$ & $3.75\!\times\!10^{-9}\;(1.44\!\times\!10^{-9})$ & $1.93\!\times\!10^{-3}\;(7.86\!\times\!10^{-4})$ & $8.60\!\times\!10^{-10}\;(3.71\!\times\!10^{-10})$ & $1.21\!\times\!10^{-4}\;(5.01\!\times\!10^{-5})$ \\
Learned Greedy & $\mathbf{3.35\!\times\!10^{-9}\;(1.43\!\times\!10^{-9})}$ & $\mathbf{7.68\!\times\!10^{-2}\;(2.40\!\times\!10^{-2})}$ & $\mathbf{2.11\!\times\!10^{-9}\;(7.89\!\times\!10^{-10})}$ & $\mathbf{1.78\!\times\!10^{-3}\;(7.25\!\times\!10^{-4})}$ & $\mathbf{2.52\!\times\!10^{-10}\;(8.96\!\times\!10^{-11})}$ & $2.18\!\times\!10^{-4}\;(8.98\!\times\!10^{-5})$ \\
\hline
\multicolumn{7}{c}{SOR(1.5)-related solvers} \\
\hline
Solver Only & $2.12\!\times\!10^{-1}\;(6.62\!\times\!10^{-2})$ & $2.59\!\times\!10^{-1}\;(8.11\!\times\!10^{-2})$ & $1.05\!\times\!10^{-5}\;(4.04\!\times\!10^{-6})$ & $\mathbf{1.19\!\times\!10^{-3}\;(4.86\!\times\!10^{-4})}$ & $1.76\!\times\!10^{-8}\;(6.56\!\times\!10^{-9})$ & $\mathbf{8.17\!\times\!10^{-5}\;(3.31\!\times\!10^{-5})}$ \\
HINTS & $1.77\!\times\!10^{-7}\;(5.94\!\times\!10^{-8})$ & $1.21\!\times\!10^{-1}\;(3.77\!\times\!10^{-2})$ & $2.43\!\times\!10^{-8}\;(9.51\!\times\!10^{-9})$ & $1.40\!\times\!10^{-3}\;(5.70\!\times\!10^{-4})$ & $1.62\!\times\!10^{-9}\;(5.47\!\times\!10^{-10})$ & $9.88\!\times\!10^{-5}\;(4.00\!\times\!10^{-5})$ \\
Learned Greedy & $\mathbf{4.96\!\times\!10^{-8}\;(2.20\!\times\!10^{-8})}$ & $\mathbf{7.68\!\times\!10^{-2}\;(2.40\!\times\!10^{-2})}$ & $\mathbf{1.01\!\times\!10^{-8}\;(3.75\!\times\!10^{-9})}$ & $1.78\!\times\!10^{-3}\;(7.25\!\times\!10^{-4})$ & $\mathbf{5.86\!\times\!10^{-10}\;(2.24\!\times\!10^{-10})}$ & $2.19\!\times\!10^{-4}\;(8.98\!\times\!10^{-5})$ \\
\hline
\multicolumn{7}{c}{MG-related solvers} \\
\hline
Solver Only & $3.71\!\times\!10^{-12}\;(1.28\!\times\!10^{-12})$ & $\mathbf{5.01\!\times\!10^{-2}\;(1.57\!\times\!10^{-2})}$ & $4.36\!\times\!10^{-9}\;(1.75\!\times\!10^{-9})$ & $\mathbf{1.26\!\times\!10^{-3}\;(5.14\!\times\!10^{-4})}$ & $1.47\!\times\!10^{-10}\;(5.22\!\times\!10^{-11})$ & $\mathbf{1.52\!\times\!10^{-4}\;(6.23\!\times\!10^{-5})}$ \\
HINTS & $\mathbf{8.71\!\times\!10^{-14}\;(4.41\!\times\!10^{-14})}$ & $5.01\!\times\!10^{-2}\;(1.57\!\times\!10^{-2})$ & $\mathbf{8.21\!\times\!10^{-14}\;(3.15\!\times\!10^{-14})}$ & $1.26\!\times\!10^{-3}\;(5.14\!\times\!10^{-4})$ & $\mathbf{1.23\!\times\!10^{-13}\;(5.09\!\times\!10^{-14})}$ & $1.52\!\times\!10^{-4}\;(6.23\!\times\!10^{-5})$ \\
Learned Greedy & $1.06\!\times\!10^{-11}\;(3.79\!\times\!10^{-12})$ & $7.76\!\times\!10^{-2}\;(2.43\!\times\!10^{-2})$ & $5.69\!\times\!10^{-10}\;(2.22\!\times\!10^{-10})$ & $1.80\!\times\!10^{-3}\;(7.32\!\times\!10^{-4})$ & $2.10\!\times\!10^{-10}\;(8.08\!\times\!10^{-11})$ & $2.20\!\times\!10^{-4}\;(9.06\!\times\!10^{-5})$ \\
\hline
\end{tabular}}
\end{table}
\begin{table}[h]
\centering
\caption{Fourier-mode error at a 25 ms wall-clock budget for ConvDiff on 64 shared test instances. Each entry is mean (standard error) of $|\widehat{e}_{k,k}|$, or its time-averaged wall-clock AUC, using NumPy's unnormalized DFT of the demeaned solution error. Lower is better. Bold marks the lowest deployable method within each solver family and column.}
\label{tab:red-fourier-modes-convdiff}
\resizebox{\linewidth}{!}{%
\begin{tabular}{lcccccc}
\hline
& \multicolumn{2}{c}{Mode $(1,1)$} & \multicolumn{2}{c}{Mode $(5,5)$} & \multicolumn{2}{c}{Mode $(10,10)$} \\
Method & Error & AUC & Error & AUC & Error & AUC \\
\hline
\multicolumn{7}{c}{Jacobi-related solvers} \\
\hline
Solver Only & $8.21\!\times\!10^{-2}\;(2.57\!\times\!10^{-2})$ & $8.78\!\times\!10^{-2}\;(2.75\!\times\!10^{-2})$ & $2.19\!\times\!10^{-4}\;(8.95\!\times\!10^{-5})$ & $1.82\!\times\!10^{-3}\;(7.43\!\times\!10^{-4})$ & $9.26\!\times\!10^{-10}\;(3.81\!\times\!10^{-10})$ & $\mathbf{7.17\!\times\!10^{-5}\;(2.95\!\times\!10^{-5})}$ \\
HINTS & $8.12\!\times\!10^{-9}\;(2.69\!\times\!10^{-9})$ & $6.22\!\times\!10^{-2}\;(1.95\!\times\!10^{-2})$ & $3.49\!\times\!10^{-9}\;(1.22\!\times\!10^{-9})$ & $1.81\!\times\!10^{-3}\;(7.37\!\times\!10^{-4})$ & $1.21\!\times\!10^{-10}\;(4.03\!\times\!10^{-11})$ & $7.20\!\times\!10^{-5}\;(2.96\!\times\!10^{-5})$ \\
Learned Greedy & $\mathbf{5.98\!\times\!10^{-9}\;(1.70\!\times\!10^{-9})}$ & $\mathbf{2.37\!\times\!10^{-2}\;(7.41\!\times\!10^{-3})}$ & $\mathbf{6.84\!\times\!10^{-10}\;(1.99\!\times\!10^{-10})}$ & $\mathbf{1.55\!\times\!10^{-3}\;(6.30\!\times\!10^{-4})}$ & $\mathbf{3.14\!\times\!10^{-13}\;(1.04\!\times\!10^{-13})}$ & $2.14\!\times\!10^{-4}\;(8.81\!\times\!10^{-5})$ \\
\hline
\multicolumn{7}{c}{GS-related solvers} \\
\hline
Solver Only & $8.10\!\times\!10^{-2}\;(2.53\!\times\!10^{-2})$ & $8.73\!\times\!10^{-2}\;(2.73\!\times\!10^{-2})$ & $2.05\!\times\!10^{-4}\;(8.48\!\times\!10^{-5})$ & $1.83\!\times\!10^{-3}\;(7.48\!\times\!10^{-4})$ & $2.79\!\times\!10^{-8}\;(1.03\!\times\!10^{-8})$ & $\mathbf{9.02\!\times\!10^{-5}\;(3.72\!\times\!10^{-5})}$ \\
HINTS & $\mathbf{1.86\!\times\!10^{-10}\;(5.36\!\times\!10^{-11})}$ & $2.85\!\times\!10^{-2}\;(8.92\!\times\!10^{-3})$ & $\mathbf{2.98\!\times\!10^{-10}\;(1.15\!\times\!10^{-10})}$ & $1.50\!\times\!10^{-3}\;(6.13\!\times\!10^{-4})$ & $\mathbf{1.99\!\times\!10^{-10}\;(8.82\!\times\!10^{-11})}$ & $1.30\!\times\!10^{-4}\;(5.37\!\times\!10^{-5})$ \\
Learned Greedy & $1.49\!\times\!10^{-8}\;(6.63\!\times\!10^{-9})$ & $\mathbf{2.11\!\times\!10^{-2}\;(6.61\!\times\!10^{-3})}$ & $8.76\!\times\!10^{-9}\;(3.33\!\times\!10^{-9})$ & $\mathbf{1.38\!\times\!10^{-3}\;(5.62\!\times\!10^{-4})}$ & $7.90\!\times\!10^{-10}\;(2.85\!\times\!10^{-10})$ & $1.92\!\times\!10^{-4}\;(7.87\!\times\!10^{-5})$ \\
\hline
\multicolumn{7}{c}{SymGS-related solvers} \\
\hline
Solver Only & $8.17\!\times\!10^{-2}\;(2.56\!\times\!10^{-2})$ & $8.77\!\times\!10^{-2}\;(2.75\!\times\!10^{-2})$ & $2.42\!\times\!10^{-4}\;(9.51\!\times\!10^{-5})$ & $1.93\!\times\!10^{-3}\;(7.82\!\times\!10^{-4})$ & $1.70\!\times\!10^{-7}\;(5.91\!\times\!10^{-8})$ & $\mathbf{9.89\!\times\!10^{-5}\;(4.09\!\times\!10^{-5})}$ \\
HINTS & $6.20\!\times\!10^{-9}\;(1.86\!\times\!10^{-9})$ & $3.60\!\times\!10^{-2}\;(1.13\!\times\!10^{-2})$ & $3.60\!\times\!10^{-9}\;(1.39\!\times\!10^{-9})$ & $1.71\!\times\!10^{-3}\;(6.94\!\times\!10^{-4})$ & $7.76\!\times\!10^{-10}\;(3.25\!\times\!10^{-10})$ & $1.23\!\times\!10^{-4}\;(5.08\!\times\!10^{-5})$ \\
Learned Greedy & $\mathbf{4.84\!\times\!10^{-9}\;(2.11\!\times\!10^{-9})}$ & $\mathbf{2.36\!\times\!10^{-2}\;(7.37\!\times\!10^{-3})}$ & $\mathbf{3.29\!\times\!10^{-9}\;(1.22\!\times\!10^{-9})}$ & $\mathbf{1.54\!\times\!10^{-3}\;(6.27\!\times\!10^{-4})}$ & $\mathbf{5.05\!\times\!10^{-10}\;(1.80\!\times\!10^{-10})}$ & $2.14\!\times\!10^{-4}\;(8.77\!\times\!10^{-5})$ \\
\hline
\multicolumn{7}{c}{SOR(1.5)-related solvers} \\
\hline
Solver Only & $4.83\!\times\!10^{-2}\;(1.51\!\times\!10^{-2})$ & $6.89\!\times\!10^{-2}\;(2.16\!\times\!10^{-2})$ & $8.09\!\times\!10^{-7}\;(2.80\!\times\!10^{-7})$ & $\mathbf{7.74\!\times\!10^{-4}\;(3.18\!\times\!10^{-4})}$ & $4.88\!\times\!10^{-8}\;(1.79\!\times\!10^{-8})$ & $\mathbf{7.21\!\times\!10^{-5}\;(2.93\!\times\!10^{-5})}$ \\
HINTS & $3.51\!\times\!10^{-7}\;(1.23\!\times\!10^{-7})$ & $3.80\!\times\!10^{-2}\;(1.19\!\times\!10^{-2})$ & $3.18\!\times\!10^{-8}\;(1.27\!\times\!10^{-8})$ & $9.07\!\times\!10^{-4}\;(3.73\!\times\!10^{-4})$ & $1.54\!\times\!10^{-9}\;(5.71\!\times\!10^{-10})$ & $7.96\!\times\!10^{-5}\;(3.23\!\times\!10^{-5})$ \\
Learned Greedy & $\mathbf{1.66\!\times\!10^{-12}\;(6.19\!\times\!10^{-13})}$ & $\mathbf{2.40\!\times\!10^{-2}\;(7.50\!\times\!10^{-3})}$ & $\mathbf{2.09\!\times\!10^{-12}\;(9.18\!\times\!10^{-13})}$ & $1.57\!\times\!10^{-3}\;(6.38\!\times\!10^{-4})$ & $\mathbf{1.24\!\times\!10^{-12}\;(5.34\!\times\!10^{-13})}$ & $2.18\!\times\!10^{-4}\;(8.93\!\times\!10^{-5})$ \\
\hline
\multicolumn{7}{c}{MG-related solvers} \\
\hline
Solver Only & $1.02\!\times\!10^{-10}\;(3.91\!\times\!10^{-11})$ & $\mathbf{1.79\!\times\!10^{-2}\;(5.61\!\times\!10^{-3})}$ & $7.59\!\times\!10^{-9}\;(3.12\!\times\!10^{-9})$ & $\mathbf{1.24\!\times\!10^{-3}\;(5.06\!\times\!10^{-4})}$ & $9.37\!\times\!10^{-10}\;(3.48\!\times\!10^{-10})$ & $\mathbf{1.72\!\times\!10^{-4}\;(7.05\!\times\!10^{-5})}$ \\
HINTS & $\mathbf{2.58\!\times\!10^{-14}\;(1.36\!\times\!10^{-14})}$ & $1.79\!\times\!10^{-2}\;(5.61\!\times\!10^{-3})$ & $\mathbf{2.72\!\times\!10^{-14}\;(1.02\!\times\!10^{-14})}$ & $1.24\!\times\!10^{-3}\;(5.06\!\times\!10^{-4})$ & $\mathbf{2.35\!\times\!10^{-14}\;(1.08\!\times\!10^{-14})}$ & $1.72\!\times\!10^{-4}\;(7.05\!\times\!10^{-5})$ \\
Learned Greedy & $1.40\!\times\!10^{-11}\;(5.40\!\times\!10^{-12})$ & $2.35\!\times\!10^{-2}\;(7.35\!\times\!10^{-3})$ & $1.54\!\times\!10^{-10}\;(6.13\!\times\!10^{-11})$ & $1.53\!\times\!10^{-3}\;(6.25\!\times\!10^{-4})$ & $1.17\!\times\!10^{-10}\;(4.41\!\times\!10^{-11})$ & $2.13\!\times\!10^{-4}\;(8.74\!\times\!10^{-5})$ \\
\hline
\end{tabular}}
\end{table}

\begin{table}[h]
\centering
\caption{Fourier-mode error at a 25 ms wall-clock budget for Anisotropic diffusion on 64 shared test instances. Each entry is mean (standard error) of $|\widehat{e}_{k,k}|$, or its time-averaged wall-clock AUC, using NumPy's unnormalized DFT of the demeaned solution error. Lower is better. Bold marks the lowest deployable method within each solver family and column.}
\label{tab:red-fourier-modes-anisodiff}
\resizebox{\linewidth}{!}{%
\begin{tabular}{lcccccc}
\hline
& \multicolumn{2}{c}{Mode $(1,1)$} & \multicolumn{2}{c}{Mode $(5,5)$} & \multicolumn{2}{c}{Mode $(10,10)$} \\
Method & Error & AUC & Error & AUC & Error & AUC \\
\hline
\multicolumn{7}{c}{Jacobi-related solvers} \\
\hline
Solver Only & $5.19\!\times\!10^{-1}\;(1.63\!\times\!10^{-1})$ & $5.68\!\times\!10^{-1}\;(1.78\!\times\!10^{-1})$ & $1.74\!\times\!10^{-4}\;(7.08\!\times\!10^{-5})$ & $3.32\!\times\!10^{-3}\;(1.35\!\times\!10^{-3})$ & $2.15\!\times\!10^{-11}\;(8.82\!\times\!10^{-12})$ & $\mathbf{1.13\!\times\!10^{-4}\;(4.65\!\times\!10^{-5})}$ \\
HINTS & $1.30\!\times\!10^{-6}\;(4.32\!\times\!10^{-7})$ & $3.34\!\times\!10^{-1}\;(1.05\!\times\!10^{-1})$ & $1.52\!\times\!10^{-7}\;(3.90\!\times\!10^{-8})$ & $3.34\!\times\!10^{-3}\;(1.36\!\times\!10^{-3})$ & $1.59\!\times\!10^{-9}\;(4.52\!\times\!10^{-10})$ & $1.14\!\times\!10^{-4}\;(4.67\!\times\!10^{-5})$ \\
Learned Greedy & $\mathbf{6.83\!\times\!10^{-7}\;(1.80\!\times\!10^{-7})}$ & $\mathbf{1.41\!\times\!10^{-1}\;(4.40\!\times\!10^{-2})}$ & $\mathbf{3.83\!\times\!10^{-8}\;(1.01\!\times\!10^{-8})}$ & $\mathbf{3.25\!\times\!10^{-3}\;(1.33\!\times\!10^{-3})}$ & $\mathbf{4.10\!\times\!10^{-12}\;(1.35\!\times\!10^{-12})}$ & $3.99\!\times\!10^{-4}\;(1.64\!\times\!10^{-4})$ \\
\hline
\multicolumn{7}{c}{GS-related solvers} \\
\hline
Solver Only & $5.46\!\times\!10^{-1}\;(1.71\!\times\!10^{-1})$ & $5.82\!\times\!10^{-1}\;(1.82\!\times\!10^{-1})$ & $7.71\!\times\!10^{-4}\;(3.21\!\times\!10^{-4})$ & $4.75\!\times\!10^{-3}\;(1.94\!\times\!10^{-3})$ & $1.40\!\times\!10^{-5}\;(6.27\!\times\!10^{-6})$ & $2.04\!\times\!10^{-4}\;(8.33\!\times\!10^{-5})$ \\
HINTS & $2.05\!\times\!10^{-5}\;(6.51\!\times\!10^{-6})$ & $4.03\!\times\!10^{-1}\;(1.26\!\times\!10^{-1})$ & $4.15\!\times\!10^{-6}\;(1.73\!\times\!10^{-6})$ & $4.54\!\times\!10^{-3}\;(1.85\!\times\!10^{-3})$ & $2.16\!\times\!10^{-7}\;(9.42\!\times\!10^{-8})$ & $\mathbf{2.01\!\times\!10^{-4}\;(8.19\!\times\!10^{-5})}$ \\
Learned Greedy & $\mathbf{7.85\!\times\!10^{-8}\;(2.76\!\times\!10^{-8})}$ & $\mathbf{1.39\!\times\!10^{-1}\;(4.35\!\times\!10^{-2})}$ & $\mathbf{4.87\!\times\!10^{-8}\;(1.87\!\times\!10^{-8})}$ & $\mathbf{3.22\!\times\!10^{-3}\;(1.31\!\times\!10^{-3})}$ & $\mathbf{9.79\!\times\!10^{-9}\;(3.86\!\times\!10^{-9})}$ & $3.97\!\times\!10^{-4}\;(1.63\!\times\!10^{-4})$ \\
\hline
\multicolumn{7}{c}{SymGS-related solvers} \\
\hline
Solver Only & $5.34\!\times\!10^{-1}\;(1.67\!\times\!10^{-1})$ & $5.77\!\times\!10^{-1}\;(1.80\!\times\!10^{-1})$ & $5.31\!\times\!10^{-4}\;(2.01\!\times\!10^{-4})$ & $4.24\!\times\!10^{-3}\;(1.72\!\times\!10^{-3})$ & $4.81\!\times\!10^{-5}\;(2.05\!\times\!10^{-5})$ & $\mathbf{2.33\!\times\!10^{-4}\;(9.39\!\times\!10^{-5})}$ \\
HINTS & $5.96\!\times\!10^{-7}\;(1.95\!\times\!10^{-7})$ & $2.16\!\times\!10^{-1}\;(6.74\!\times\!10^{-2})$ & $2.42\!\times\!10^{-7}\;(9.85\!\times\!10^{-8})$ & $3.65\!\times\!10^{-3}\;(1.49\!\times\!10^{-3})$ & $6.43\!\times\!10^{-8}\;(3.19\!\times\!10^{-8})$ & $2.42\!\times\!10^{-4}\;(9.91\!\times\!10^{-5})$ \\
Learned Greedy & $\mathbf{1.73\!\times\!10^{-7}\;(6.85\!\times\!10^{-8})}$ & $\mathbf{1.41\!\times\!10^{-1}\;(4.40\!\times\!10^{-2})}$ & $\mathbf{6.89\!\times\!10^{-8}\;(2.81\!\times\!10^{-8})}$ & $\mathbf{3.26\!\times\!10^{-3}\;(1.33\!\times\!10^{-3})}$ & $\mathbf{1.52\!\times\!10^{-8}\;(5.54\!\times\!10^{-9})}$ & $4.01\!\times\!10^{-4}\;(1.65\!\times\!10^{-4})$ \\
\hline
\multicolumn{7}{c}{SOR(1.5)-related solvers} \\
\hline
Solver Only & $4.22\!\times\!10^{-1}\;(1.32\!\times\!10^{-1})$ & $5.16\!\times\!10^{-1}\;(1.62\!\times\!10^{-1})$ & $9.91\!\times\!10^{-5}\;(3.78\!\times\!10^{-5})$ & $\mathbf{2.46\!\times\!10^{-3}\;(9.97\!\times\!10^{-4})}$ & $2.45\!\times\!10^{-5}\;(9.88\!\times\!10^{-6})$ & $\mathbf{1.88\!\times\!10^{-4}\;(7.46\!\times\!10^{-5})}$ \\
HINTS & $1.06\!\times\!10^{-6}\;(3.40\!\times\!10^{-7})$ & $2.30\!\times\!10^{-1}\;(7.19\!\times\!10^{-2})$ & $1.55\!\times\!10^{-7}\;(6.40\!\times\!10^{-8})$ & $2.74\!\times\!10^{-3}\;(1.12\!\times\!10^{-3})$ & $\mathbf{2.30\!\times\!10^{-8}\;(8.58\!\times\!10^{-9})}$ & $1.97\!\times\!10^{-4}\;(7.82\!\times\!10^{-5})$ \\
Learned Greedy & $\mathbf{4.08\!\times\!10^{-7}\;(1.56\!\times\!10^{-7})}$ & $\mathbf{1.39\!\times\!10^{-1}\;(4.35\!\times\!10^{-2})}$ & $\mathbf{9.04\!\times\!10^{-8}\;(3.37\!\times\!10^{-8})}$ & $3.22\!\times\!10^{-3}\;(1.31\!\times\!10^{-3})$ & $3.06\!\times\!10^{-8}\;(1.17\!\times\!10^{-8})$ & $3.98\!\times\!10^{-4}\;(1.63\!\times\!10^{-4})$ \\
\hline
\multicolumn{7}{c}{MG-related solvers} \\
\hline
Solver Only & $3.33\!\times\!10^{-5}\;(1.05\!\times\!10^{-5})$ & $9.49\!\times\!10^{-2}\;(2.97\!\times\!10^{-2})$ & $7.51\!\times\!10^{-6}\;(3.27\!\times\!10^{-6})$ & $2.40\!\times\!10^{-3}\;(9.76\!\times\!10^{-4})$ & $3.88\!\times\!10^{-6}\;(1.46\!\times\!10^{-6})$ & $2.92\!\times\!10^{-4}\;(1.19\!\times\!10^{-4})$ \\
HINTS & $1.18\!\times\!10^{-7}\;(2.83\!\times\!10^{-8})$ & $\mathbf{9.49\!\times\!10^{-2}\;(2.97\!\times\!10^{-2})}$ & $9.05\!\times\!10^{-8}\;(3.54\!\times\!10^{-8})$ & $\mathbf{2.40\!\times\!10^{-3}\;(9.76\!\times\!10^{-4})}$ & $1.50\!\times\!10^{-7}\;(5.93\!\times\!10^{-8})$ & $\mathbf{2.92\!\times\!10^{-4}\;(1.19\!\times\!10^{-4})}$ \\
Learned Greedy & $\mathbf{1.24\!\times\!10^{-10}\;(4.55\!\times\!10^{-11})}$ & $1.39\!\times\!10^{-1}\;(4.34\!\times\!10^{-2})$ & $\mathbf{1.15\!\times\!10^{-10}\;(3.71\!\times\!10^{-11})}$ & $3.21\!\times\!10^{-3}\;(1.31\!\times\!10^{-3})$ & $\mathbf{1.49\!\times\!10^{-10}\;(4.92\!\times\!10^{-11})}$ & $3.94\!\times\!10^{-4}\;(1.62\!\times\!10^{-4})$ \\
\hline
\end{tabular}}
\end{table}

\begin{table}[h]
\centering
\caption{Fourier-mode error at a 25 ms wall-clock budget for Poisson growing solver ensembles on 64 shared test instances. Each entry is mean (standard error) of $|\widehat{e}_{k,k}|$, or its time-averaged wall-clock AUC, using NumPy's unnormalized DFT of the demeaned solution error. Lower is better. Each ensemble router is compared with its validation-selected best constituent pairwise router; bold marks the lower method within each ensemble size and column.}
\label{tab:fourier-modes-ensemble-poisson}
\resizebox{\linewidth}{!}{%
\begin{tabular}{llcccccc}
\hline
Actions & Method & \multicolumn{2}{c}{Mode $(1,1)$} & \multicolumn{2}{c}{Mode $(5,5)$} & \multicolumn{2}{c}{Mode $(10,10)$} \\
& & Error & AUC & Error & AUC & Error & AUC \\
\hline
\multirow{2}{*}{$\mathcal{A}_3$} & Best pairwise router & $5.69\!\times\!10^{-15}\;(2.58\!\times\!10^{-15})$ & $0.012\;(5.41\!\times\!10^{-3})$ & $1.05\!\times\!10^{-15}\;(4.13\!\times\!10^{-16})$ & $1.52\!\times\!10^{-4}\;(7.32\!\times\!10^{-5})$ & $1.75\!\times\!10^{-16}\;(7.88\!\times\!10^{-17})$ & $5.19\!\times\!10^{-5}\;(2.77\!\times\!10^{-5})$ \\
& Router$(\mathcal{A}_3)$ & $\mathbf{2.36\!\times\!10^{-15}\;(1.04\!\times\!10^{-15})}$ & $\mathbf{0.012\;(5.40\!\times\!10^{-3})}$ & $\mathbf{7.05\!\times\!10^{-18}\;(2.88\!\times\!10^{-18})}$ & $\mathbf{1.51\!\times\!10^{-4}\;(7.29\!\times\!10^{-5})}$ & $\mathbf{2.94\!\times\!10^{-18}\;(1.22\!\times\!10^{-18})}$ & $\mathbf{5.19\!\times\!10^{-5}\;(2.78\!\times\!10^{-5})}$ \\
\hline
\multirow{2}{*}{$\mathcal{A}_4$} & Best pairwise router & $\mathbf{5.69\!\times\!10^{-15}\;(2.58\!\times\!10^{-15})}$ & $\mathbf{0.012\;(5.41\!\times\!10^{-3})}$ & $1.05\!\times\!10^{-15}\;(4.13\!\times\!10^{-16})$ & $1.52\!\times\!10^{-4}\;(7.32\!\times\!10^{-5})$ & $1.75\!\times\!10^{-16}\;(7.88\!\times\!10^{-17})$ & $\mathbf{5.19\!\times\!10^{-5}\;(2.77\!\times\!10^{-5})}$ \\
& Router$(\mathcal{A}_4)$ & $5.98\!\times\!10^{-15}\;(2.94\!\times\!10^{-15})$ & $0.012\;(5.41\!\times\!10^{-3})$ & $\mathbf{4.01\!\times\!10^{-17}\;(1.85\!\times\!10^{-17})}$ & $\mathbf{1.52\!\times\!10^{-4}\;(7.31\!\times\!10^{-5})}$ & $\mathbf{3.16\!\times\!10^{-18}\;(1.51\!\times\!10^{-18})}$ & $5.21\!\times\!10^{-5}\;(2.78\!\times\!10^{-5})$ \\
\hline
\multirow{2}{*}{$\mathcal{A}_5$} & Best pairwise router & $5.69\!\times\!10^{-15}\;(2.58\!\times\!10^{-15})$ & $\mathbf{0.012\;(5.41\!\times\!10^{-3})}$ & $1.05\!\times\!10^{-15}\;(4.13\!\times\!10^{-16})$ & $\mathbf{1.52\!\times\!10^{-4}\;(7.32\!\times\!10^{-5})}$ & $1.75\!\times\!10^{-16}\;(7.88\!\times\!10^{-17})$ & $\mathbf{5.19\!\times\!10^{-5}\;(2.77\!\times\!10^{-5})}$ \\
& Router$(\mathcal{A}_5)$ & $\mathbf{2.23\!\times\!10^{-15}\;(8.87\!\times\!10^{-16})}$ & $0.012\;(5.46\!\times\!10^{-3})$ & $\mathbf{1.37\!\times\!10^{-17}\;(7.59\!\times\!10^{-18})}$ & $1.52\!\times\!10^{-4}\;(7.38\!\times\!10^{-5})$ & $\mathbf{3.70\!\times\!10^{-18}\;(1.70\!\times\!10^{-18})}$ & $5.24\!\times\!10^{-5}\;(2.80\!\times\!10^{-5})$ \\
\hline
\multirow{2}{*}{$\mathcal{A}_6$} & Best pairwise router & $5.69\!\times\!10^{-15}\;(2.58\!\times\!10^{-15})$ & $\mathbf{0.012\;(5.41\!\times\!10^{-3})}$ & $1.05\!\times\!10^{-15}\;(4.13\!\times\!10^{-16})$ & $\mathbf{1.52\!\times\!10^{-4}\;(7.32\!\times\!10^{-5})}$ & $1.75\!\times\!10^{-16}\;(7.88\!\times\!10^{-17})$ & $\mathbf{5.19\!\times\!10^{-5}\;(2.77\!\times\!10^{-5})}$ \\
& Router$(\mathcal{A}_6)$ & $\mathbf{2.90\!\times\!10^{-15}\;(1.37\!\times\!10^{-15})}$ & $0.012\;(5.42\!\times\!10^{-3})$ & $\mathbf{2.37\!\times\!10^{-17}\;(1.22\!\times\!10^{-17})}$ & $1.52\!\times\!10^{-4}\;(7.33\!\times\!10^{-5})$ & $\mathbf{8.21\!\times\!10^{-17}\;(4.83\!\times\!10^{-17})}$ & $5.21\!\times\!10^{-5}\;(2.78\!\times\!10^{-5})$ \\
\hline
\end{tabular}}
\end{table}

\begin{table}[h]
\centering
\caption{Fourier-mode error at a 25 ms wall-clock budget for convection--diffusion growing solver ensembles on 64 shared test instances. Each entry is mean (standard error) of $|\widehat{e}_{k,k}|$, or its time-averaged wall-clock AUC, using NumPy's unnormalized DFT of the demeaned solution error. Lower is better. Each ensemble router is compared with its validation-selected best constituent pairwise router; bold marks the lower method within each ensemble size and column.}
\label{tab:fourier-modes-ensemble-convdiff}
\resizebox{\linewidth}{!}{%
\begin{tabular}{llcccccc}
\hline
Actions & Method & \multicolumn{2}{c}{Mode $(1,1)$} & \multicolumn{2}{c}{Mode $(5,5)$} & \multicolumn{2}{c}{Mode $(10,10)$} \\
& & Error & AUC & Error & AUC & Error & AUC \\
\hline
\multirow{2}{*}{$\mathcal{A}_3$} & Best pairwise router & $4.47\!\times\!10^{-9}\;(1.73\!\times\!10^{-9})$ & $\mathbf{1.97\!\times\!10^{-3}\;(5.03\!\times\!10^{-4})}$ & $1.14\!\times\!10^{-9}\;(6.59\!\times\!10^{-10})$ & $\mathbf{1.45\!\times\!10^{-4}\;(7.93\!\times\!10^{-5})}$ & $3.76\!\times\!10^{-12}\;(1.66\!\times\!10^{-12})$ & $\mathbf{2.71\!\times\!10^{-5}\;(1.45\!\times\!10^{-5})}$ \\
& Router$(\mathcal{A}_3)$ & $\mathbf{6.04\!\times\!10^{-16}\;(2.38\!\times\!10^{-16})}$ & $1.98\!\times\!10^{-3}\;(5.05\!\times\!10^{-4})$ & $\mathbf{2.64\!\times\!10^{-17}\;(1.28\!\times\!10^{-17})}$ & $1.47\!\times\!10^{-4}\;(8.03\!\times\!10^{-5})$ & $\mathbf{1.12\!\times\!10^{-18}\;(5.09\!\times\!10^{-19})}$ & $2.73\!\times\!10^{-5}\;(1.47\!\times\!10^{-5})$ \\
\hline
\multirow{2}{*}{$\mathcal{A}_4$} & Best pairwise router & $1.64\!\times\!10^{-10}\;(8.82\!\times\!10^{-11})$ & $\mathbf{1.97\!\times\!10^{-3}\;(5.00\!\times\!10^{-4})}$ & $9.04\!\times\!10^{-11}\;(4.16\!\times\!10^{-11})$ & $1.47\!\times\!10^{-4}\;(7.97\!\times\!10^{-5})$ & $2.49\!\times\!10^{-11}\;(1.48\!\times\!10^{-11})$ & $2.75\!\times\!10^{-5}\;(1.46\!\times\!10^{-5})$ \\
& Router$(\mathcal{A}_4)$ & $\mathbf{6.02\!\times\!10^{-16}\;(2.25\!\times\!10^{-16})}$ & $1.97\!\times\!10^{-3}\;(5.03\!\times\!10^{-4})$ & $\mathbf{3.78\!\times\!10^{-17}\;(1.82\!\times\!10^{-17})}$ & $\mathbf{1.45\!\times\!10^{-4}\;(7.89\!\times\!10^{-5})}$ & $\mathbf{7.55\!\times\!10^{-19}\;(2.03\!\times\!10^{-19})}$ & $\mathbf{2.71\!\times\!10^{-5}\;(1.44\!\times\!10^{-5})}$ \\
\hline
\multirow{2}{*}{$\mathcal{A}_5$} & Best pairwise router & $7.28\!\times\!10^{-10}\;(3.57\!\times\!10^{-10})$ & $1.97\!\times\!10^{-3}\;(5.04\!\times\!10^{-4})$ & $6.28\!\times\!10^{-11}\;(3.35\!\times\!10^{-11})$ & $1.47\!\times\!10^{-4}\;(7.99\!\times\!10^{-5})$ & $5.99\!\times\!10^{-11}\;(3.55\!\times\!10^{-11})$ & $2.79\!\times\!10^{-5}\;(1.48\!\times\!10^{-5})$ \\
& Router$(\mathcal{A}_5)$ & $\mathbf{6.47\!\times\!10^{-16}\;(2.39\!\times\!10^{-16})}$ & $\mathbf{1.96\!\times\!10^{-3}\;(5.01\!\times\!10^{-4})}$ & $\mathbf{2.74\!\times\!10^{-17}\;(1.28\!\times\!10^{-17})}$ & $\mathbf{1.46\!\times\!10^{-4}\;(7.96\!\times\!10^{-5})}$ & $\mathbf{1.10\!\times\!10^{-18}\;(5.04\!\times\!10^{-19})}$ & $\mathbf{2.71\!\times\!10^{-5}\;(1.45\!\times\!10^{-5})}$ \\
\hline
\multirow{2}{*}{$\mathcal{A}_6$} & Best pairwise router & $7.28\!\times\!10^{-10}\;(3.57\!\times\!10^{-10})$ & $1.97\!\times\!10^{-3}\;(5.04\!\times\!10^{-4})$ & $6.28\!\times\!10^{-11}\;(3.35\!\times\!10^{-11})$ & $1.47\!\times\!10^{-4}\;(7.99\!\times\!10^{-5})$ & $5.99\!\times\!10^{-11}\;(3.55\!\times\!10^{-11})$ & $2.79\!\times\!10^{-5}\;(1.48\!\times\!10^{-5})$ \\
& Router$(\mathcal{A}_6)$ & $\mathbf{1.23\!\times\!10^{-16}\;(3.07\!\times\!10^{-17})}$ & $\mathbf{1.96\!\times\!10^{-3}\;(5.02\!\times\!10^{-4})}$ & $\mathbf{8.18\!\times\!10^{-18}\;(2.86\!\times\!10^{-18})}$ & $\mathbf{1.45\!\times\!10^{-4}\;(7.94\!\times\!10^{-5})}$ & $\mathbf{3.02\!\times\!10^{-18}\;(1.22\!\times\!10^{-18})}$ & $\mathbf{2.71\!\times\!10^{-5}\;(1.45\!\times\!10^{-5})}$ \\
\hline
\end{tabular}}
\end{table}

\begin{table}[h]
\centering
\caption{Fourier-mode error at a 25 ms wall-clock budget for anisotropic diffusion growing solver ensembles on 64 shared test instances. Each entry is mean (standard error) of $|\widehat{e}_{k,k}|$, or its time-averaged wall-clock AUC, using NumPy's unnormalized DFT of the demeaned solution error. Lower is better. Each ensemble router is compared with its validation-selected best constituent pairwise router; bold marks the lower method within each ensemble size and column.}
\label{tab:fourier-modes-ensemble-anisodiff}
\resizebox{\linewidth}{!}{%
\begin{tabular}{llcccccc}
\hline
Actions & Method & \multicolumn{2}{c}{Mode $(1,1)$} & \multicolumn{2}{c}{Mode $(5,5)$} & \multicolumn{2}{c}{Mode $(10,10)$} \\
& & Error & AUC & Error & AUC & Error & AUC \\
\hline
\multirow{2}{*}{$\mathcal{A}_3$} & Best pairwise router & $\mathbf{2.56\!\times\!10^{-15}\;(6.66\!\times\!10^{-16})}$ & $\mathbf{0.047\;(0.012)}$ & $1.26\!\times\!10^{-12}\;(6.90\!\times\!10^{-13})$ & $2.42\!\times\!10^{-3}\;(1.32\!\times\!10^{-3})$ & $\mathbf{2.25\!\times\!10^{-17}\;(6.13\!\times\!10^{-18})}$ & $3.99\!\times\!10^{-4}\;(2.13\!\times\!10^{-4})$ \\
& Router$(\mathcal{A}_3)$ & $1.90\!\times\!10^{-14}\;(9.20\!\times\!10^{-15})$ & $0.047\;(0.012)$ & $\mathbf{1.15\!\times\!10^{-16}\;(5.16\!\times\!10^{-17})}$ & $\mathbf{2.30\!\times\!10^{-3}\;(1.32\!\times\!10^{-3})}$ & $2.35\!\times\!10^{-17}\;(6.25\!\times\!10^{-18})$ & $\mathbf{3.84\!\times\!10^{-4}\;(2.14\!\times\!10^{-4})}$ \\
\hline
\multirow{2}{*}{$\mathcal{A}_4$} & Best pairwise router & $2.69\!\times\!10^{-15}\;(7.08\!\times\!10^{-16})$ & $\mathbf{0.047\;(0.012)}$ & $1.56\!\times\!10^{-11}\;(6.59\!\times\!10^{-12})$ & $2.42\!\times\!10^{-3}\;(1.32\!\times\!10^{-3})$ & $1.51\!\times\!10^{-11}\;(7.02\!\times\!10^{-12})$ & $4.00\!\times\!10^{-4}\;(2.14\!\times\!10^{-4})$ \\
& Router$(\mathcal{A}_4)$ & $\mathbf{2.48\!\times\!10^{-15}\;(6.38\!\times\!10^{-16})}$ & $0.047\;(0.012)$ & $\mathbf{1.88\!\times\!10^{-12}\;(1.09\!\times\!10^{-12})}$ & $\mathbf{2.30\!\times\!10^{-3}\;(1.33\!\times\!10^{-3})}$ & $\mathbf{2.51\!\times\!10^{-15}\;(9.32\!\times\!10^{-16})}$ & $\mathbf{3.84\!\times\!10^{-4}\;(2.14\!\times\!10^{-4})}$ \\
\hline
\multirow{2}{*}{$\mathcal{A}_5$} & Best pairwise router & $2.69\!\times\!10^{-15}\;(7.08\!\times\!10^{-16})$ & $0.047\;(0.012)$ & $1.56\!\times\!10^{-11}\;(6.59\!\times\!10^{-12})$ & $2.42\!\times\!10^{-3}\;(1.32\!\times\!10^{-3})$ & $1.51\!\times\!10^{-11}\;(7.02\!\times\!10^{-12})$ & $4.00\!\times\!10^{-4}\;(2.14\!\times\!10^{-4})$ \\
& Router$(\mathcal{A}_5)$ & $\mathbf{2.47\!\times\!10^{-15}\;(6.38\!\times\!10^{-16})}$ & $\mathbf{0.047\;(0.012)}$ & $\mathbf{1.83\!\times\!10^{-12}\;(1.09\!\times\!10^{-12})}$ & $\mathbf{2.30\!\times\!10^{-3}\;(1.33\!\times\!10^{-3})}$ & $\mathbf{9.12\!\times\!10^{-16}\;(6.13\!\times\!10^{-16})}$ & $\mathbf{3.84\!\times\!10^{-4}\;(2.14\!\times\!10^{-4})}$ \\
\hline
\multirow{2}{*}{$\mathcal{A}_6$} & Best pairwise router & $2.69\!\times\!10^{-15}\;(7.08\!\times\!10^{-16})$ & $\mathbf{0.047\;(0.012)}$ & $1.56\!\times\!10^{-11}\;(6.59\!\times\!10^{-12})$ & $2.42\!\times\!10^{-3}\;(1.32\!\times\!10^{-3})$ & $1.51\!\times\!10^{-11}\;(7.02\!\times\!10^{-12})$ & $4.00\!\times\!10^{-4}\;(2.14\!\times\!10^{-4})$ \\
& Router$(\mathcal{A}_6)$ & $\mathbf{2.48\!\times\!10^{-15}\;(6.38\!\times\!10^{-16})}$ & $0.047\;(0.012)$ & $\mathbf{1.84\!\times\!10^{-12}\;(1.09\!\times\!10^{-12})}$ & $\mathbf{2.30\!\times\!10^{-3}\;(1.32\!\times\!10^{-3})}$ & $\mathbf{1.05\!\times\!10^{-15}\;(6.26\!\times\!10^{-16})}$ & $\mathbf{3.84\!\times\!10^{-4}\;(2.13\!\times\!10^{-4})}$ \\
\hline
\end{tabular}}
\end{table}

\clearpage
\subsection{Ensemble Solver Decisions and DeepONet Usage} \label{sec:solverdecisions}

\begin{figure}
    \centering
    \includegraphics[width=0.5\linewidth]{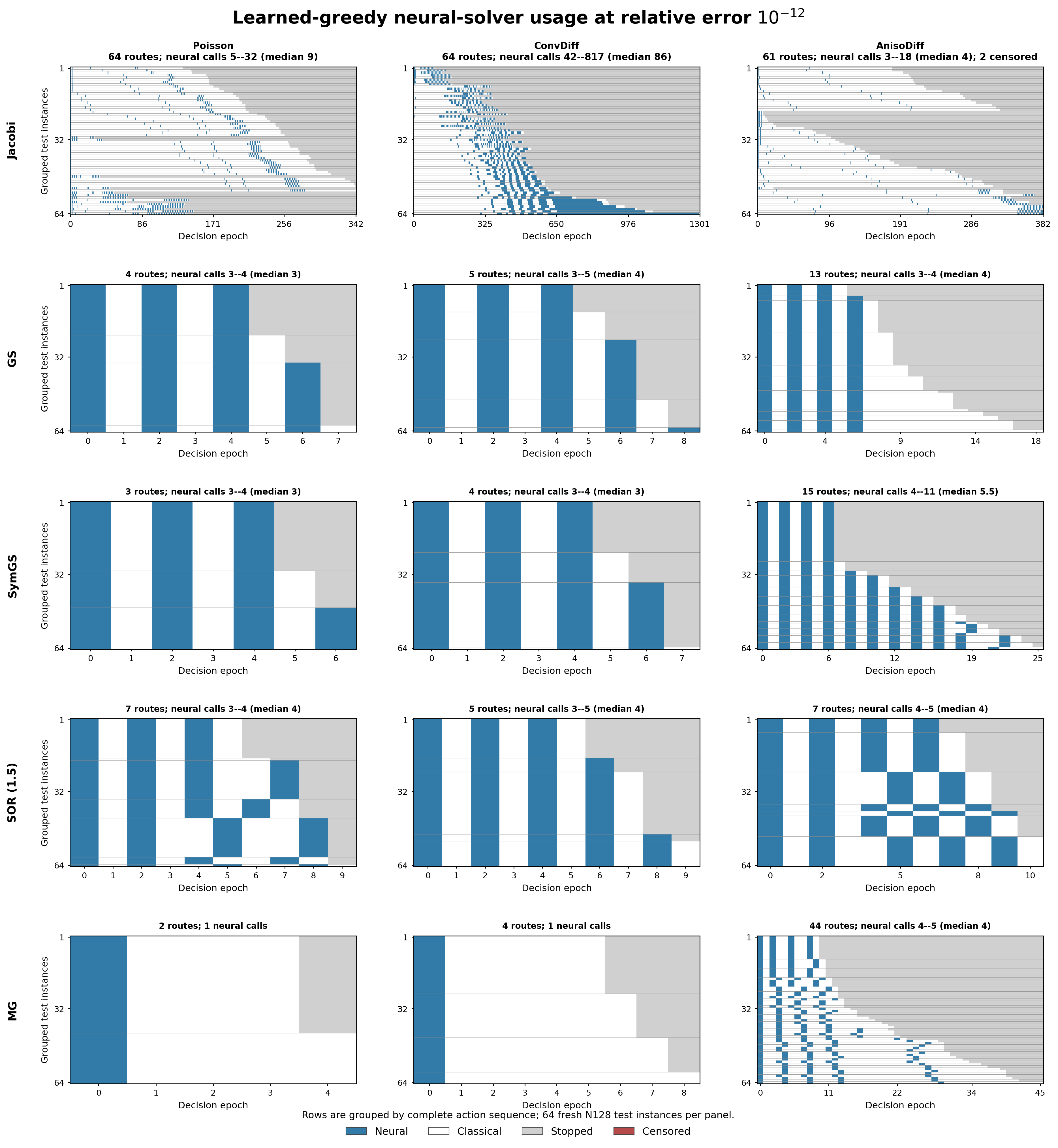}
    \caption{Learned-router solver usage for paired classical--neural solvers across 64 test instances at a relative solution-error tolerance of $10^{-12}$. Each panel shows the selected action at each decision epoch, with instances grouped by their complete action sequence. Gray indicates that the tolerance has been reached, while red denotes censored instances that do not reach the tolerance within the evaluation horizon.}
    \label{fig:complete-solver-usage}
\end{figure}

\Cref{fig:complete-solver-usage} visualize the routing decisions of the learned greedy policies across all test samples and solver families.

This figure suggests that, in some solver families, the routing strategy is highly sample-dependent, indicating that a fixed routing schedule across all instances would be suboptimal and can slow convergence. This is especially true in cases where the solver requires many decision epochs to achieve the desired error tolerance.


\begin{figure}
    \centering
    \includegraphics[width=0.5\linewidth]{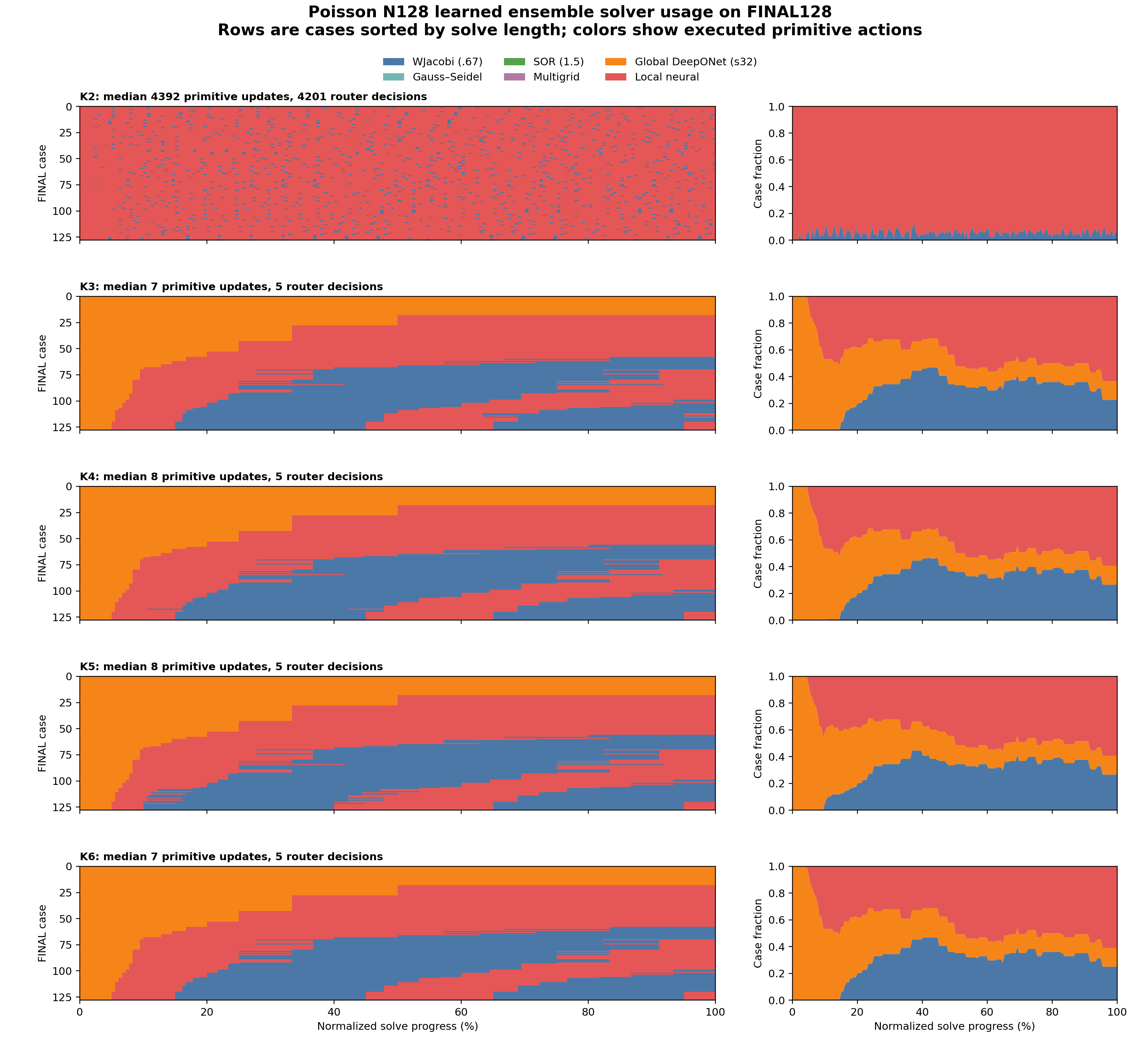}
    \caption{Solver usage by the learned router for growing Poisson solver ensembles on 128 test instances. Rows correspond to instances sorted by solve length, and colors indicate the executed primitive actions over normalized solve progress. The right column reports the fraction of instances using each action at each point in the solve.}
    \label{fig:solver-usage-ensemble-possion}
\end{figure}

\begin{figure}
    \centering
    \includegraphics[width=0.5\linewidth]{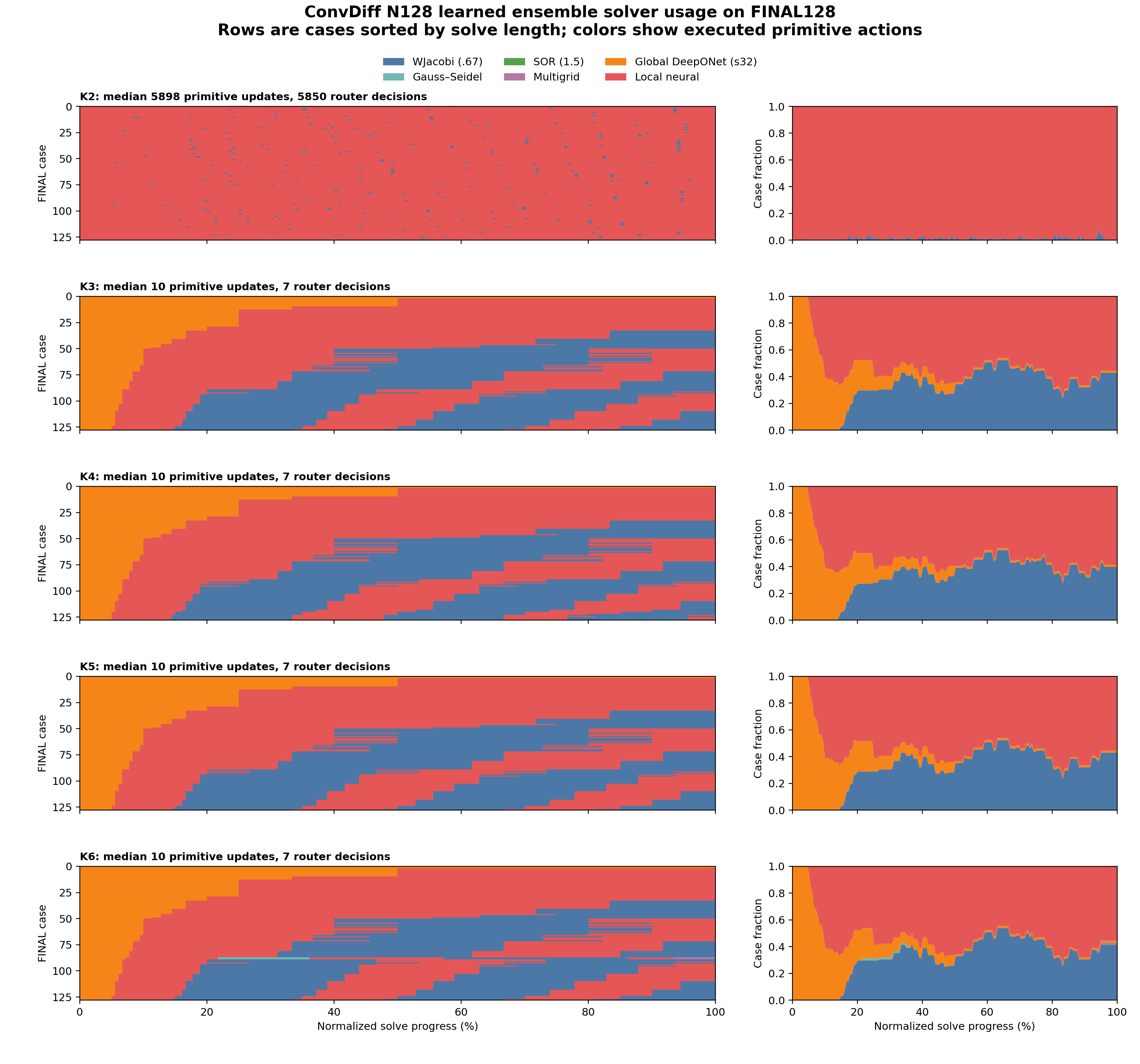}
    \caption{Solver usage by the learned router for growing convection-diffusion solver ensembles on 128 test instances. Rows correspond to instances sorted by solve length, and colors indicate the executed primitive actions over normalized solve progress. The right column reports the fraction of instances using each action at each point in the solve.}
    \label{fig:solver-usage-ensemble-convdiff}
\end{figure}

\begin{figure}
    \centering
    \includegraphics[width=0.5\linewidth]{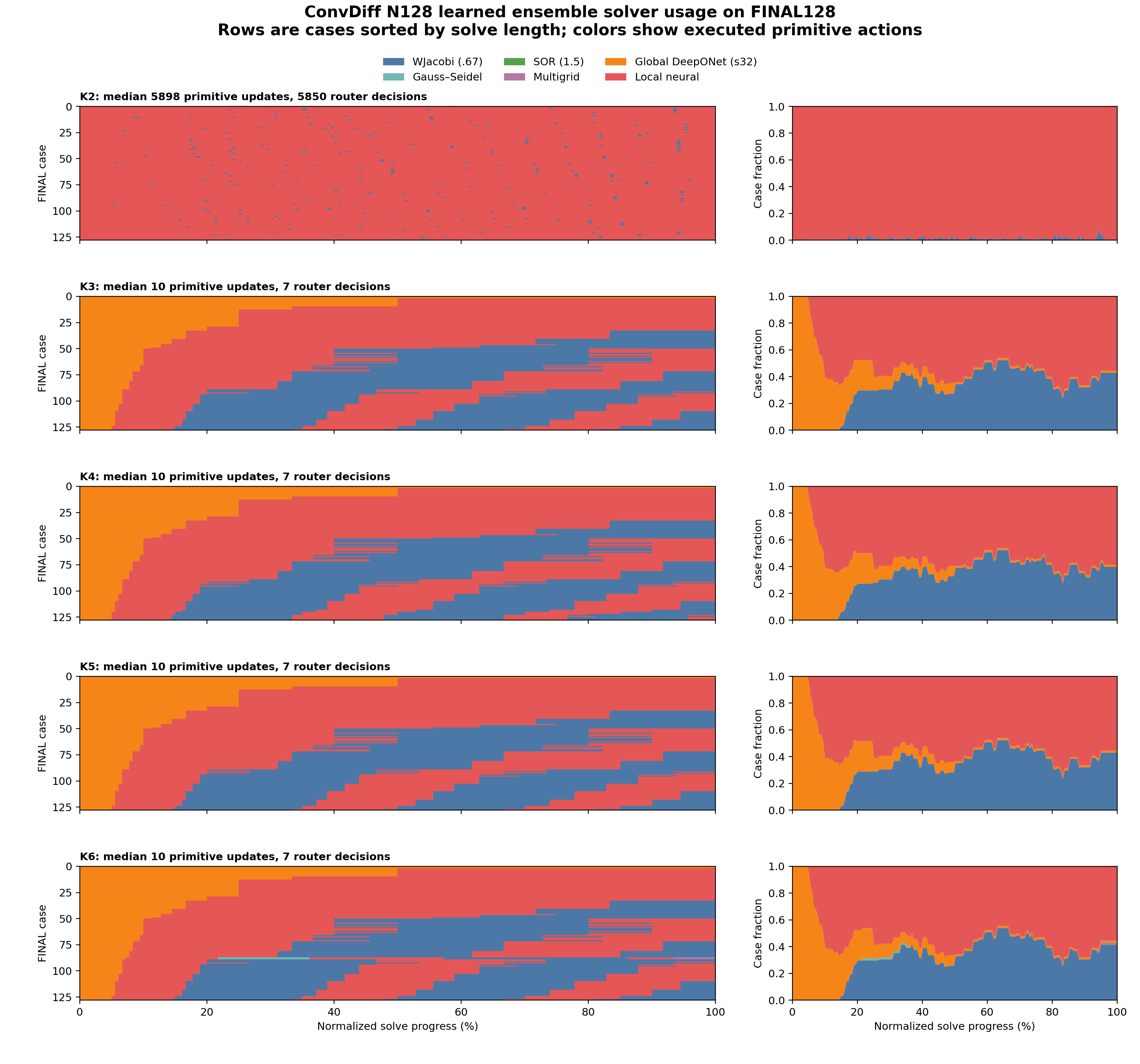}
    \caption{Solver usage by the learned router for growing anisotropic diffusion solver ensembles on 128 test instances. Rows correspond to instances sorted by solve length, and colors indicate the executed primitive actions over normalized solve progress. The right column reports the fraction of instances using each action at each point in the solve.}
    \label{fig:solver-usage-ensemble-aniso}
\end{figure}

\Cref{fig:solver-usage-ensemble-possion,fig:solver-usage-ensemble-convdiff,fig:solver-usage-ensemble-aniso} report solver selection patterns across different solver ensembles. The non-uniformity across samples is more apparent in these experiments. As expected, there seems to be a heavy reliance on 3 solvers, which explains the diminishing returns observed in wall-clock time with ensemble sizes larger than 3.

\subsection{Results on a finer grid} \label{sec:finer_grids}

We have replicated the Paired Solver experiments for a grid of $512 \times 512$. 

In \Cref{tab:wallclock_fine}, we notice that Jacobi is anomalous at \(512\times512\), with incomplete coverage for both HINTS and the learned router, whereas across GS, SymGS, SOR, and MG, the learned router consistently achieves the lowest deployable wall-clock time. Moreover, most standalone iterative solvers generally fail to reach the target accuracy at this resolution, highlighting the increasing advantage of hybrid routing on finer grids


\begin{table}[H]
    \centering
    \caption{Selected-action wall-clock replay time to relative solution error $10^{-6}$ on the $512\times512$ grid (lower is better). Values are mean (standard error) in milliseconds over successful cases among 16 shared test instances; entries without brackets have 16/16 successes. Brackets report incomplete coverage, and a dash indicates no successes. The reported HINTS result uses the best observed cadence in $\{5,10,25,50\}$ for each equation and solver family at this tolerance. One-sided paired $t$-tests compare each full-coverage deployable baseline with the learned router under the alternative that the router is faster; tests are omitted when either method lacks full coverage. Bold marks the fastest deployable method with full coverage, while True-Greedy is an oracle with access to the exact error.}
    \resizebox{\textwidth}{!}{%
    \begin{tabular}{lcccc}
        \hline\hline
        & \multicolumn{2}{c}{Poisson} & \multicolumn{2}{c}{ConvDiff} \\
        \cline{2-5}
        Method & Time (ms) & $p$-value & Time (ms) & $p$-value \\
        \hline
        \multicolumn{5}{c}{Jacobi-related solvers} \\
        \hline
        Jacobi Only & -- & -- & -- & -- \\
        HINTS-Jacobi ($m=50/50$) & 2999.34 (1822.99) [14/16] & -- & 948.34 (302.28) [12/16] & -- \\
        Learned Greedy-Jacobi & 302.78 (71.31) [12/16] & -- & 878.31 (305.26) [12/16] & -- \\
        \textit{True-Greedy-Jacobi} & \textit{2780.05 (1710.63) [14/16]} & -- & \textit{859.51 (293.84) [12/16]} & -- \\
        \hline
        \multicolumn{5}{c}{GS-related solvers} \\
        \hline
        GS Only & -- & -- & -- & -- \\
        HINTS-GS ($m=25/25$) & 639.92 (62.71) & 0.002 & 825.77 (64.83) & 0.003 \\
        Learned Greedy-GS & \textbf{524.84 (91.17)} & -- & \textbf{677.56 (102.12)} & -- \\
        \textit{True-Greedy-GS} & \textit{520.29 (90.21)} & -- & \textit{673.75 (102.22)} & -- \\
        \hline
        \multicolumn{5}{c}{SymGS-related solvers} \\
        \hline
        SymGS Only & -- & -- & -- & -- \\
        HINTS-SymGS ($m=10/10$) & 529.52 (55.57) & $<10^{-3}$ & 649.97 (55.09) & 0.001 \\
        Learned Greedy-SymGS & \textbf{428.67 (73.93)} & -- & \textbf{534.63 (79.04)} & -- \\
        \textit{True-Greedy-SymGS} & \textit{425.87 (73.70)} & -- & \textit{530.39 (78.59)} & -- \\
        \hline
        \multicolumn{5}{c}{SOR (1.5)-related solvers} \\
        \hline
        SOR (1.5) Only & -- & -- & -- & -- \\
        HINTS-SOR (1.5) ($m=10/10$) & 333.37 (31.02) & $<10^{-3}$ & 407.01 (30.25) & $<10^{-3}$ \\
        Learned Greedy-SOR (1.5) & \textbf{259.79 (43.64)} & -- & \textbf{322.55 (46.52)} & -- \\
        \textit{True-Greedy-SOR (1.5)} & \textit{258.26 (43.36)} & -- & \textit{319.34 (46.19)} & -- \\
        \hline
        \multicolumn{5}{c}{MG-related solvers} \\
        \hline
        MG Only & 289.14 (10.25) & $<10^{-3}$ & 344.36 (9.70) & $<10^{-3}$ \\
        HINTS-MG ($m=5/5$) & 262.52 (5.12) & $<10^{-3}$ & 314.69 (2.11) & $<10^{-3}$ \\
        Learned Greedy-MG & \textbf{150.68 (22.59)} & -- & \textbf{199.65 (24.74)} & -- \\
        \textit{True-Greedy-MG} & \textit{149.29 (22.22)} & -- & \textit{199.83 (24.62)} & -- \\
        \hline\hline
    \end{tabular}}
    \label{tab:wallclock_fine}
\end{table}

\subsection{Results with Dirichlet Boundary Conditions} \label{sec:dirichlet}

We have replicated the Paired Solver experiments for solutions with zero dirichlet boundary conditions.

In \Cref{tab:wallclock_dirichlet}, we notice that, the learned router generally provides the lowest deployable wall-clock time for Jacobi, GS, SymGS, and most SOR pairings, often closely tracking the true-greedy oracle. The advantage is not universal: selected HINTS is faster for all MG pairings and for ConvDiff–SOR, indicating that effective fixed schedules can remain preferable when the constituent solver is already strong.

\begin{table}[H]
    \centering
    \caption{Selected-action wall-clock replay time to relative solution error $10^{-6}$ with zero Dirichlet boundary conditions (lower is better). Values are mean (standard error) in milliseconds over successful cases among 64 shared test instances with $128\times128$ interior unknowns; entries without brackets have 64/64 successes. Brackets report incomplete coverage, and a dash indicates no successes. The reported HINTS result uses the best tested cadence in $\{5,10,25,50\}$ for each equation and solver family on these test instances and is therefore an optimistic comparator. One-sided paired $t$-tests compare each full-coverage deployable baseline with the learned router under the alternative that the router is faster; tests are omitted when either method lacks full coverage. Bold marks the fastest deployable method with full coverage, while True-Greedy is an oracle with access to the exact error.}
    \resizebox{\textwidth}{!}{%
    \begin{tabular}{lcccccc}
        \hline\hline
        & \multicolumn{2}{c}{Poisson} & \multicolumn{2}{c}{ConvDiff} & \multicolumn{2}{c}{Aniso} \\
        \cline{2-7}
        Method & Time (ms) & $p$-value & Time (ms) & $p$-value & Time (ms) & $p$-value \\
        \hline
        \multicolumn{7}{c}{Jacobi-related solvers} \\
        \hline
        Jacobi Only & -- & -- & 1116.25 (8.57) & $<10^{-3}$ & -- & -- \\
        HINTS-Jacobi ($m=10/10/50$) & \textbf{599.86 (177.87)} & -- & 583.50 (109.39) & $<10^{-3}$ & 454.80 (78.08) & $<10^{-3}$ \\
        Learned Greedy-Jacobi & 124.03 (37.49) [62/64] & -- & \textbf{148.24 (23.96)} & -- & \textbf{314.83 (52.80)} & -- \\
        \textit{True-Greedy-Jacobi} & \textit{148.39 (40.34)} & -- & \textit{144.40 (23.41)} & -- & \textit{299.83 (49.03)} & -- \\
        \hline
        \multicolumn{7}{c}{GS-related solvers} \\
        \hline
        GS Only & 9160.73 (547.70) [6/64] & -- & 1495.98 (12.02) & $<10^{-3}$ & -- & -- \\
        HINTS-GS ($m=10/10/50$) & 28.42 (0.52) & $<10^{-3}$ & 32.08 (0.16) & $<10^{-3}$ & 109.45 (1.53) & $<10^{-3}$ \\
        Learned Greedy-GS & \textbf{20.80 (0.45)} & -- & \textbf{30.55 (0.42)} & -- & \textbf{90.78 (1.57)} & -- \\
        \textit{True-Greedy-GS} & \textit{20.10 (0.46)} & -- & \textit{30.20 (0.42)} & -- & \textit{88.62 (1.37)} & -- \\
        \hline
        \multicolumn{7}{c}{SymGS-related solvers} \\
        \hline
        SymGS Only & 8879.15 (69.42) & $<10^{-3}$ & 1247.99 (9.36) & $<10^{-3}$ & 9090.66 (33.18) & $<10^{-3}$ \\
        HINTS-SymGS ($m=10/10/25$) & 27.34 (0.20) & $<10^{-3}$ & 28.67 (0.07) & 0.023 & 93.81 (1.37) & $<10^{-3}$ \\
        Learned Greedy-SymGS & \textbf{19.48 (0.36)} & -- & \textbf{27.52 (0.59)} & -- & \textbf{79.59 (1.18)} & -- \\
        \textit{True-Greedy-SymGS} & \textit{19.40 (0.41)} & -- & \textit{27.15 (0.59)} & -- & \textit{77.63 (1.12)} & -- \\
        \hline
        \multicolumn{7}{c}{SOR (1.5)-related solvers} \\
        \hline
        SOR (1.5) Only & 3658.02 (27.28) & $<10^{-3}$ & 455.89 (3.52) & $<10^{-3}$ & 3739.55 (13.66) & $<10^{-3}$ \\
        HINTS-SOR (1.5) ($m=10/10/25$) & 30.52 (0.41) & $<10^{-3}$ & \textbf{27.74 (0.48)} & 0.999 & 52.93 (0.47) & $<10^{-3}$ \\
        Learned Greedy-SOR (1.5) & \textbf{20.37 (0.41)} & -- & 29.28 (0.27) & -- & \textbf{43.41 (0.64)} & -- \\
        \textit{True-Greedy-SOR (1.5)} & \textit{19.79 (0.37)} & -- & \textit{29.10 (0.30)} & -- & \textit{43.33 (0.61)} & -- \\
        \hline
        \multicolumn{7}{c}{MG-related solvers} \\
        \hline
        MG Only & 11.14 (0.13) & 0.949 & 14.74 (0.19) & 0.987 & 314.56 (3.94) & $<10^{-3}$ \\
        HINTS-MG ($m=25/5/25$) & \textbf{11.09 (0.09)} & 0.993 & \textbf{14.73 (0.20)} & 0.990 & \textbf{125.08 (2.76)} & 0.852 \\
        Learned Greedy-MG & 11.40 (0.14) & -- & 15.13 (0.22) & -- & 126.90 (2.12) & -- \\
        \textit{True-Greedy-MG} & \textit{11.11 (0.09)} & -- & \textit{14.63 (0.17)} & -- & \textit{124.49 (2.12)} & -- \\
        \hline\hline
    \end{tabular}}
    \label{tab:wallclock_dirichlet}
\end{table}

\section{Limitations}

A primary limitation of our approach is the higher on-time computational cost requires to learn the greedy routing strategy. In contrast, the simpler baselines such as HINTS incur no training overhead. While this cost is amortized during inference time and leads to improved solution quality, it may be prohibitive in settings with limited computational resources or when rapid deployment is required. 

Additionally, the effectiveness of the learned router depends on the neural operator exhibiting heterogeneous performance across samples. When the neural operators produces predictions of relatively uniform quality, the benefit of adaptive routing diminishes and it is more practical to set a fixed schedule based on the quality of the neural operator. In practice, achieving such heterogeneity can be sensitive to training choices, and neural operators themselves can be finicky to train.

\end{appendices}

\end{document}